\documentclass[12pt]{article}
\usepackage[margin=0.7in,bottom=1.5in]{geometry}

\usepackage{amsmath}\usepackage{amsthm}
\usepackage{times}
\usepackage{bm}
\usepackage{natbib}
\usepackage{color}
\usepackage{amssymb}
\usepackage{dsfont}
\usepackage{graphicx,caption}
\usepackage{enumerate}
\usepackage{slashbox}
\usepackage{subcaption}
\usepackage{float}

\usepackage[shortlabels]{enumitem}

\usepackage[plain,noend]{algorithm2e}
\makeatletter
\renewcommand{\algocf@captiontext}[2]{#1\algocf@typo. \AlCapFnt{}#2} % text of caption
\def\@algocf@capt@plain{top}
\renewcommand{\algocf@makecaption}[2]{%
  \addtolength{\hsize}{\algomargin}%
  \sbox\@tempboxa{\algocf@captiontext{#1}{#2}}%
  \ifdim\wd\@tempboxa >\hsize%     % if caption is longer than a line
    \hskip .5\algomargin%
    \parbox[t]{\hsize}{\algocf@captiontext{#1}{#2}}% then caption is not centered
  \else%
    \global\@minipagefalse%
    \hbox to\hsize{\box\@tempboxa}% else caption is centered
  \fi%
  \addtolength{\hsize}{-\algomargin}%
}
\makeatother

\newcommand{\eps}{\epsilon}

\newcommand{\calD}{{\mathcal D}}
\newcommand{\iid}{\stackrel{\mathrm{i.i.d.}}{\sim}}
\newcommand{\sigs}{\sigma^2}

\newcommand{\given}{\,|\,}
\newcommand{\calA}{{\mathcal A}}
\newcommand{\calK}{{\mathcal K}}
\newcommand\independent{\protect\mathpalette{\protect\independenT}{\perp}}
\def\independenT#1#2{\mathrel{\rlap{$#1#2$}\mkern2mu{#1#2}}}

\newcommand{\blue}[1]{{\leavevmode\color{black}{#1}}}

\def\independenT#1#2{\mathrel{\rlap{$#1#2$}\mkern2mu{#1#2}}}
\usepackage{mathtools}

\newcommand{\pstar}{}

\def\independenT#1#2{\mathrel{\rlap{$#1#2$}\mkern2mu{#1#2}}}
\usepackage{mathtools}
\newtheorem{theorem}{Theorem}
\newtheorem{assume}{Assumption}
\newtheorem{corollary}{Corollary}
\newtheorem{lemma}{Lemma}
\newtheorem{proposition}{Proposition}

\newcommand{\bp}{{p}}

\newcommand{\bC}{{C}}

\newcommand{\bQ}{{Q}}

\newcommand{\bu}{{u}}
\newcommand{\bv}{{v}}
\newcommand{\bw}{{w}}

\newcommand{\bx}{{x}}
\newcommand{\bX}{{X}}

\newcommand{\bY}{{Y}}

\newcommand{\bZ}{{Z}}
\newcommand{\bL}{{L}}

\newcommand{\bR}{{R}}

\usepackage{color}

\newcommand{\calB}{{\mathcal B}}

\newcommand{\calF}{{\mathcal F}}

\newcommand{\calT}{{\mathcal T}}

\newcommand{\bV}{ V}

\theoremstyle{definition}

\begin{document}

%% Here are the title, author names and addresses

\title{\bf Random forests for binary geospatial data}
 \author{Arkajyoti Saha$^{1}$, Abhirup Datta$^{2}$}
 
 \date{
	$^1$ Department of Statistics, Donald Bren School of Information and Computer Sciences,
University of California, Irvine\\
	$^2$ Department of Biostatistics,  Johns Hopkins Bloomberg School of Public Health\\[2ex]
}
\maketitle

\begin{abstract}
  \textcolor{black}{The manuscript develops new method and theory for non-linear regression for binary dependent data using random forests. Existing implementations of random forests for binary data cannot explicitly account for data correlation common in geospatial and time-series settings. For continuous outcomes, recent work has extended random forests (RF) to RF-GLS that incorporate spatial covariance using the generalized least squares (GLS) loss. However, adoption of this idea for binary data is challenging due to the use of Gini impurity measure in classification trees, which has no known extension to model dependence. We show that for binary data, the GLS loss is also an extension of the Gini impurity measure, as the latter is exactly equivalent to the ordinary least squares (OLS) loss. This justifies using RF-GLS for non-parametric mean function estimation for binary dependent data. We then consider the special case of generalized mixed effects models, the traditional statistical model for binary geospatial data, which models the spatial random effects as a Gaussian process (GP). We propose a novel link-inversion technique that embeds the RF-GLS estimate of the mean-function from the first step within the generalized mixed effects model framework, enabling estimation of non-linear covariate effect and offering spatial predictions. We establish consistency of our method, RF-GP, for both mean function and covariate effect estimation. The theory holds for a general class of stationary $\beta$-mixing dependent processes that includes common choices like Gaussian processes with Matérn or compactly supported covariances, and autoregressive processes. The theory relaxes the common assumption of additive mean functions which does not hold for binary data, and accounts for the non-linear link. To our knowledge, this is the first theory of random forests for binary spatial data. We demonstrate that RF-GP outperforms competing methods for estimation and prediction in both simulated and real-world data.} 
 
\end{abstract}

\vfill

\section{Introduction}
\label{sec:new}
Geospatial applications in ecology, climatology, agriculture, forestry, and environmental health often involve data that is either naturally binary, e.g., abundance of species \citep{spnngppaper} or incidence of diseases in  agriculture \citep{zhang2002}, or is dichotomized by thresholding a continuous variable like wind-speed \citep{cao2022}, air-pollutant concentrations \citep{chen2020deepkriging} or land-cover percentage \citep{berrett2012data}. A staple for analysis of binary geospatial data is the spatial generalized linear mixed effects model (spGLMM), given as:

\begin{equation}
    \label{eqn:hglm}
    \mathbb{E}(Y_i| X_i, w) = h(X_i^\top {\beta} + w(s_i)),
\end{equation}
with binary response $Y_i$, a linear  fixed \blue{covariate} effect $X_i^\top {\beta}$ \blue{of a $D$-dimensional covariate}, a Gaussian process (GP) distributed spatial random effect $w(s_i)$ at location $s_i$, and a probit or logit link $h$ 
\citep{albert1993bayesian,diggle1998model,gelfand2000modeling,de2000bayesian,berrett2012data,banerjee2014hierarchical,cressie2015statistics}. Various estimation and prediction strategies for spGLMM have been explored \citep[see e.g.][]{zhang2002,gemperli2003fitting,cao2022,saha2022scalable}. The correlation among the observed outcomes is induced by the correlation structure on the spatial random effects. \blue{The mixed model framework is versatile. The separation of fixed and random effect via additivity allows inference on both the covariate effect (via estimate of $\beta$) and the spatial structure in the outcome (via estimates of $w(s)$ and the parameters governing its law). It also allows spatial predictions of $w(s_{new})$ at a new location $s_{new}$ via kriging, which in turn can be used to estimate predict the outcome at that location.} However, the assumption of a linear fixed effect $X_i^\top\beta$ in the spGLMM in \eqref{eqn:hglm}, although predominant in the spatial generalized mixed models, is \blue{restrictive 
    and can 
    misspecify the true data generation mechanism leading to incorrect inference. 
    
Non-linear models for binary spatial data have often used basis function expansions or splines of the covariate $X_i$ -- including non-linear mixed effects models \citep[]{diggle1989spline,smith1998additive,holmes2003generalized}, mixture models \citep{wood2008locally}, and Gaussian processes on the covariate domain \citep{schmidt2011considering}. 
Basis functions in the $D$-dimensional covariate space can suffer from curse of dimensionality even for $D$ as small as $3$-$4$ \citep{taylor2013challenging}. Generalized additive models (GAM) have also been used to model non-linear covariate effects in  spatial data \citep{wood2003thin,nandy2017additive}.  
GAMs do not model interactions between covariates which can bias estimation when strong interaction is present \citep[see, for example, illustrations in][]{zhan2024neural}.

Additionally, sometimes the primary goal in regression is  
estimating the probability or mean function $p(X) = \mathbb E \left( Y| X\right)$ of binary data, while accounting for correlation in $Y$.  
Using a model-based approach 
can be inaccurate for mean function estimation if any component of the model is misspecified, such as the link, a restricted function class for the covariate effect (linear or GAM), or the  distribution of the random effect. Here, 
a more non-parametric approach that does not require correct specification of the full data generation process is more desirable.} 

There is \blue{now widespread} adoption of machine learning methods like random forests  \citep{breiman1984classification,breiman2001random}, neural networks, \blue{and Bayesian additive regression trees \citep[BART,][]{chipman2012bart} for non-parametric, non-linear spatial analysis. However, accounting for spatial dependence arising from unmeasured spatially structured variables is tricky to do within these machine learning algorithms. In traditional statistical models, the spatial effect is either modeled using a high-dimensional quantity (Gaussian process random effects, splines, etc.) or via the covariance. It is challenging to estimate high-dimensional nuisance terms or directly incorporate covariance in most machine learning algorithms. Additionally, traditional orthogonalization techniques \citep[e.g., Robinson’s transformation,][]{robinson1988root} that adjust for nuisance variables might be difficult for spatial dependence, as dependence may not be summarized by unit-specific variables with respect to which the outcome can be orthogonalized.} 

The two main approaches \blue{to using spatial information in machine learning have been  `residual/hybrid kriging' and  `added-spatial-features'.} The former  \citep{fayad2016regional,fox2020comparing} completely ignores the correlation in $Y$ while estimating $\mathbb E \left( Y| X\right)$ \blue{followed by kriging on the residuals for spatial predictions. 
Ignoring correlation severely impacts the performance of machine learning methods \citep{saha2023random,zhan2024neural}. Also, kriging on the residuals may not be justified for binary data.  
The `added-spatial-features' method} uses spatial covariates like \blue{latitude-longitude,} pairwise distances between locations or basis functions \blue{as additional features (covariates) in the machine learning algorithm \citep{hengl2018random,wang2019nearest,chen2024deepkriging}. These methods suffer from curse-of-dimensionality from adding many  features and do not offer estimates of $\mathbb E(Y \given X)$ (only of $\mathbb E(X\given X,s)$.} 
\blue{A detailed review of binary spatial data analysis methods and their limitations is in Section \ref{sec:litrev} of the Supplementary Materials.} 

In this paper, \blue{we develop 
a non-parametric mean estimation method for binary spatial data (and more generally correlated data), using random forests that accounts for data correlation. 
We adapt the RandomForestsGLS \blue{(RF-GLS)} approach of \cite{saha2023random}.} For continuous outcomes, 
RF-GLS estimates a non-linear covariate effect $m(X)$ instead of $X_i^\top\beta$, \blue{via a novel random forests algorithm that} uses a GP to model the spatial dependence. The crux of the algorithm was representation of the regression tree creation as series of ordinary least squares (OLS) optimizations, thereby allowing extension to use generalized least squares (GLS) losses that explicitly encodes the GP spatial covariance in the tree-building algorithm. \blue{The GLS loss has been introduced in other machine learning algorithms like boosting \citep{sigrist2022gaussian} and neural networks \citep{zhan2024neural} for estimating non-linear mean of continuous spatial outcomes.}

Statistical methods for continuous outcomes in general do not translate to the binary setting. 
There are two fundamental challenges for implementing a random forest for binary spatial data that do not arise for continuous outcomes, and requires  two methodological innovations: 
\begin{enumerate}
    \item \textit{\blue{Mean function estimation using} connection between Gini measure and least squares:} 
    For binary data, regression trees and random forests are built using the {\em Gini impurity} measure, which has 
    no natural generalization 
    that accounts for correlation. We \blue{use} a simple but powerful result that for binary data, the local (within-node) Gini impurity measure is exactly equivalent to a global OLS loss. 
    This facilitates an easy extension of the Gini impurity measure for dependent binary data by switching to a global GLS loss \blue{with a spatial working covariance matrix.} 
    We can then use RF-GLS to estimate the marginal mean function $\mathbb{E}(Y|X)$   \blue{of binary spatial data while adjusting for correlation. 
    We note that this first part of our method, mean estimation using RF-GLS, is non-parametric and does not rely on specific models for binary spatial data (like the mixed effects model).} \\

    \item \textit{\blue{Covariate effect estimation and spatial prediction using} link inversion:} 
    
    \blue{In the special case of a mixed effects model for binary data, 
    the non-linear covariate effect $m(X)$ (see (\ref{eqn:hgnlm}) for definition)} is also of scientific interest, as it can help interpret the relative odds (logit link) or the underlying thresholding mechanism (probit link). It is also required for spatial prediction at new locations. For continuous data, the covariate effect \blue{$m(X)$} coincides with the marginal mean $\mathbb{E}(Y|X)$, but for binary data, these two are different, due to the non-linearity of the link function. 
 \blue{The first part of our method, i.e.,} application of RF-GLS, only yields estimate of the marginal mean function $\mathbb{E}(Y|X)$. 
     We propose a novel link-inversion to simultaneously deconvolute the \blue{covariate} effect and estimate the spatial random effect parameters from estimates of $\mathbb{E}(Y|X)$. For the probit link, we derive a closed form for the \blue{covariate} effect $m(X)$ in terms of $\mathbb{E}(Y|X)$,
    and show that one can estimate \blue{$m(X)$ and the random effect parameters using the RF-GLS estimate of $\mathbb E(Y\given X)$} by simply fitting a standard generalized linear mixed model. 
\end{enumerate}
\blue{We refer to the proposed method as} {\em RF-GP} (Random Forest with Gaussian Process \blue{covariance). The first part of RF-GP, i.e., mean estimation using RF-GLS} is a well-principled generalization of Breiman's RF classification problem, \blue{offering non-parametric mean function estimation} for binary spatial data. It subsumes Breiman's RF as a special case when \blue{using identity working covariance matrix in the GLS loss.}  \blue{Additionally, we provide estimates of partial dependence functions, summarizing the effect of a single covariate. We also show how our approach can be used for conditional average treatment effect (CATE) estimation in spatial settings by using RF-GLS estimates of the mean functions for two treatment groups. The second part of RF-GP is tailored to the widely used spatial generalized mixed effects models. Via the novel link-inversion, it offers estimates of the covariate effect and spatial random effect parameters, and spatial predictions. This part relies on the parametric specification of the link function and the GP covariance family.} We show that RF-GP outperforms RF, spatial variations of RF, and other statistical and machine learning tools for analysis of spatial binary data. 

We prove consistency \blue{of the RF-GP mean function estimate for binary dependent data}. \blue{Also, we prove consistency of partial dependence functions, and of the conditional average treatment effect in a setting where binary \blue{dependent} data is observed for two treatment groups.} 
These set of results are established under very general conditions -- \blue{primarily that the binary outcome process is stationary and absolutely regular ($\beta$-mixing). Knowledge of the link or exact covariance structure is not needed. 
 For the special case of binary data} under mixed effects framework,  
 we also establish consistency of the estimate of the covariate effect obtained using the link-inversion.

To our knowledge, these are the first consistency results for Breiman-style random forests for binary dependent data. The $\beta$-mixing class of dependent processes encompasses \blue{common models including} auto-regressive time-series, Gaussian processes with Mat\'ern \blue{and compactly supported covariances}. As a corollary, we also establish the consistency of the classical Breiman's RF for binary data. This result distinguishes from the previous theoretical work on tree and forest estimators for binary setting 
where the node partition did not depend on the data itself \citep{biau2008consistency,biau2010layered} or used additivity of the mean \citep{scornet2015consistency}.

\section{\blue{RF-GLS for correlated binary data}}\label{sec:rfgls}
\subsection{\blue{Data generation process}}\label{sec:method}

 {\color{black} We consider a dependent binary process $Y_i$  for $i =1, 2, \ldots$ with first two moments given by

\begin{equation}
    \label{eqn:genbin}
    \mathbb{E}(Y| X_i) = p(X_i);\, 
    \textrm{cov}(Y) =  \left(\textrm{cov}Y(s_i), Y(s_j) \right)
    := \Sigma. 
\end{equation}
Here $X_i$ is a $D$-dimensional covariate specifying the mean of $Y_i$, 

and $\Sigma$ is a symmetric positive definite matrix which satisfies the Frechet–Hoeffding upper bound \citep{nelsen2006introduction} to ensure that \eqref{eqn:genbin} is a valid binary process \citep{dubrule2017indicator,de2020models}. 
 The moment based approach to modeling the spatial binary outcome has been used in \cite{journel1983nonparametric,solow1986mapping,albert1995generalized,gotway1997generalized,lin2005analysis}, where the mean was either constant, or modeled as a linear effect. We relax the linearity assumption and intent to use random forest algorithm to flexibly estimate the mean function $p(X)=\mathbb{E}(Y|X)$.}

\subsection{\blue{Overview of random forests for binary data}}\label{sec:rfrev}
One can consider  
using Breiman's RF to estimate $p(X)$ \blue{using the data $\{(Y_i,X_i): 1=1,\ldots,n\}$}. However, implicit in the tree creation algorithm for RF is the assumption of independent data points. To elucidate this point, we briefly outline the classification tree algorithm for binary data. To split a node $\mathcal T$ of the current tree, the Gini impurity measure of the node is:
\begin{equation}
    \label{eqn:Gini}
    I(\mathcal{T}) = p^{(\mathcal{T})}(1 - p^{(\mathcal{T})}) + (1 - p^{(\mathcal{T})})p^{(\mathcal{T})} = 2p^{(\mathcal{T})}(1 - p^{(\mathcal{T})}),
\end{equation}
where, $p^{(\mathcal{T})}$ is the fraction of node members in $\mathcal T$ with label ``$1$". The best split is then obtained by maximizing the following split criterion over the pair $(d, c)$ where $d$ is the direction (choice of covariate) and $c$ is the corresponding cutoff value $(c:= c(d))$,

\begin{equation}
\label{eqn:split_Gini}
    \Delta_n^{CT}(d,c) =  I(\mathcal{T}) - \frac{\sum_{i = 1}^n \mathds 1 _{\{X_i \in \mathcal L\}}}{\sum_{i = 1}^n \mathds 1 _{\{X_i \in \mathcal T\}}} I(\mathcal{L}) - \frac{\sum_{i = 1}^n \mathds 1 _{\{X_i \in \mathcal R\}}}{\sum_{i = 1}^n \mathds 1 _{\{X_i \in \mathcal T\}}}I(\mathcal{R}),
\end{equation}
where $\mathcal L$ and $\mathcal R$ are the potential child (left and right) nodes creates by splitting $\mathcal T$ at $(d,c)$.

Node splitting using Gini measure \eqref{eqn:Gini} thus only makes use of the information on the response and covariates contained in the node $\mathcal T$ to be split. No information is utilized from data points in other nodes which may be 
correlated with the members of $\mathcal T$. 
Furthermore, even within $\mathcal T$, the Gini measure only depends on the proportion of 1's i.e., the mean of the node members for binary data, and does not use any covariance information. So, na\"ive application of random forest to estimate the mean function do not account for any data dependence. 
For continuous data, ignoring spatial covariance during node-splitting has been shown to be detrimental to the performance of RF \citep{saha2023random}.

\subsection{Gini impurity measure for correlated binary data}\label{sec:gini}
 We seek a direct extension of the Gini measure to account for data correlation  
 during node splitting. 
 We first make the following connection between the Gini measure used in classification trees and the least-squares-based split criterion used in regression trees. The latter is given by  
\begin{align}
\begin{split}
\label{eqn:CART_original}
    v_{n}^{RT}(d,c) = & \frac{1}{\sum_{i = 1}^n \mathds 1 _{\{X_i \in \mathcal T\}}} \left[\sum_{i = 1}^n \mathds 1 _{\{X_i \in \mathcal T\}} (Y_i - \bar{Y}^{(\mathcal T)})^2 
     - \right. \\
     &\left.\left(\sum_{i = 1}^n \mathds 1 _{\{X_i \in \mathcal L\}} (Y_i - \bar{Y}^{(\mathcal L)})^2 + \sum_{i = 1}^n \mathds 1 _{\{X_i \in \mathcal R\}} (Y_i - \bar{Y}^{(\mathcal R)})^2\right)\right],
\end{split}
\end{align}
where $\mathcal T, \mathcal L, \mathcal R$ are as defined in Section \ref{sec:method} and $\bar{Y}^{(\mathcal T)}, \bar{Y}^{(\mathcal L)}$ and $\bar{Y}^{(\mathcal R)}$ are the mean of the responses of the members of the nodes $\mathcal{T}, \mathcal{L}$ and $\mathcal{R}$ respectively. The following result provides an exact relationship between the Gini measure and  least-squares for binary data. 
\vskip 0.4em 
\begin{theorem}
\label{theorem:Classification_equals_CART} For any $(Y_i, X_i) \in \{0, 1\} \times \mathbb R^{D}; i = 1, 2, \cdots, n$, $(d,c) \in \{1,2,\ldots, D \} \times \mathbb R$, we have $\Delta_n^{CT}(d,c) = 2 v_{n}^{RT}(d,c)$.
\end{theorem}
\vskip 0.4em 
Theorem \ref{theorem:Classification_equals_CART} establishes that for binary outcome, split criteria for the classification tree and regression tree coincide up to a constant multiplier. The equivalence implies that for binary data a regression tree and a classification tree will give the same estimate of the mean function $p(X)$. \blue{As pointed out by a reviewer, this result is likely to be known implicitly, since in the \texttt{grf}
software package \citep{grf}, the propensity score, which is a probability function,  
is by default estimated with a regression
forest. However, we could not find any explicit statement or formal proof of this. Given this is a key component of our paper, we have added it as a Theorem 
and provide the proof in Supplementary materials.} 

Theorem \ref{theorem:Classification_equals_CART} 
provides the foundation for  generalizing regression trees and RF for binary data that accounts for the spatial correlation. 
We exploit the characterization of  
the local regression tree split criteria (\ref{eqn:CART_original}) as a global OLS loss optimization,
\begin{subequations}
\begin{equation}
    \label{eqn:CART}
    v_{n}^{CT}(d,c) =  \frac{1}{n} \left( \|{Y} -  Z^{(0)}{\hat{\beta}}( Z^{(0)})\|_2^2 - \| Y -  Z {\hat{\beta}}( Z)\|_2^2 \right).
\end{equation}
Here $Z^{(0)}$ and $ Z$ are the membership matrices for the leaf nodes of the tree before and after the potential node split at $(d,c)$, and ${\hat{\beta}} ( Z)$ are the node means which can be expressed as the OLS estimate of regressing $Y$ on $Z$, i.e., ${\hat{\beta}} (Z) = \left(Z ^\top  Z \right)^{-1}  Z^\top  Y$.

For correlated continuous data, \cite{saha2023random} replaced the OLS loss with the GLS loss, which is the natual analog under spatial correlation:
\begin{equation}
    \label{eqn:DART}
    \begin{aligned}
v_{n, Q}^{}(d,c) = 
&\frac{1}{n} \Bigg[\left({Y} - {Z}^{(0)}{\hat{\beta}}_{GLS}( Z^{(0)}) \right)^\top  Q\left({Y} - {Z}^{(0)}{\hat{\beta}}_{GLS}( Z^{(0)}) \right)\\ &-\left({Y} - {Z}{\hat{\beta}}_{GLS}( Z) \right)^\top  Q\left({Y} - {Z}{\hat{\beta}}_{GLS}( Z) \right) \Bigg],
\end{aligned}
\end{equation}

Here $ Q$ is a working precision matrix (inverse of a working covariance matrix capturing the spatial dependence), 

and ${\hat{\beta}}_{GLS} ( Z)$ are the corresponding GLS estimates for each node, i.e., 
\begin{equation}
    \label{eqn:DART_mean}
    {\hat{\beta}}_{GLS} ( Z) = \left( Z ^\top  Q  Z \right)^{-1}  Z ^\top  Q  Y.
\end{equation}
\end{subequations}

\subsection{\blue{Mean function estimation}}\label{sec:mean}
 \blue{For mean function estimation in correlated binary data, we can now simply use RF-GLS, leveraging the characterization of the GLS split criterion as a generalization of the Gini measure for correlated binary data.} We construct \blue{GLS-style} trees by recursive splitting of nodes \blue{using (\ref{eqn:DART})}.  
Upon building the tree, the node estimates are given by the GLS estimate (\ref{eqn:DART_mean}). \blue{The RF-GLS estimate of the  mean function $p(X)=\mathbb{E}(Y \given X)$ is the average of estimates of several such GLS-style trees. This part of our method can be accomplished by using the  open source software \texttt{RandomForestsGLS} \citep{Saha2022}.} 
 Choice of a suitable working precision matrix $Q$  
 for the GLS loss \blue{will be context-specific} and is discussed in Section \ref{sec:correlation_estimation} of the Supplement.

\subsection{\blue{Partial dependence functions}}\label{sec:pdf}
\blue{ 
In non-linear regression,  
there are various ways to quantify the effect of a single covariate, including permutation-based variable importance, Shapley values, LIME, partial dependence functions. We refer the readers to \cite{wikle2023illustration} for a recent overview of these explainability approaches for machine learning algorithms in the context of environmental applications.

Any of these explainability metrics can be used with the RF-GLS estimate of the mean function. Here we focus on estimates of partial dependence functions (PDF) from RF-GLS which comes with theoretical guarantees (see Section \ref{sec:main}). The PDF estimates the mean effect curve of the outcome with respect to one covariate by integrating out the remaining covariates. For data generated from a non-linear regression model \eqref{eqn:genbin}, the PDF of the $j^{th}$ covariate $X^{(j)}$ is defined as
\begin{equation}\label{eq:pdf}
p_j(X^{(j)}) = \int p(X) d(X^{(1)}) \ldots dX^{(j-1)} dX^{(j+1)} \ldots dX^{(D)}.
\end{equation}
As RF-GLS provides an estimate $\hat p$ of the mean function $p$ that can be evaluated at any value of the covariate $X$, the estimated pdf  for the $j^{th}$ covariate $\hat p_j(X^{(j)})$ is simply obtained by replacing $p$ with $\hat p$ in (\ref{eq:pdf}). In practice, the integral is replaced by an average over a fine grid of values.}

\subsection{\blue{Conditional average treatment effect estimation}}\label{sec:cate}

\blue{Use of machine learning methods and in particular random trees and forests has become widely popular for estimation of CATE \citep{wager2018estimation}, but many of the approaches do not consider spatial/serial correlation in the errors. We demonstrate how RF-GLS also be extended to estimation of conditional average treatment effects (CATE) between two groups with different treatments, while accounting for correlation. 

We consider data $\{(Y_i,X_i,T_i)\}$ where $T_i$ denote the binary treatment assignment, $Y_i$ and $X_i$ are as before. The model for $Y_i$ is given by $ \mathbb E(Y_i \given X_i=x, T_i=k) = p_k(x)$ for any $x$ and $k \in \{0,1\}$ and $Y_i$ is a dependent process. 
The CATE is
\begin{equation}\label{eq:cate}
    \tau(x) = \mathbb E(Y_i \given X_i = x, T_i=1) - \mathbb E(Y_i \given X_i = x, T_i=0) = p_1(x) - p_0(x).
\end{equation}

We can use a T-learner estimate of $\tau(x)$. Let $D^{[k]}=\{Y^{[k]}_i,X^{[k]}_i\}_i$ denote the subset of the data $\{Y_i,X_i\}_i$ corresponding to $T_i=k$, for $k \in \{0,1\}$. Using $D^{[k]}$ and a suitable working covariance matrix for the corresponding subset of locations, we can estimate the mean function $p_k(x)$ for the $k^{th}$ subgroup for $k \in\{0,1\}$ using RF-GLS. Letting $\hat p_k(x)$ denote this estimate, the T-learner estimate of CATE is given by $\hat\tau_T(x)=\hat p_1(x) - \hat p_0(x)$.}

\section{\blue{RF-GP}}\label{sec:rfgp}
\subsection{\blue{Spatial non-linear generalized mixed effects model}}\label{sec:mixed_model}
  \textcolor{black}{Our proposed use of RF-GLS for estimation of the mean function and estimands based on it like the PDF and CATE can be used for any binary dependent process. In this Section, we focus on the special case of binary spatial data generated from the widely used generalized mixed effects model. 
 We show how we can extend RF-GLS to a method {\em RF-GP} that accomplishes other estimation and prediction tasks common for mixed effects models with Gaussian Processes.}

 We relax the linearity assumption in \eqref{eqn:hglm} and consider \blue{data generated from} a generalized non-linear mixed effects model with non-linear covariate effect $m(X)$, i.e., 
\begin{equation}
    \label{eqn:hgnlm}
    \mathbb{E}(Y_i| X_i, w_i) = h(m(X_i) + w); w \sim GP(0, C(.,.| \theta)).
\end{equation}

\blue{Here $m(X)$ is the unknown fixed non-linear covariate effect, replacing the linear  effect $X^\top\beta$ in \eqref{eqn:hglm} and $h$ is a known link (probit/logit). The \blue{spatial} effect $w_i=w(s_i)$ is a realization at location $s_i$ of a Gaussian Process $w=\{w(s) : s \in \calD\}$ on a domain $\calD$. Abusing notation here, we use $w$ to denote the entire process or its realizations at the data locations, i.e., $w=(w(s_1),\ldots,w(s_n))^\top$ or a random variable having the same distribution as $w(s)$ which does not depend on $s$ due to stationarity (and same for $X$ or $Y$); the meaning should be clear from the context. The function $m(X)$ can be used to understand the importance of the covariates. If $h$ is the logit link, then $m(X+\delta) - m(X)$ gives the log odds ratio, i.e., $\log \left(\frac{P(Y=1 \given X=x+\delta, w)}{P(Y=0 \given X=x+\delta, w)}\; / \; \frac{P(Y=1 \given X=x, w)}{P(Y=0 \given X=x, w)}\right)$ for a shift in the covariate value by an amount $\delta \in \mathbb R^D$, keeping everything else fixed. For a probit link $h$, the observed binary outcome can be viewed as a thresholded version of the latent GP, i.e., $U_i = m(X_i) + w_i$ and $Y_i = I(U_i < 0)$. So $m$ can be interpreted as the mean of the latent GP $U$. Hence, the covariate effect $m$ can help interpret such mixed effects models. Additionally,} estimation of $m$ is needed for obtaining spatial predictions of the binary outcome at new locations.

\blue{The mixed model (\ref{eqn:hgnlm}) has a functional component (the covariate effect) $m(X)$ and a spatial component $w(s)$. Historical connections between functional and spatial modeling have been long-established \citep[see, e.g., works of ][ and others on equivalence between splines and Gaussian processes]{kimeldorf1971some}. Hence, it is very natural to consider modeling $m$ and/or $w$ using basis functions or GP. However, as discussed before a basis function model for $m$ may suffer from curse of dimensionality unless $D$ is very small. Approaches like GAMs reduce dimensionality but also restrict the function class. It is also not straightforward to estimate $m$ using machine learning approaches within the model based framework of (\ref{eqn:hgnlm}) due to the presence of the high-dimensional parameter $w$ (see Section \ref{sec:litrev} for more details). 

Next we present a novel approach to estimate $m$ in the spatial mixed effects model (\ref{eqn:hgnlm}) using the RF-GLS estimate of the mean function $p(X) = \mathbb E(Y\given X)$ from Section \ref{sec:mean}, that bypasses estimation of $w$ and at the same time estimates $m$ without restricting the function class.}

\subsection{Link inversion for \blue{covariate} effect estimation in mixed models}
\label{sec:link_inversion}

\blue{Estimating the covariate effect $m(X)$ using random forests within the non-linear generalized mixed effects model for binary data is not straightforward.} For continuous (Gaussian) data with \blue{linear} (identity) link, the covariate effect $m(X)$ is same as the mean $p(X)$ after integrating out $w$.  
This convenience is lost \blue{for binary data within the generalized mixed model framework  
as $m(X)$ is not same as $p(X)$ due to the non-linear link $h$.} Trying to directly use the hierarchical model (\ref{eqn:hgnlm}) in random forests estimation is challenging 
because of the high-dimensional \blue{spatial} effect parameter $w=(w(s_1),\ldots,w(s_n))^\top$. Also, the marginal likelihood, integrating out $w$, generally cannot be obtained in closed form because of the non-linear link. 

We propose a novel link-inversion approach to estimate the covariate effect $m(X)$ using the \blue{RF-GLS} estimate of the mean $p(X)$ from 
\blue{Section \ref{sec:mean}.}  
Integrating out $w$ from (\ref{eqn:hgnlm}), we have 

\begin{equation}\label{eq:marg}
    p(X_i) := \mathbb{E}(Y_i | X_i) =  \int h(m(X_i) + w(s_i))dw(s_i). 
\end{equation}

\blue{For any $x$, letting $m(x)=m$ and $p(x)=p$, we have 
from (\ref{eq:marg}), 
$p=\int h(m + w(s_i))dw(s_i)=g(m)$ where $g$ will be a smooth function of $m$ contingent on  
the link $h$ being smooth and $w$ being stationary.  
Then $m$ can be recovered from the estimated $p$ if $g$ is invertible. 
The following result presents conditions under which this happens.} 
\vskip 0.4em 
\begin{proposition}
\label{lemma:link_inverse_existance}
Let $h: \mathbb{R} \to (0,1)$ be a strictly increasing link function in (\ref{eqn:hgnlm}) and $w$ is a stationary process such that  
$w(s)$ has distribution $\mathbb F_w$ for all $s$. Let $\mathbb{E} (Y | X) = g_{h,\mathbb F_w}\left(m(X)\right)$.
\begin{enumerate}[(a)]
    \item Then there exists an inverse function $g_{h,\mathbb F_w}^{+}(\cdot)$ such that 
    $m(x) = g_{h,\mathbb F_w}^{+}(\mathbb{E} (Y | X = x))$.
    \item When $h$ is the probit link and $w$ is a stationary GP with variance $\sigma^2$, 
    \begin{equation}
    \label{eqn:probit_link_inverse}
    m(x)  = g_{h,\mathbb F_w}^{+}(p(x))  
    = \left( 1 + \sigma^2\right)^{\frac{1}{2}} \Phi^{-1} \left( p(x) \right),
\end{equation}
where $\Phi$ is the cumulative distribution function (cdf) of standard normal distribution.
\end{enumerate}
\end{proposition}

Part (a) guarantees  
recovery of \blue{the function} $m$ 
under very general conditions (increasing link and stationary random effect process).   
In part (b) we show that for probit link, we can obtain a closed-form estimate of $m(x)$ in terms of the mean function $p(x)$. The probit link is widely used in spGLMM,  
due to its latent variable representation  \citep{gelfand2000modeling,albert1993bayesian,de2000bayesian,berrett2012data,cao2022,saha2022scalable}.
  For a probit link, the covariate effect simply becomes a linear model of a known transformation $(\Phi^{-1})$ of the mean. Hence, subsequent to the estimation of the mean function,  
we can use off-the-shelf software for spatial generalized linear mixed models to estimate the covariate effect $m$ as follows.

Let $R$ be the spatial correlation function corresponding to $C$, \blue{parameterized by some $\alpha$, i.e.,} $C(\cdot, \cdot \given \theta)=\sigma^2 R(\cdot, \cdot \given \alpha)$, and let $\theta=(\sigma^2,\alpha)$.    
From Proposition \ref{lemma:link_inverse_existance} (b), we can write 
    \begin{equation}
    \label{eqn:probit_GLMM_prediction}
    \mathbb{E} (Y_i | X_i,w(s_i)) = h(\Phi^{-1}(p(X_i))\beta + w(s_i)); w \sim GP(0, \sigma^2 R(.,.| \alpha)).
    \end{equation}
     
This 
is now a spGLMM  \eqref{eqn:probit_GLMM_prediction} with a linear regression on $\Phi^{-1}(p(X_i))$  
and with $\beta = (1 + \sigs)^{\frac{1}{2}}$.  
Plugging in the  estimate $\hat p(x)$ from Section \ref{sec:mean}, the unknown parameters $\sigma^2$ and $\alpha$ can be  estimated using any standard method for spGLMM (likelihood optimization, MCMC, or cross-validation). Subsequent to estimating the covariate effect $m$  
is simply recovered using (\ref{eqn:probit_link_inverse}). 

\blue{The estimate of $m$ is thus proportional to $\Phi^{-1}(\hat p(X_i)))$ where $\hat p(X_i)$ is the RF-GLS estimate from Section \ref{sec:mean}. As the class of random forest estimators for $p$ are universal approximators for any function in $(0,1)$ and inverse the Gaussian cdf $\Phi^{-1}: (0,1) \to \mathbb R$ is a monotone fucntion, the $\Phi^{-1}(\hat p(X_i))$ is an estimate of $m$ from a class of universal approximators on $\mathbb R$, this facilitates accurate identification of $m$ as proved in Theorem \ref{th:main_m} and seen empirically.

We note one that 
if $X$ is a function of only spatial coordinates $s$ and does not have any non-spatial component, then interpretion of $m$ can be challenging. Then $m$ and $w$ can become unidentifiable as the spatial effect can be modeled via $m$ or $w$ or $m+w$. However, this will be the case even in linear mixed effects models, i.e., when $m(X)=X^\top\beta$. The $\beta$'s will lose interpretation when $X$ is only a function of space. This issue of interpretation is applicable to all mixed effects models for spatial data. However, in this case $m+w$ can still be interpreted as the joint covariate-spatial effect, and our method should still work well for predictions.
}

\subsection{Prediction}\label{sec:prediction}
The link inversion technique proposed above seamlessly harmonizes with GP-based predictions in spGLMM. Once we estimate $m$ and all the spatial parameters $\theta$ using the link inversion of Section \ref{sec:link_inversion}, the task of spatial predictions at a new location is simply recast as an spGLMM prediction problem. For probit links, with GP distributed spatial effects, predictions can be obtained in closed form.  
We denote by $m=(m(X_1),\ldots,m(X_n))^\top$,  
$C=Cov(w \given \theta)$ the estimated covariance matrix of the $w$ at the data locations, and $D$ a diagonal matrix with entries $2Y(s_1) -1, \ldots, 2Y(s_n)-1$.  
For predicting at $Y_{new} = Y(s_{new})$ at a new location $s_{new}$ with covariates $X_{new}$, let $Y^*$, $m^*$, $C^*$ and $D^*$ respectively denote the analogs of $Y$, $m$, $C$ and $D$ when including a new datapoint $(Y_{new}=1,X_{new},s_{new})$ to the dataset. Then following \cite{cao2022}, the predicted conditional mean probability of $Y_{new}=1$ is given by 
\begin{equation}
    \label{eqn:GLMM_prediction}
    \mathbb{E}(Y_{new} \given X_{new}, s_{new} Y,X) = 
    \frac{\Phi_{n+1}(D^*m^*,I_{n+1} + D^*C^*D^*)}{\Phi_{n}(Dm,I_{n+1} + DCD)}
\end{equation}
where $\Phi_n(u,V)$ is the cdf of an $n$-dimensional normal distribution with mean $u$ and variance $V$. \cite{saha2022scalable} developed a scalable approach to compute such multivariate normal cdf's for spatial data using the Nearest Neighbor GP covariance matrices. 
Adopting this method, the ratio in (\ref{eqn:GLMM_prediction}) can be calculated efficiently with the estimates of $m$ and $\theta$ obtained from Section \ref{sec:link_inversion}.

This completes all the steps of our method. \blue{The first part in Section \ref{sec:rfgls} offers non-parametric mean function estimation for binary dependent processes using RF-GLS  
without requiring knowledge of the full data distribution or the exact form of spatial dependence. The second part embeds this mean function estimate} in the generalized mixed model framework with Gaussian Process based random effects -- offering covariate effect estimation, and spatial prediction for binary spatial data \blue{via the link-inversion. This relies on a parametric link function (probit) and a choice of the GP covariance family.} 
We refer to this model as {\em RF-GP} as it combines Random Forest with Gaussian Process within the generalized mixed model framework.

\section{Theory}\label{sec:consistency}
\subsection{\blue{Brief literature review of RF theory}}\label{sec:threv}
We establish the main results on the asymptotic consistency of the RF-GP for dependent binary data.   
To our knowledge, there is no current theory for RF or variants like spatial RF \citep{hengl2018random} for spatially dependent binary data. 
 \blue{There is a large literature on theoretical properties of `honest' random trees and forests for various data types \citep{wager2018estimation,athey2019generalized,denil2014narrowing,wu2022uncertainty}. For more details please refer to \cite{havelka2022honesty} and references there in}.   
\blue{However, honest RF  
uses independent splits of data for tree partitioning and node representative assignment which makes its theoretical study  
quite different from that of Breiman's RF or RF-GLS which use the whole data for both node splitting and representative assignment.} For continuous data, consistency of Breiman's regression tree and the random forest was established for independent errors in \cite{scornet2015consistency} which was extended to dependent settings in \cite{saha2023random} and \cite{goehry2020random}. 
\cite{saha2023random} also proved the consistency of RF-GLS for dependent continuous data. 

There are fundamental differences in the data generation procedure for continuous and binary spatial responses, creating new \blue{new challenges for the theoretical study of RF for binary data} that does not have an analog in the continuous case theory of \cite{scornet2015consistency} and \cite{saha2023random}. \blue{In mixed effects models for} continuous data, $Y_i = m(X_i) + w(s_i) + \eps_i$ where $\eps_i$ are the random i.i.d. errors. This additive form allows studying the RF-GLS estimator (\ref{eqn:DART_mean}) by separately studying convergences of empirical averages of $X$, $w$, and $\eps$ processes. This separability does not arise for the binary data which are a \blue{sequence of correlated} random coin tosses with outcome probability, that is non-linear in $X$. Another aspect of a non-linear mean is the challenge of controlling the variance of leaf nodes of the decision trees. For continuous data, this is accomplished by assuming an additive structure on the mean. For binary data, the mean $p(X)$ is no longer additive in the components even if $m$ is additive \blue{in (\ref{eqn:hgnlm}).} Also, unlike continuous data, the mean function $p(X)=\mathbb{E}(Y|X)$ does not equal the covariate effect $m(X)$. 
Hence, in addition to studying the consistency of the estimates of the mean function, properties of the estimate of the covariate effect obtained via the link inversion of Section \ref{sec:link_inversion} also need to be \blue{separately} studied.

\subsection{\blue{General} results}\label{sec:main}
\blue{We first present a set of general results for consistency of RF-GLS for estimation of the mean and quantities derived from it like PDF and CATE for dependent binary data. These results do not rely on parametric data generation assumptions. Specific examples are discussed in Sections \ref{sec:mm} through \ref{sec:rftheory}.}
 We outline the assumptions required for the general consistency results.  
\vskip 0.4em 
\begin{assume}[ \textcolor{black}{Data generating process}]\label{as:data_gen} \blue{$Y_i \in \left\{ 0, 1\right\}$ is a binary process such that $p(x) = \mathbb E \left(Y_i | X_i =x \right) = g(m(x))$ for some 
 increasing continuously differentiable function $g: \mathbb R \to \left[0,1 \right]$ and some continuous and additive function $m(X)=\sum_{d=1}^D m_d(X^{(d)})$ where $X^{(d)}$ denotes the $d^{th}$ component of $X$. $X^{(d)} \overset{i.i.d}{\sim} \mathrm{Unif}[0,1]$ and $\{(Y_i,X_{i-l})\}_i$ is a stationary ergodic {\em absolutely regular} ($\beta$-mixing) process for any fixed finite lag $l \in \mathbb Z$, the set of integers.} 
\end{assume}

\begin{assume}[Regularity of working precision matrix]\label{as:working_cov} 
The working precision matrix $Q$ is diagonally dominant with a banded Cholesky factor. 
\end{assume}

\begin{assume}[Rate of tree growth]\label{as:tn_rate} 
Let $t_n$ be the maximum number of leaves in a tree. Then, $
t_n (\log n)/n \to 0 \text{ as } n \to \infty$.
\end{assume}

\blue{Assumption \ref{as:data_gen} lays out the data generation process. The theory assumes that there is true value of the parameter of interest (in this case the function-valued parameter $p$) which generates the data. We do not make any assumption about the full distribution of the outcome process or even about the covariance structure. We only assume that the mean function is a composition of an additive function $m$ and a monotone link $g$. 
We do not assume knowledge of $g$ or $m$.} 
The additivity of $m$ is commonly used in the theoretical study of random forests \cite{scornet2015consistency,saha2023random} and is critical to control the variation of the mean function $p$ in the leaf nodes. However, for continuous data, $p$ and $m$ are the same, implying $p$ is also additive. 
For binary data, due to the nonlinear $g$, even for an additive $m$, the mean $p(X)$ is not additive or separable in the individual \blue{components of $X$.}  
We will develop a novel proof for controlling the variation of $p$ in the leaf nodes using the \blue{Taylor expansion} and the regularity of the link. \blue{The assumption of the covariates being Uniform$[0,1]$ variables is also standard as random forests or RF-GLS estimators are invariant to any monotone transformation of the covariates. Finally, the nature of dependence in the outcome process is assumed to be stationary and $\beta$-mixing \citep[see][for a definition]{bradley2005basic} as it allows using existing uniform laws of large numbers \citep{nobel1993note} to control the estimation error. 
 We assume $\{(Y_i, X_{i-l})\}_i$ to be jointly stationary ergodic and $\beta$-mixing for any fixed lag $l$. This is needed for the general results in this section as the GLS estimator involves terms of $Y$ and lagged $X$ and we make no further assumptions on the exact data generation mechanism for $Y$.  
 For the results in the next section focused on the special case of the generalized mixed-model (\ref{eqn:hgnlm}), it suffices to just assume $w$ to be stationary, ergodic and $\beta$-mixing and $w \independent X$, same as what was assumed in \cite{saha2023random} for continuous data.}

The regular structure of the working precision matrix (Assumption \ref{as:working_cov})
 exists in many common dependent processes in spatial and time-series literature like GP with exponential covariance functions or Nearest Neighbor Gaussian Processes \cite[NNGP,][]{datta2016nearest} on the 1-dimensional lattice, and autoregressive time-series. 
 The restriction is only on the working precision matrix, not on the true precision matrix of the spatial random effects generating the data. We specify the technical regularity conditions of Assumptions 1 and 2 in Supplement Section \ref{sec:assumptions}.  

Assumption 3 is about the scaling of the trees relative to the sample size. The assumed rate for binary data is in line with the rate mentioned in \cite{scornet2015consistency} for bounded error and is a more lenient bound on the number of leaf nodes $t_n$ compared to the $t_n (\log n)^9/n \to 0$ scaling used in \cite{scornet2015consistency} and \cite{saha2023random} for unbounded continuous outcomes. 

We first present general results under these assumptions. 
We then discuss specific examples of dependent data generation processes and working covariances that are covered by the  results. 
Let $\mathcal{D}_n$ be the data and 
$p_n ({x}_0; \Theta, \mathcal{D}_n)$  
be the RF-GLS tree estimate of the mean function, built with a random i.i.d. $\Theta$. The RF-GLS forest estimator is then given by $\bar{p}_n  ({x}_0; 
\mathcal{D}_n) = \mathbb{E}_{\Theta}p_n ({x}_0; \Theta,
\mathcal{D}_n)$, i.e., an average of trees generated from all possible instances of randomness $\Theta$. Our first result shows that \blue{the mean function estimate $\bar{p}_n(x) $ from RF-GLS in Section \ref{sec:mean}} is an $\mathbb{L}_2$-consistent estimator of $p(x )= \mathbb{E} (Y | X = x)$. 
\vskip 0.4em 
\begin{theorem}[Consistency for mean function]\label{th:main_gh}
	Under Assumptions \ref{as:data_gen} - \ref{as:tn_rate}, \blue{the RF-GLS estimate $\bar p_n(x)$ is} $\mathbb{L}_2$-consistent for the mean function \blue{$p(x)$ of binary data}, i.e., 
	$
	\lim_{n \to \infty} \mathbb{E} \int \left(\bar{p}_n(X) - p(X) \right)^2 \, dX = 0$.
\end{theorem}
\vskip 0.4em 
\blue{The consistency of the mean directly ensures the consistency of numerous estimators derived from it. Our next result shows that the estimates of partial dependence function $\hat p_j(X^{(j)})$ from RF-GLS (Section \ref{sec:pdf}) is a consistent estimate of the true PDF $p_j(X^{(j)})$ for each variable $j$.
\vskip 0.4em 
\begin{corollary}[Consistency of PDF]\label{cor:pdf} For any $1\leq j \leq D$, the estimated partial dependence function $\hat p_j(X^{(j)})$ is an $\mathbb{L}_2$ consistent estimate of $p_j(X^{(j)})$, i.e., 
$ \mathbb E \int (p_j(X^{(j)}) -  \hat p_j(X^{(j)}))^2 dX^{(j)} \to 0$.   
\end{corollary}
\vskip 0.4em 
The result shows that partial dependence functions estimated from RF-GP can be used to understand the true relationship between the outcome and any covariate by plotting the partial dependence function $\hat p_j(X^{(j)})$ as a function of $X^{(j)}$. 

In Section \ref{sec:cate}, we have described how the RF-GLS mean estimate for 2 different groups can be used to estimate the conditional average treatment effect. 
Theorem \ref{th:main_gh} also implies consistency 
of this T-learner based estimate of CATE from RF-GLS as stated in the following result.
\vskip 0.4em 
\begin{corollary}[Consistency of CATE]\label{cor:cate}
If each subsequence $\{Y^{[k]}_i,X^{[k]}_i\}_i$ satisfies Assumptions 1-3 for $k \in \{0,1\}$, then $\mathbb E \int (\tau(X) - \hat \tau_T(X))^2 dX \to 0.$
\end{corollary}
\vskip 0.4em

We note that the consistency of estimators derived from the RF-GLS mean estimate, like PDF and CATE, have not been studied previously even in the context of continuous data. Our results on consistency of these estimates also apply for continuous data, as they are solely reliant on the mean estimate being consistent. 

All results in this Section are  
established for the data generation process outlined in Assumption \ref{as:data_gen}, only requiring  specification of the mean function and stationary $\beta$-mixing dependence. We did not make any parametric assumptions, like knowledge of link function or covariance structure. In the next sections we present theory 
for specific examples. }

{\color{black}\subsection{Example: Generalized mixed effects models}\label{sec:mm}
In this Section, we focus on the generalized non-linear mixed effects model (\ref{eqn:hgnlm}). As we demonstrated in Section \ref{sec:link_inversion}, this is a special case of (\ref{eqn:genbin}) 
with $p(X)=g(m(X))$ where $m$ is the covariate effect and $g$  
is given by $g_{h,\mathbb F_w}$ in Proposition \ref{lemma:link_inverse_existance}. The  difference between the specifications in \eqref{eqn:genbin} and \eqref{eqn:hgnlm} 
 is that \eqref{eqn:genbin} only specifies the moments without reauiring any further distributional assumptions, while \eqref{eqn:hgnlm} specifies the whole distribution of the $Y$-process. 

We show in Lemma \ref{lemma:joint_mixing} of the Supplementary Materials that for the data generating process in  \eqref{eqn:hgnlm}, Assumption \ref{as:data_gen} on the mixing strength of $Y$ is be satisfied if simply the random effects $w$ is a stationary, ergodic  absolutely regular ($\beta$-mixing) spatial process, independent of $X$. %with a finite first moment . 
This immediately implies consistency of the RF-GLS mean function estimate for data generated from (\ref{eqn:hgnlm}).  
See Corollary \ref{th:main_gh_mixed_corollary} in  Section \ref{sec:outline} of Supplementary Materials for the full statement.

As discussed in Section \ref{sec:link_inversion}, within the generalized mixed effects model framework, estimation of the covariate effect $m$ is of importance.} 
Our next result shows that  
the link inversion approach of Section \ref{sec:link_inversion} yields an $\mathbb L_2$ consistent estimator of the covariate effect $m$ from the estimate of the mean function $p$. Let $Im(m)$ be the image of $m$, and for any set $A$ and function $f$, denote the restriction of $f$ to the set $A$ as $f|_A$. \blue{Let $ f\circ g$ denote the composition of two functions $f$ and $g$.}
\vskip 0.4em
\begin{theorem}[Consistent recovery of covariate effect]\label{th:main_m}  
\blue{For data generated from (\ref{eqn:hgnlm}),} under Assumptions \ref{as:data_gen} - \ref{as:tn_rate}, 
there exists a function $\kappa : \mathbb R \to \mathbb R$ with $\kappa |_{g_{h,\mathbb F_w}(Im(m))} = g^+_{h,\mathbb F_w}$ where $g^+_{h,\mathbb F_w}$ is the link-inversion function of Proposition \ref{lemma:link_inverse_existance}, such that 
$\kappa \circ \bar{p}_n$ is an $\mathbb L_2$-consistent estimator of $m$, i.e.
    	$\lim_{n \to \infty} \mathbb{E} \int \left(\kappa(\bar{p}_n(X)) - m(X) \right)^2 \, dX = 0$.
\end{theorem}
\vskip 0.4em

Theorem \ref{th:main_gh} \blue{(Corollary \ref{th:main_gh_mixed_corollary})} and Theorem \ref{th:main_m} show that \blue{for generalized mixed models (\ref{eqn:hgnlm})}, RF-GP is consistent both for the mean function and the covariate effect. In the next sections, we \blue{present common spatial and time-series models for which these consistency results hold}. 

\subsection{Example: Mat\'ern Gaussian process}\label{sec:gp}

\blue{We consider the generalized mixed effects model (\ref{eqn:hgnlm}) with the covariance of the Gaussian Process $w$ being from the} Mat\'ern family of covariance functions. \blue{The Matérn class is} widely popular as it can characterize the smoothness of the spatial surfaces \citep{stein2012interpolation} and subsumes the popular exponential and the Gaussian (squared-exponential) covariances as special or limiting cases. The Mat\'ern covariance between locations $s_i$ and $s_j$ is given by 
\begin{equation}\label{eq:matern}
C(s_i,s_j \given {\theta}) = C(\|s_i - s_j\|_2) = \sigs \frac{2^{1-\nu} \left(\sqrt 2 \phi\|s_i - s_j\|_2\right)^\nu}{\Gamma(\nu)} \calK_\nu \left(\sqrt 2 \phi\|s_i - s_j\|_2\right), 
\end{equation}
where ${\theta}=(\sigs,\phi,\nu)^\top$ is the set of spatial parameters, specifying the covariance function and $\calK_\nu$ is the modified Bessel function of the second kind. 

For binary spatial data with random effects $w \sim GP(0,C(\cdot,\cdot))$, the spatial structure of the working covariance matrix $Q^{-1}$ ideally needs to be chosen from the same family as $C$ (see Section \ref{sec:correlation_estimation}). However, for spatial data measured at $n$ locations, the Mat\'ern covariance yields an $n \times n$ dense matrix, and inverting it to obtain $Q$ would involve $O(n^2)$ storage and $O(n^3)$ parameters.  
To circumvent the computational challenges, we recommend choosing $Q$ to be the Nearest Neighbor Gaussian Process \citep[NNGP,][]{datta2016nearest,finley2019efficient} precision matrix based on the Mat\'ern covariance family. NNGP precision matrices provide an excellent approximation to their full GP analogs while 
requiring $O(n)$ storage and time. 

The following result proves the consistency of RF-GP, for binary spatial data generated from a mixed effects model with Mat\'ern GP random effects and when using NNGP working precision matrix in the RF-GP algorithm. 
\vskip 0.5em 
\begin{corollary}
\label{prop:spatias_application}
Consider binary data generated from a generalized mixed effects model $Y(s_i) \given X_i,w(s_i) \overset{ind}{\sim}  \mathrm{Bernoulli}\left(h\left(m(X_i) + w(s_i\right)\right)$ where $h$ is a probit or logit link, $m$ is continuous and additive, $w(s)$ is a Mat\'ern GP (independent of $X$), sampled on the one-dimensional regular lattice, with spatial parameters given by $(\sigs_0,\phi_0,\nu_0)$, $\nu_0$ being a half-integer. Let $Q$ denote a working precision matrix obtained from a Nearest Neighbor Gaussian Process (NNGP) under Mat\'ern covariance, with parameters $(\sigs,\phi,\nu)$. Then there exists some $K >0$ such if $\phi > K$, then RF-GP using $Q$ yields an $\mathbb L_2$ consistent estimate of $\mathbb{E}(Y | X)$ and $m(X)$.
\end{corollary}
\vskip 0.5em 
The lattice design is widely accepted in theoretical investigations of spatial processes \citep{du2009fixed,stein2002screening}. The restriction on $\nu$ to half-integers $\nu \in {1.5,2.5,\ldots}$ has also been of special interest as these facilitate efficient likelihood computation \blue{due to close form covariance expressions.} 
These assumptions simultaneously ensure that the dependence in the true data generation process is $\beta$-mixing, 
and that the NNGP working precision matrix $Q$ has a banded Cholesky factor. Finally, choosing $\phi > K$ for the spatial range parameter ensures that $Q$ is diagonally dominant. We reiterate that this constraint is only for the working covariance matrix and does not restrict the true generation process.

{\color{black}\subsection{Example: Compactly supported covariance functions on two-dimensional domains}\label{sec:2d}

In many applications, spatial data is from two-dimensional domains. The main challenge prohibiting extension of Corollary \ref{prop:spatias_application} to two-dimensional domains is the $\beta$-mixing condition on the \blue{spatial} effect process.  
In this Section, we consider the wide class of compactly covariance functions for which consistency of RF-GP can be established in two-dimensional strips. A stationary covariance function $C_a$ is compactly supported if $C_a(d)=0$ for all distances $|d|$ greater than some threshold $a$. Compactly supported covariance functions are widely popular  in spatial analysis. Examples include the spherical covariance function given by
\begin{equation}\label{eq:sphere}
C_a(d) =
\begin{cases}
\sigma^2 \left( 1 - \frac{3|d|}{2a} + \frac{|d|^3}{2a^3} \right), & 0 \leq |d| \leq a, \\
0, & |d| > a.
\end{cases}
\end{equation}
This covariance function has historically been prevalent in geological and hydrological applications 
\citep{stein2012interpolation}. Another compactly supported covariance family is the generalized Wendland function class \citep{gneiting2002compactly} given in (\ref{eq:generalized_wendland}). 
Recently, \cite{bevilacqua2022unifying} showed that the  
generalized Wendland functions approximates the Matérn class which is attained as a special limit case. 
More generally, compactly supported covariance functions have become popular  to reduce computational costs in spatial analysis via `covariance tapering' \citep{furrer2006covariance,kaufman2008covariance}. Given any valid covariance function $C$ and a compactly-supported covariance function $C_a$, the tapered covariance is given by the Schur-product $C_{taper,a}=C * C_a$. Theorem 5.2.1 of \cite{horn1994topics} ensures that $C_{taper}$ is a valid covariance function, enabling its use in spatial models, replacing $C$. As $C_{taper,a}$ is zero outside the threshold distance, using it as the GP covariance in spatial models yields sparse covariance matrices alleviating storage costs. Also, $C_{taper,a}(d) \to C(d)$ as $a \to \infty$. So any covariance function can be approximated using 
compactly supported covariance functions. 

\begin{equation}\label{eq:generalized_wendland}
\begin{aligned}
    &C_a(d) = \sigs \phi_{\nu, \mu, a} \left(|d|\right) ;\\
&\phi_{\nu, \mu, a} \left(|d|\right) := \mathcal{GW}_{\nu, \mu, \delta_{\nu, \mu, a}}\left(|d|\right); \:\:\:\:\:\delta_{\nu, \mu, a} = a \left( \frac{\Gamma \left( \mu + 2 \nu + 1\right)}{\Gamma \left( \mu\right)}\right)^{\frac{1}{1 +2\nu}}; \\
&\mathcal{GW}_{\nu, \mu, a}\left(r\right) := \begin{cases}
    \frac{1}{B\left(2\nu, \mu + 1 \right)} \int_{r/a}^1 u \left( u^2 - \left( r/a\right)^2 \right)^{\nu - 1} \left(1 - u \right)^\mu du, & \text{if $0 \leq r \leq a$}\\
0, & \text{if $r > a$}
\end{cases}\\
&\mathcal{GW}_{0, \mu, a}\left(r\right) := \begin{cases}
    \left( 1 - \frac{r}{a}\right)^\mu, & \text{if $0 \leq r \leq a$}\\
0, & \text{if $r > a$}.
\end{cases}\\
\end{aligned}
\end{equation}

The following result proves the consistency of RF-GP on two-dimensional strips for binary spatial data generated from a mixed model in \eqref{eqn:hgnlm} with the \blue{spatial} effects following a GP with any stationary compactly-supported covariance function. 
\vspace{0.4em}
\begin{proposition}
\label{prop:spatias_application_2D}
Consider binary data generated on a two-dimensional lattice $\{(k,l) : 1 \leq k \leq n_1, 1 \leq l \leq n_2, k,l \in \mathbb Z\}$ 
from a generalized mixed effects model (\ref{eqn:hgnlm}) with continuous and additive link $m$, probit or logit link $h$, and $w \sim GP(0, C(\theta_0))$. where $C(\theta)$ is any class of stationary compactly supported covariance functions.  
As either $n_1 \to \infty$ or $n_2 \to \infty$, RF-GP using a working precision matrix $Q$ from a Nearest Neighbor Gaussian Process (NNGP) based on $C(\theta)$ for suitable choices of $\theta$ yield $\mathbb L_2$ consistent estimates of $p(X)=\mathbb{E}(Y | X)$ and $m(X)$.
\end{proposition}
\vspace{0.4em}
A technical version of the result is stated as Proposition \ref{prop:compacttech} in the Supplement where it is proved. 
Compactly-supported of covariance functions vanish outside a ball with given radii embedded in $\mathbb R^2$. This ensures that the GP $w$ is an $m$-dependent process on a two-dimensional strip (lattice growing in one direction). This in turn implies that it is a $\beta$-mixing process \citep{bradley2005basic}. 

We state the result as a proposition rather than a corollary of Theorems \ref{th:main_gh} and \ref{th:main_m} as due to the two-dimensional design,  
the NNGP working covariance matrix $Q$ does not satisfy the banded condition of Assumption \ref{as:working_cov}.  
The proof needs to account for this more complex banding structure. Proposition \ref{prop:spatias_application_2D} provides consistency guarantees of RF-GP estimates of both the mean function and the covariate effect on such two-dimensional lattices for the set of all compactly supported function, which as we discussed above, approximates any covariance function.}

\subsection{Example: Binary time-series}\label{sec:ar}
In this article, we primarily focus on binary spatial data generated from and modeled using Gaussian processes. However, the Assumptions \ref{as:data_gen} - \ref{as:tn_rate} are general enough to study the consistency of RF-GP for a  broader class of binary dependent processes. We can consider binary time-series data, with the common autoregressive covariance structure for the temporal random effect and show that RF-GP with an autoregressive working covariance structure yields a consistent estimator of both the mean function and the covariate effect. The details of the data generation model and consistency result are provided in Section \ref{sec:supar} of the Supplementary Material.

\subsection{Example: Na\"ive Random forests for binary data}\label{sec:rftheory} 
The GLS loss (\ref{eqn:DART}) used for our proposed estimator reduces to the Gini measure by setting the working precision matrix $Q = I$, which also satisfies Assumption \ref{as:working_cov}. 
\blue{Thus RF-GLS subsumes Breiman's RF as a special case,}  immediately establishing the consistency of Breiman's RF estimator for binary data as a special case of RF-GP. 
\vskip 0.5em
\begin{corollary}
\label{cor:rf-binary} Under Assumptions \ref{as:data_gen} and \ref{as:tn_rate}, the na\"ive RF \citep{breiman2001random} produces $\mathbb L_2$ consistent estimators of $\mathbb{E} (Y | X)$ for binary dependent data. 
\end{corollary}
\vskip 0.5em

Consistency of original RF for binary data even under dependence is not surprising and the result aligns with the corresponding result for continuous data in \cite{saha2023random}. The Gini measure used in RF is equivalent to OLS loss which is known to produce consistent estimates even under dependence but is generally less efficient than a GLS loss that captures the spatial dependence. This is reflected in our empirical results where we will see that RF-GP performs considerably better than the original RF due to accounting for the spatial correlation. 

A further sub-case of interest is when $w \equiv 0$, i.e., the responses are independent. 
\vskip 0.5em
\begin{corollary}
\label{cor:rf-binaryind} For i.i.d. binary data, RF \citep{breiman2001random}, using Gini measure for node splitting, is $ \mathbb L_2$ consistent for the mean $\mathbb{E}(Y\given X)=h(m(X))$ 
under Assumptions \ref{as:data_gen} and  \ref{as:tn_rate}. 
\end{corollary}
\vskip 0.5em
We think this is the first result on the consistency of Breiman's RF using the Gini impurity measure for binary data. Prior results on consistency of RF classification focused on ``honest" or non-adaptive trees, which assumed the node splits to be independent of $Y$  
\citep{lin2006random,biau2008consistency,biau2010layered,biau2012analysis} \blue{or assumed that the mean function is additive \citep{scornet2015consistency} which does not hold for binary data.} Our result accounts for \blue{data-dependent node splitting as in Breiman's RF and a non-additive mean function.}

\blue{\section{Illustrations}\label{sec:sim}
\subsection{Simulations}
We demonstrate the advantages of the RF-GP over other state-of-the-art approaches through simulation experiments here and on an application to soil type prediction in the next section. 

We compare the performance of our approach with $10$ state-of-the-art methods from both model-based and machine learning paradigms on three tasks : a) mean estimation of $p(X) = \mathbb E \left(Y | X\right )$, b) estimation of covariate effect $m \left(X \right)$ in the mixed model framework \eqref{eqn:hgnlm}, and c) spatial predictions at a new location also in the mixed model framework. The competing methods vary widely in terms of the function classes used (basis functions, generalized additive models, random forests, neural networks, Bayesian additive regression trees or BART) and how they incorporate spatial dependence (mixed effects models, added spatial features). Also, while RF-GP performs all the three tasks, this is not true of all methods. For example, none of the added-spatial-features methods provide  estimates of the conditional mean $\mathbb E(Y \given X)$ or the covariate effect $m$ as they use the covariates and the spatial locations together, and are only suitable for spatial predictions, i.e., estimates of $\mathbb E(Y \given X,s)$. For space constraints, we provide the details of the competing methods, and detailed performance comparisons in Section \ref{subsec:simulation} of the Supplementary Materials. Here, we briefly summarize the main findings.

We first compared performance of the methods in terms of estimation of mean function $p$ and covariate effect $m$. We summarize one set of representative results in Figure \ref{Fig:performance_plot} (left). Apart from RF-GP, the 3 competing methods that estimate $p$ are RF (Breiman's RF for binary data), BART for binary data \citep{chipman2010bart}, and a Basis GMM which uses basis functions in the covariate $X$ to model the coariate effect and a GP to model the spatial random effect within the mixed model setup. Both RF and BART cannot not account for data dependence and their performance worsens when the spatial variance $\sigma^2$ increases, producing MISEs almost $50\% - 100\%$ larger than RF-GP. When $\sigma^2$ is small, their performance is better as the price for ignoring spatial correlation is lower. Basis GLMM, like RF-GP accounts for dependence in the mixed model setup, but uses basis functions for the covariate effect. As the true covariate is $5$-dimensional, basis functions perform poorly. So, when $\sigma^2$ is small, i.e., the covariate effect dominates the spatial effect, Basis GLMM performs very poorly, producing MISE almost $200\%$ larger than RF-GP. When $\sigma^2$ is large, and the spatial effect dominates, poor estimation of the covariate effect in Basis GLMM has less of an impact and it performs similar to RF-GP. 

\begin{figure}[t!]
    \centering
    \begin{subfigure}[b]{0.5\textwidth}
        \centering
        \includegraphics[height=2.7in]{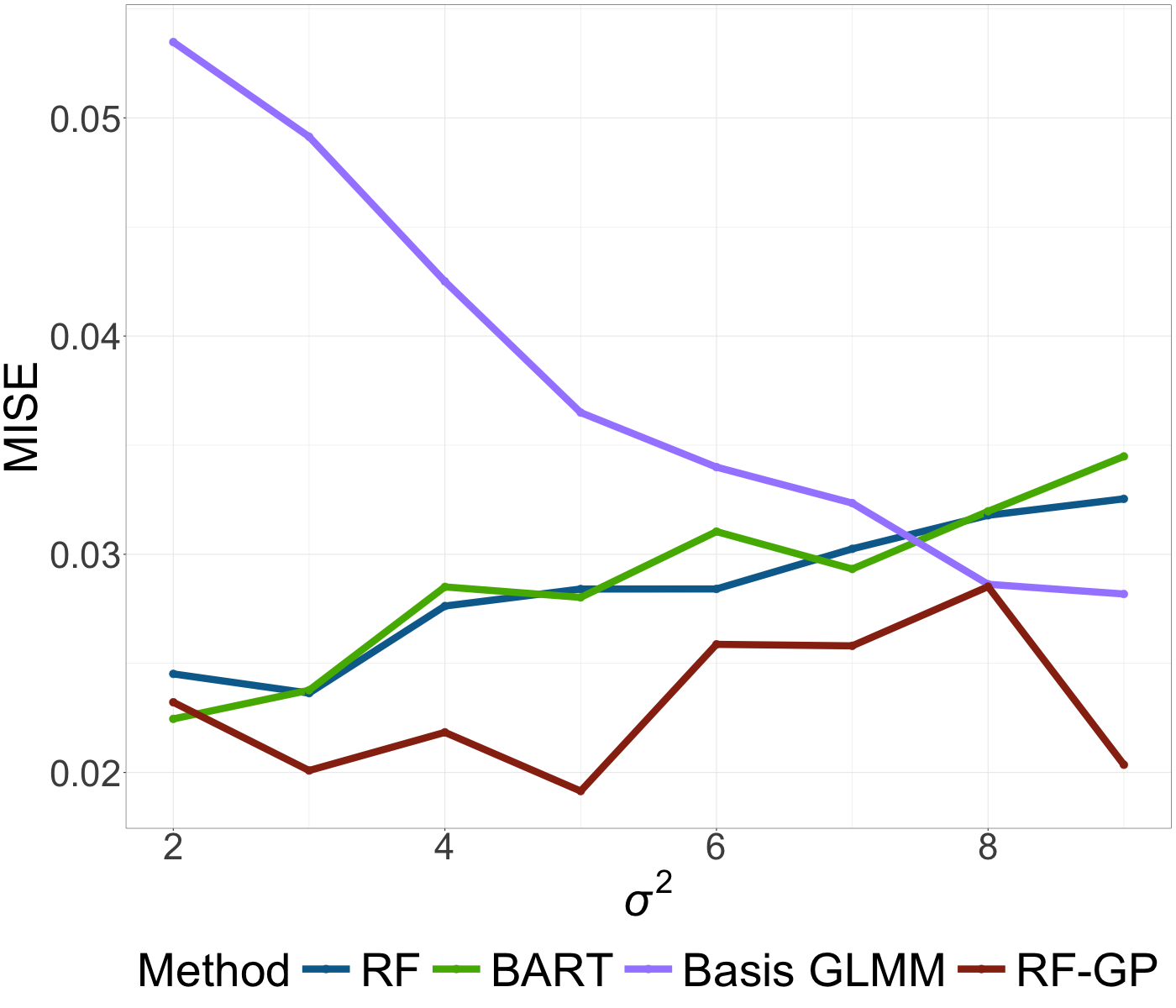}
        \captionsetup{labelfont={color=black},font={color=black}}
        \caption{\blue{$p$ estimation performance}}
    \end{subfigure}%
    ~
    \begin{subfigure}[b]{0.5\textwidth}
        \centering
        \includegraphics[height=2.7in]{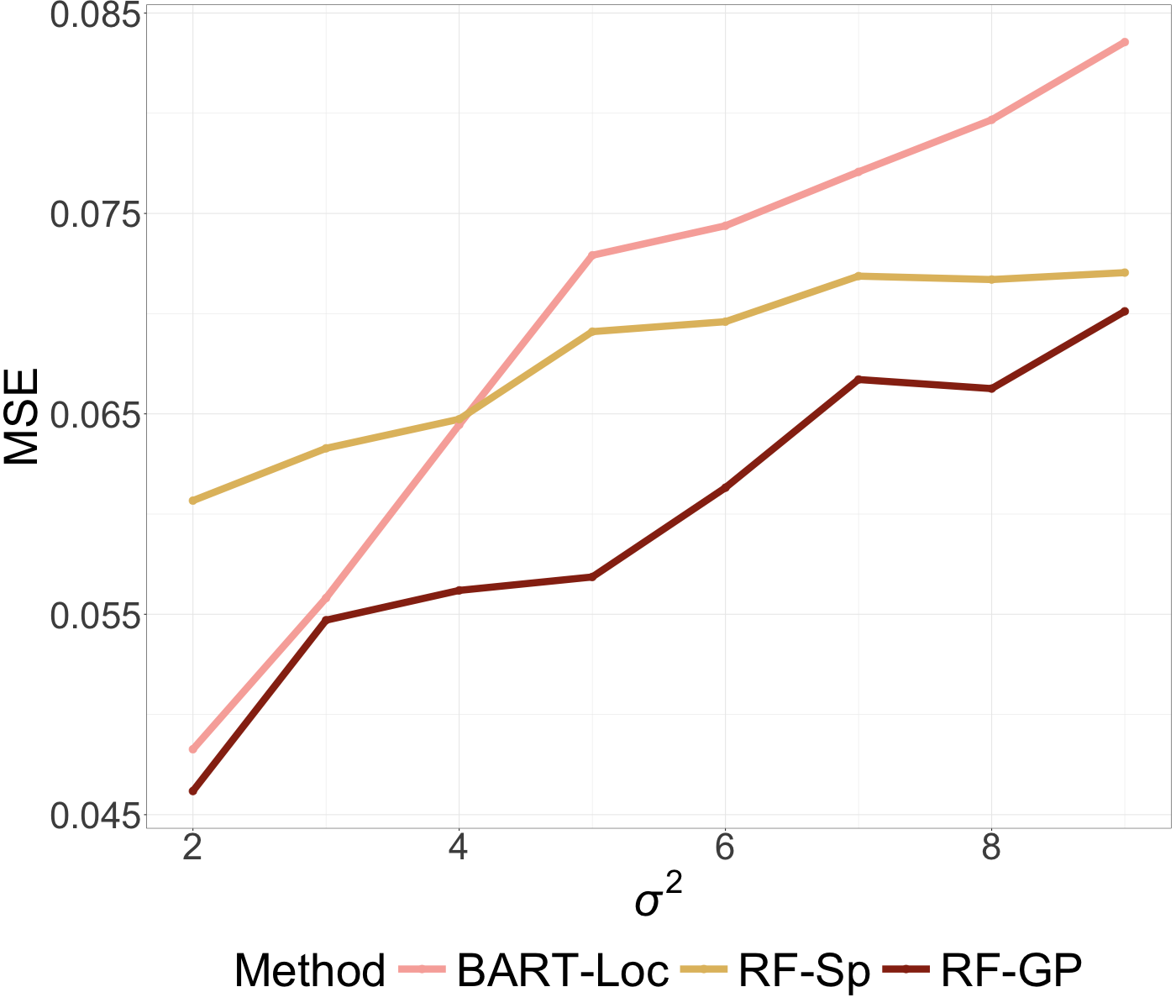}
        \captionsetup{labelfont={color=black},font={color=black}}
        \caption{\blue{Prediction performance}}
    \end{subfigure}
   \captionsetup{labelfont={color=black},font={color=black}}
   \caption{ \blue{Performance on mixed model setup \eqref{eqn:hgnlm}, where the spatial effect $w$ comes from an exponential GP, with sparial variance parameter $\sigma^2 \in \{2,3,\ldots, 9 \}$, and spatial decay parameter $\phi = 3$.}}
    \label{Fig:performance_plot}
\end{figure}
 
The detailed performance comparisons are in Tables \ref{tab:m_mixed} and \ref{tab:p_mixed} in the Supplement. We see that, across the 45 different parameter combinations, RF-GP is the best performning method in an overwhelming majority. When it is not the best method, its performance is very close to the best, whereas for each of the alternate methods, under certain scenarios their MISEs are $50\% - 200\%$ worse than RF-GP.

We then looked at prediction performance for all 11 methods. Here we only present a comparison between the 3 best methods -- RF-GP, BART-Loc (BART with spatial locaitons as added-spatial-features) and RF-Sp \citep{hengl2015mapping} which addes pairwise distances as added spatial features. We see from Figure \ref{Fig:performance_plot} (right) that BART-Loc and RF-Sp perform well at different ends of the spatial variance spectrum. When the spatial variance is low, i.e., the covariate effect dominate, RF-Sp, adding many spatial-features, drowns out the the true covariates, failing to estimate their effect, and performs poorly. It does better when the spatial variance is large, as in that scenario, not estimating the covariate effect correctly has less of an impact. The trends are reversed for BART-Loc, which only adds the spatial-coordinates, and struggles to model complex spatial dependence. This worsens its performance when the spatial effect is dominant, i.e., $\sigma^2$ is large, but has less impact when the spatial effect is low. Both methods, in one or other end of the spectrum, offer $40\%$ worse MSE than RF-GP, which performs well at both ends of the spectrum. The full comparison of the 11 methods in terms of predictive performance is given in Tables \ref{tab:prediction_mixed} and \ref{tab:misclass_mixed} of the Supplement. We see that across multiple data generation scenarios, RF-GP produces best or close-to-best performance in all scenarios, whereas each of the other method perform considerably worse for some subsets of scenarios. 

We also conducted multiple other experiments, including generating data directly from a binary auto-regressive process, to assess the mean estimation performance of the methods outside of the mixed effects model framework. We also considered quality of predictions under model misspecification where the data is not generated from a GP. RF-GP performed the best for both tasks. Furthermore, we have also shown that our method can  a) estimate the effect of each variable via a partial dependence function plot, b) can estimate the conditional average treatment effect in a causal inference context, and c) be extended to use anisotropic GP covariances and correctly identify the scale of anisotropy. For more details on the simulation results, please refer to Supplementary Section \ref{subsec:simulation} in the revised manuscript.}

\subsection{Soil type prediction}
We demonstrate the utility of RF-GP to improve the classification of spatial binary data in a real-world application. We study the spatial pattern of soil types in the Meuse data (available in the \texttt{sp} package in \texttt{R}). 
This dataset contains information on soil type, heavy metal concentration, and landscape variables at 155 locations, spread over approximately $15 \textrm{m}\times 15 \textrm{m}$ area in the flood plain of the river Meuse in the Netherlands. We consider the presence and absence data of the dominant soil type (Type 1) in the area. Soil type is primarily determined by its distance from the river. Additionally, surface water occurrence may also play a role in the soil-forming process \citep{pekel2016high}. Following \cite{hengl2018random}, we use both distance from river Meuse and surface water occurrence as the covariates for predicting the presence of the dominant soil type. 

In Figure \ref{Fig:prediction_map}, we demonstrate an example of a test and training data split ($20\%-80\%$), whereas the background color (in grey contours) demonstrates discretized levels of distances from river Meuse. A spatial structure is evident in the binary response, with most of the locations with the dominant soil type (yellow dots) being close to the river. So the probability of presence of the soil type seems to be correlated with distance from the river, a highly spatial covariate by definition. Now we investigate if there is any spatial correlation in the binary response beyond the effect of the spatial covariates. We perform prediction with classical RF which uses only the covariates and no other spatial information and compare its performance with RF-GP, which unlike the former, explicitly models the spatial correlation in addition to using the covariates in the \blue{covariate} effect. In Figure \ref{Fig:prediction_map}, we observe that RF prediction fails to account for spatial trends beyond covariate effects whereas RF-GP successfully accounts for this. 

To highlight this, we focus on two pairs of locations in the RF prediction in Figure \ref{Fig:prediction_map}. For the pair of locations in the blue oval near the river, RF predicts it is highly probable that the dominant soil type will be present in these locations, whereas \blue{in truth it is absent, as we see in the plot of the test data.} Since RF only makes use of covariate information, this is consistent with the training data, where the dominant soil type is present near the river. As RF doesn't take into consideration spatial information, it fails to account for the fact that in a number of nearby locations, albeit a little further from the river, the dominant soil type is absent. RF-GP accounts for this information and correctly predicts a much lower probability of observing the specific soil type in these locations. For the pair of points in the yellow oval in RF prediction in Figure \ref{Fig:prediction_map}, we observe an opposite scenario where RF predicts low-probability of the dominant soil type as these locations are further away from the river, but RF-GP accurately estimates a high probability of the soil type by leveraging spatial correlations. This indicates that incorporating the spatial information in RF might improve the prediction performance. 

\begin{figure}[t!]
\centering
    \includegraphics[height=2.in]{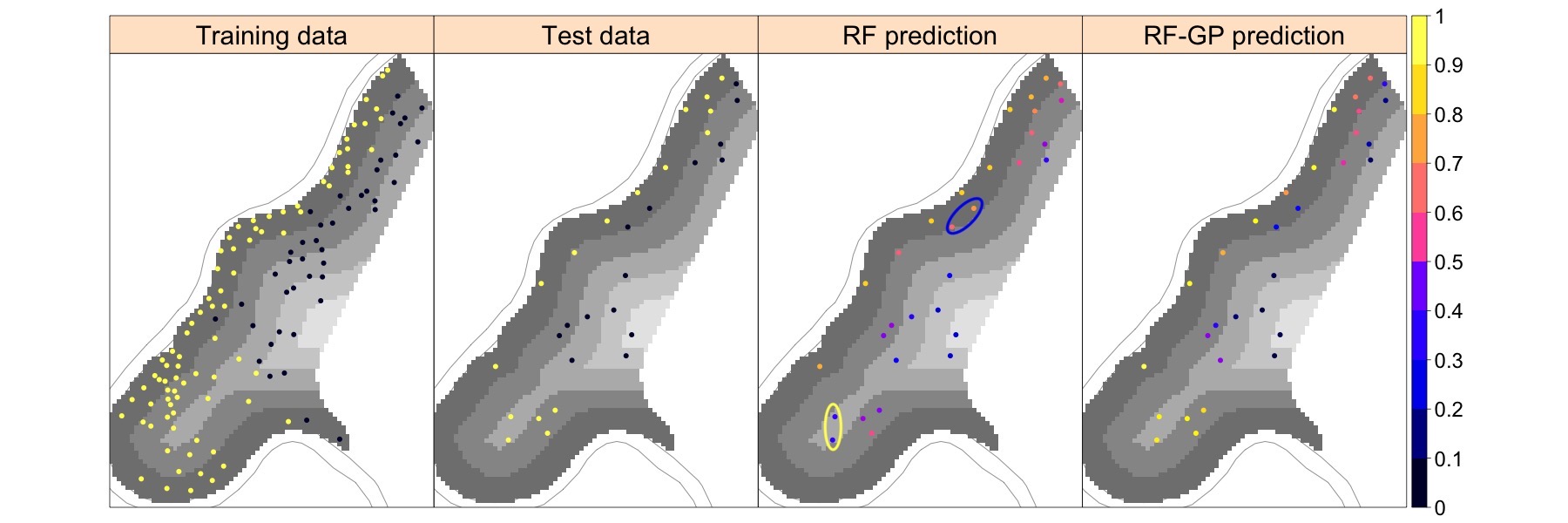}
    \caption{Plot of training and test data alongside prediction from classical RF and RF-GP overlaid on discretized distances from the river. The blue and yellow oval in RF prediction points out two pairs of points, which classical RF misclassifies due to not accounting for spatial correlation beyond the covariate effect. RF-GP explicitly models this spatial correlation and correctly classifies these pairs of points.}\label{Fig:prediction_map}
\end{figure}

In order to ensure that this finding is not an artifact of this specific train-test split, we consider \blue{500} random $20\%-80\%$ test-train splits of the data.\blue{We compare the prediction performance of naïve RF} with RF-Sp, \blue{BART-Loc,} and RF-GP -- \blue{the three best methods in terms of predictive performance from the simulations.}  
Here we do not have access to the true conditional probabilities corresponding to the underlying process generating the test data response, hence we measure the prediction performance based on test misclassification error. This computes the fraction of misclassification in binary prediction problems. 
In Figure \ref{Fig:prediction_hist}, we plot the histogram of the misclassification errors and we report the median misclassification rate corresponding to each method  
in Table \ref{table:MSE}. As expected, classical RF, which doesn't use any spatial information beyond the covariate effects, has the highest  misclassification error ($22.5\%$). 
Both \blue{BART-Loc} and RF-sp, which incorporates spatial information through additional covariates, perform similarly to each other with a misclassification rate of around $9.5\%$, and  improve  over na\"ive RF.
RF-GP, with its parsimonious modeling of spatial information, performs the best among all the competing algorithms with a misclassification rate of $6.5\%$.  
\blue{So, compared to RF-GP, both RF-Sp and \blue{BART-Loc} results in about $50\%$ worse median classification error .} RF-GP also has the thinnest tail (Figure \ref{Fig:prediction_hist}), i.e., the lowest probability of producing extremely poor misclassification rates due to sampling variability in the training data.  

\begin{figure}[t]
\centering
    \includegraphics[height=2.5in]{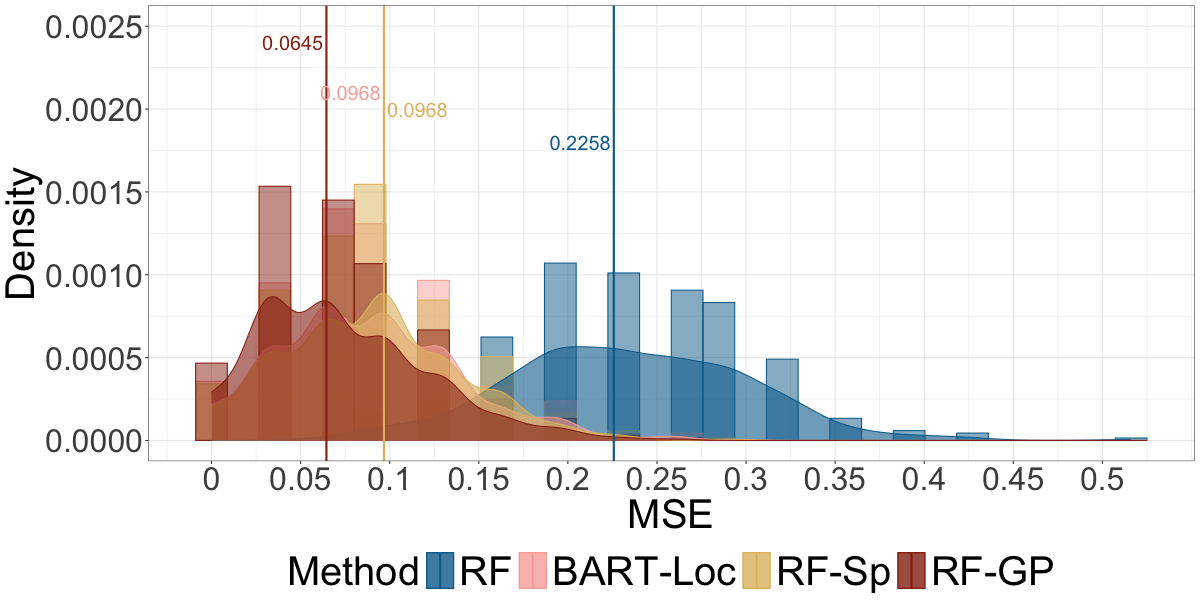}
    \caption{Histogram of the misclassification errors for RF, \blue{BART-Loc}, RF-Sp, and RF-GP in Meuse data. The vertical lines indicate the median values corresponding to different methods.}\label{Fig:prediction_hist}
\end{figure}

\begin{table}[h]
\centering
\begin{tabular}{||c | c  | c | c | c||} 
 \hline
& RF  & \blue{BART-Loc} & RF-Sp & RF-GP \\ [0.5ex] 
 \hline\hline
Misclassification error & \blue{$0.2258$}  & \blue{$0.0968$} & $0.0968$ & $0.0645$ \\  [1ex] 
 \hline
\end{tabular}
\caption{Median values of the misclassification error for the competing methods.}
\label{table:MSE}
\end{table}

\blue{We also conducted sensitivity analysis to see if the application of RF-GP with an isotropic covariance matrix is adequate for this dataset. 
As the presence of the river is a key factor here, there is definitely non-stationarity in the data and we already account for this by incorporating the distance from the river as a covariate in RF-GP. 
We first assessed if there is any further non-stationarity in the data that is not accounted for. For this, we fit another RF-GP which includes the location as additional covariates. The misclassification error rate for this version of RF-GP was identical to that of RF-GP without the location covariates.  
    
Additionally, given that direction is so important in this scenario, we checked if there is any benefit in using an anisotropic covariance matrix in RF-GP. 
We consider scale anisotropy and estimate the anisotropy parameter $\alpha$ (see Section \ref{sim:anisotropy} for details) by cross-validation by varying  $\alpha \in \{1/8, 1/4, 1/2, 1, 2, 4, 8 \}$. As no aanisotropy corresponds to $\alpha=1$, an estimate of $0.82$ suggests that there is very little anisotropy in the data beyond what is explained by the covariate. 
Also, the misclassification error from the best anisotropic model is basically indistinguishable from the misclassification error of the isotropic RF-GP. This indicates that though there might be a small amount of anisotropy, not accounting for that does not have an adverse effect on the performance of the method in terms of misclassification error.

Both of these sensitivity analyses indicate that RF-GP with distance to the river as a covariate and an isotropic GP covariance matrix is able to capture all the spatial variance that can be explained by the locations in the mean function.}

\section{Discussion}\label{sec:discussion}
\blue{Our proposed method, RF-GP, can be used, in principle, with any family of covariance functions. While most of our experiments have used isotropic covariances, in Section \ref{sec:ns} we discuss some potential ways to use RF-GP for non-stationary data and in Section \ref{sim:anisotropy} we show that RF-GP works with anisotropic covariance matrices. In highly non-stationary settings and when only spatial predictions is the goal, we belive that the fully non-parametric methods (using added-spatial-features) or directly modeling non-stationarity will do better than RF-GP with an isotropic covariance when there is adequate training data. Hence, an important future direction would be to directly use nonstationary family of working covariances in RF-GP. As non-stationary covariances often have many parameters, choosing parameters using cross-validation would be challenging. In the future, we plan to directly use the marginalized probit likelihood (Section \ref{sec:prediction}) to get MLE for the covariance parameters which should help circumvent this issue. Additionally, in Section \ref{sec:compute}, we have outlined a simple approach to further improve computational efficiency of RF-GLS and RF-GP using quantile binning. The tradeoffs of this approximation and its potential to be used with non-stationary matrices need to be studied. Extensions of RF-GLS and RF-GP to categorical and multivariate spatial data also need to be explored \citep{dey2022graphical}. 

The primary focus of the paper was on mean function estimation and spatial prediction for dependent binary data. So, we while we presented a proof-of-concept on using RF-GP for CATE estimation in dependent binary data, we do not fully explore the potential of our method for this.  
More extensive studies needs to be conducted to assess its performance under a broader set of scenarios. The theory for CATE estimation also can be strengthened by relaxing some of the design assumptions on the nature of dependence and by considering a formal causal setup. Other estimates of CATE, like ones based on S-learner should also be explored.  Also, while we used simple averaging to obtain the forest estimate from the tree estimates, conceptually one can use any other (and possibly data-driven) weights. 

It is likely that for CATE estimation, such weights would lead to better variance estimates as they downweight the noisy leaf estimates. Finally, while RF-GLS is an extension of  Breiman's RF, in principle, our proposed methodology can be adopted for honest trees, where both the split selection (tree creation) and the mean estimation will use the GLS loss but on separate subsets of data. However, the theory of honest trees relies on the independence of the data subsets used for node splitting and mean estimation. In a spatial setup where all the datapoints are correlated, this assumption will be violated, which may impact the performance of honest trees for spatial data and the vslidity of the theory. A nice aspect of honest trees is that they allow not just estimation but also inference via asymptotic normality results. So, an extension of RF-GLS using honest trees may enable inference for dependent data. 
We leave these studies for the future.

We have provided, to our knowledge, some of the first theoretical results on consistency of random forests for dependent binary data. The theoretical results are of course valid contingent upon compliance with the data generation assumptions. 
For example,  
they rely on regular design and stationarity of the dependence. Extending the results to irregular designs and non-stationary processes will be an important direction. Performance of the method under spatial confounding also needs to be studied. Recent work has shown that GLS estimators can adjust for certain types of spatial confounding. It would be interesting to see if such properties hold for random forests with GLS loss \citep{gilbert2024consistency}. We leave these studies for the future. 
} 
 \vspace{-1em}
 \section*{Acknowledgement}
 AD was partially supported by the National Institute of Environmental Health Sciences (NIEHS)
grant R01 ES033739 and by National Science Foundation (NSF)
Division of Mathematical Sciences grant DMS-1915803.
\vspace{-1 em}
\section*{Supplementary material}
Supplementary materials 
\blue{provide detailed literature review, discussion on using RF-GP under non-stationarity, computational strategies,} an outline of the main ideas used for establishing asymptotic consistency of the RF-GP for spatially dependent binary data, followed by the technical proofs, and details of the implementation and simulations. 

\vspace{-1em}
 \bibliographystyle{apalike}
\bibliography{ref}

\newpage

\renewcommand\thesection{S\arabic{section}}
\renewcommand\theequation{S\arabic{equation}}
\renewcommand\thefigure{S\arabic{figure}}
\renewcommand\thetable{S\arabic{table}}
\renewcommand\thelemma{S\arabic{lemma}}
\renewcommand\thetheorem{S\arabic{theorem}}
\renewcommand\thecorollary{S\arabic{corollary}}
\renewcommand\theproposition{S\arabic{proposition}}

\setcounter{section}{0}
\setcounter{equation}{0}
\setcounter{figure}{0}
\setcounter{table}{0}
\setcounter{lemma}{0}
\setcounter{assume}{0}
\setcounter{theorem}{0}
\setcounter{corollary}{0}
\setcounter{proposition}{0}

\begin{center}
\Large{Supplementary Materials for ``Random forests for binary geospatial data"}
\end{center}

\blue{\section{Literature review on non-linear methods for binary spatial data}\label{sec:litrev}

The existing literature on approaches for non-linear regression of binary spatial data can be broadly classified into model-based approaches and non-parametric/machine-learning approaches. We review relevant literature in both paradigms below.} 

\blue{\subsection{Model-based approaches}\label{sec:litmodel} 

\noindent\textbf{Basis functions:} We introduced the generalized linear mixed effects models for binary spatial data in (\ref{eqn:hglm}) and discussed its many benefits and one major limitation -- the assumed linearity of the covariate effect. Extensions to non-linear covariate effects for binary spatial data have often relied on basis functions in the covariate $X_i$. 
Examples include non-linear mixed effects models like the one we consider in Section \ref{sec:mixed_model} where the covariate effect $m$ is modeled using a basis function expansion in  $X_i$ 
\citep{diggle1989spline,smith1998additive}.

In the mixed effects model setup, where $m(X)$ is a functional component and $w(s)$ is a spatial random effect, it is important to discuss historical connections between functional and spatial modeling. Connections between basis functions. splines and Gaussian processes (GP) are long established. See, for example, works of \cite{kosambi1943statistics,karhunen1947under,loeve1948fonctions,kimeldorf1971some,matheron1981splines,salkauskas1982some,watson1984smoothing,nychka2000spatial}. Thus non-linear models for $m$ and/or $w$ can use any of these methods, and implement it either in a frequentist approach via regularized optimization or a Bayesian approach (by assigning priors on the coefficients or using a GP). In fact, from the representor theorem point estimators from Bayesian functional models (like GP regression) often coincide with those from frequentist methods (like RKHS regression). 

\cite{holmes2003generalized} developed a free-knot spline regression for multivariate outcome, which can deal with spatial binary outcome as a special case. They adopted a Bayesian basis function approach, assigning priors to the basis coefficients, thereby effectively modeling the covariate effect as a random effect. The covariate effect $m$ can also be modeled as a GP on the covariate domain, as done in the Bayesian model of \cite{schmidt2011considering}. \cite{wood2008locally} proposed a locally adaptive Bayesian estimation for binary regression by modelling the binary regression as a mixture of probit regression and imposing a thin plate spline prior on the each regression. In their spatial application, the covariate was just the spatial coordinates (i.e., two-dimensional) so this can be thought of a model only with a spatial effect and no covariate effect. 

While splines or basis functions work great in 1 or 2-dimensional domains and are thus often used to model the spatial effects, basis functions of the actual $D$-dimensional covariates can suffer from curse of dimensionality even for $D$ as small as $3$-$4$ \cite{taylor2013challenging}. This is because these approaches are reliant on defining local neighborhoods in the covariate domain and neighborhoods in higher dimensions are either almost empty or are non-local \citep{scott2012multivariate}.

Few basis functions approaches circumvent the curse of dimensionality by using lower-dimensional domains. 
\cite{chen2022bayesian} developed a Bayesian solution to spatial autoregressive models using free-knot splines. For the covariate effect, they use single-index models. They thus use basis-functions on a one-dimensional linear transformation of the covariate but in the process restrict the function class. Also, they did not consider binary data.\\

\noindent\textbf{Generalized additive models (GAMs):} GAMs are another example of dimension-reduced basis functions and have also been extended for spatial correlated data \cite{wood2003thin,nandy2017additive}.  
GAMs alleviate the curse of dimensionality by using additive univariate basis functions in each covariate. The spatial effect is added via a spline term as in the \texttt{mgcv} package \citep{wood2003thin}, or using Gaussian processes for the spatial random effect \citep{nandy2017additive}. The latter approach has only been done for continuous data and not binary data. However, the class of functions that can be modeled by GAMs is restrictive, excluding any function with interactions among covariates. This results in biased estimation when strong interaction is present \citep[see, for example, illustrations in][]{zhan2024neural}.\\

\noindent\textbf{Software:} Publicly-available code is lacking for many of the aforementioned model-based approaches for non-linear regression in binary spatial data. However, there are a few possible ways to estimate parameters in the basis-function based non-linear mixed effects model \eqref{eqn:hgnlm} using publicly available software. One can use a basis function expansion in $X$ to represent the unknown function $m(X)$. Additionally the spatial effect can also be represented using basis functions or splines. The model then reduces to a generalized linear model (GLM) where the actual covariates are replaced by basis functions of the covariates and the space,  and can be implemented using existing software for GLM. For example, the \texttt{mgcv} package \citep{wood2003thin} can be used to estimate $m$ and $w$. 

Another option would be to model the covariate effect using basis functions, and thus expressing it as a linear model in the basis, while modeling the spatial effect $w$ as a GP. The \texttt{spmodel} package \cite{dumelle2023spmodel} gives us a convenient way to implement such spatial generalized linear models for binary data.\\

\noindent\textbf{Bayesian and frequentist estimation and `fixed effect',`random effect' terminology:} In the mixed effects model, we use the terminology 
`covariate effect' 
to denote the effect $m(X)$ of the observed covariates $X$, and call $w$ the 
`spatial effect' as it is often modeling the effect of unmeasured spatial covariates. To the extent possible, we have 
avoided the term 
`fixed effect'  
as this terminology can be confusing, and is closely related to the interchangeability of modeling the effect of covariates or space through mean or covariance, as discussed in \cite{cressie2021few}. For example, in any fully Bayesian implementation of the mixed effects model, priors will be assigned to the basis function coefficients and thus modeling 
the covariate effect as a random effect, which when marginalized becomes part of the covariance.
Conversely, when modeling both the covariate effect and spatial effects using basis functions/splines and using a frequentist implementation (as in the \texttt{mgcv} package), both effects are modeled as fixed functions and are part of the mean. 
Our definitions of covariate effect and spatial effect are not tied to a specific analysis model or estimation method used (Bayesian or frequentist) but rather specifies a hypothesised data generation mechanism for the observed data as outlined in Section \ref{sec:mm}.}

\blue{\subsection{Machine learning approaches}\label{sec:ml}

 Machine learning approaches like regression trees and random forests  \citep{breiman1984classification,breiman2001random}, neural networks, and Bayesian additive regression trees \citep[BART][]{chipman2012bart} are now widely used for non-parametric and non-linear regression and classification. These approaches do not impose restrictive functional assumptions (many machine learning function classes are universal approximators) while having some in-built dimension reduction to counter curse of dimensionality (using compositions in neural networks, random covariate selection in random forests and regression trees). Given a set of covariates $X$ and a binary response $Y$, a machine learning method can be used to estimate the mean function $\mathbb E(Y|X)=P(Y=1\given X)$. For spatial/time-series binary data, the mean function estimate provides a direct insight into the non-linear relationship between $X$ and $Y$ and is central to accurate classification of the binary outcome in new data where the spatial/time-series information is either absent or irrelevant (e.g., predictions far away from the training data locations, or long-term time-series forecasting). 
Using a generalized mixed effects model (linear or non-linear) to estimate the mean function $\mathbb E(Y\given X)$  
can be inaccurate if any component of the mixed effects model is misspecified (such as the link, the linearity or any other restrictive assumption on the covariate effect function class, and the assumed distribution of the spatial random effect). Here, one would argue that it makes more sense to use the machine learning approaches for more non-parametric estimation of the mean.
 
However, when the data are spatially or temporally correlated machine learning regression methods do not generally explicitly encode spatial correlation as in mixed effects models. Some methods like that of  \cite{uria2016neural} consider autoregressive structure in the binary time-series data, but focuses on density estimation and not regression, i.e., they do not consider any covariates.   
Particularly, in the context of irregular spatial data which is the main focus of the paper, there are relatively few machine learning methods tailored to the features of such data. The following are two main genre of approaches to account for spatial correlation in machine learning algorithms.

\noindent \textbf{Residual kriging:} This method completely ignores the correlation in $Y$ while estimating $\mathbb E \left( Y| X\right)$, followed by kriging on the residuals for spatially informed predictions \citep{fayad2016regional,fox2020comparing}. Completely ignoring correlations during mean function estimation  has been shown to severely impact the performance of machine learning methods \citep{saha2023random,zhan2024neural}. Furthermore, the idea of residual kriging works for spatial predictions on continuous outcomes because of the additivity of the outcome in terms of the mean function and the spatial error. Such additivity does not hold for binary data.  

\noindent\textbf{Added-spatial-features: }
Another approach is to add spatial covariates like latitude-longitude, pairwise distances between locations or basis functions \citep{hengl2018random,wang2019nearest,chen2024deepkriging} \blue{as additional features (covariates) in the machine learning algorithm. `Added-spatial-features' can be used in any machine learning method.   
However, these approaches only offer an estimate of $\mathbb E(Y \given X,s)$ which} works for \blue{spatial} predictions.  
\blue{These methods do not} estimate the mean function $\mathbb E \left( Y| X\right)$. Also, the high dimensionality of the augmented feature space leads to curse of dimensionality issues as the many added spatial covariates can often drown out the few true covariates $X$ leading to poor performance when the covariate effect is dominant. This has been illustrated in  \cite{saha2023random} and \cite{zhan2024neural}.\\

\noindent \textbf{Challenges of using machine learning within generalized mixed-effect models: } 
A third alternative to residual kriging and added spatial features approaches would be to directly embed a machine learning method within the mixed effects model to model the covariate effect $m$, while keeping everything else same and thus enjoying all benefits of the mixed effects model framework. While any machine learning method can conceptually be used with the negative log-likelihood of the mixed effects model as the loss function, the challenge would be to deal with the high-dimensional vector of spatial random effects in an optimization based approach.  
In a Bayesian setup, one possibility to accomplish this would be to model $m$ using BART while modeling $w$ via GP. We did not find any implementation of such a `BART+GP' method that that allows an additive GP random effect to a BART-modeled covariate effect. As both BART and GP are computationally expensive, it is also challenging to run such a model in off-the-shelf Bayesian software like Stan.  

Alternative ways of using BART for binary spatially-dependent dat would be to use the locations (e.g., latitude-longitude) or basis functions of locations as additional covariates in BART. This {\em added-spatial-features BART}, will inherit the same limitations as their neural network and random forest analogues as discussed above.

\section{Modeling non-stationarity}\label{sec:ns}

Much of the spatial literature on Gaussian processes focuses on isotropic or stationary models with no unmeasured covariates and stationarity. Our methodology is also primarily illustrated using stationary covariance functions. Additionally, we assume that all relevant covariates are measured whereas in practice there may be some missing covariates, which is closely related to the issue of non-stationarity. We now discuss 
various ways to account for missing covariates and non-stationarity.\\ 

    \noindent\textbf{Non-spatial missing covariate: } Consider the scenario where we are missing one or more covariates $X_{miss}$ on which the outcome depends. If $X_{miss}$ are not functions of space, and are also independent of the observed covariates $X_{obs}$, then letting $X=(X_{obs},X_{miss})$ we can write $\mathbb E(Y \given X_{obs}) = p_{obs}(X_{obs}) = \int p(X) dX_{miss}$. So,  
    using RF-GLS with the observed $X_{obs}$ will yield a good estimate of this quantity of $p_{obs}(X)$, the marginal mean function after integrating out the unobserved covariates. \\
    
    \noindent\textbf{Missing confounders:} If $X_{miss}$ also influences some of the observed $X_{obs}$, then it is a confounder as it impacts both $X$ and $Y$. Then there will be bias in the mean function estimate. This will be true of any non-parameteric regression method under the presence of unmeasured confounding. If there is an unmeasured covariate that is not a measureable function of observed quantities but is correlated with them, then there is no validity for most methods (unless making very strong parametric assumptions), and one can primarily just do sensitivity analysis. \\

    \noindent\textbf{Spatial missing covariate and mean non-stationarity: } Under the mixed effects model setup, if the missing covariate $X_{miss}=f(s)$ is a fixed function of space, and $m$ is additive, then $m(X)=m_{obs}(X_{obs}) + m_{miss}(X_{miss}) = m_{obs}(X_{obs}) + m_{miss}(f(s))$. Then $$m(X) + w(s) = m_{obs}(X_{obs}) + w^*(s)$$ where $w^*(s) \sim GP(\mu(s),C(s,s'))$ is a GP with a non-stationary mean $\mu(s) = m_{miss}(f(s))$ and stationary covariance $C$. As long as $\mu(s)$ is a relatively smooth function of $s$ (compared to $C$), the GP $w^*(s)$ can be modeled well by a zero-mean stationary GP \citep[see discussion on equivalence of two GPs with zero and non-zero means in Chapter 4 of ][]{stein2012interpolation}. Also, Matérn kernel is a universal kernel and can estimate any fixed smooth function of space arbitrarily closely. In fact, non-parametric function estimation properties of Matern are well established \cite{van2008rates,van2011information,yang2015semiparametric}. \\
    
    In our simulations in Section \ref{sec:sim}, we have also shown that the methodology is robust to unmeasured covariates that are spatially structured fixed functions. We consider the case where the spatial error component comes from a fixed function of the locations. Simulation shows that our proposed method performs comparably or better (in low signal-to-noise ratio) to tree-based methods with location as a covariate, in spite of the later having a clear advantage in this simulation setup. For more details, please refer to simulation section \ref{sim:mean_nonstationary}. \\

    \noindent\textbf{Added spatial features + stationary covariance:} For the more general case where the additive assumption on the covariate effect does not hold, we have the decomposition, 
    $$
    \begin{aligned}
    m(X) + w(s) = m_{obs} \left(X_{obs} \right) + m \left(X_{obs}, X_{miss} \right) - m_{obs} \left(X_{obs} \right)  + w \left(s\right).
    \end{aligned}$$

    If $X_{miss}$ depends only on $X_{obs}$ and $s$, then we can write $w^*(X_{obs},s) = m \left(X_{obs}, X_{miss} \right) - m_{obs} \left(X_{obs} \right)  + w \left(s\right)$. \cite{schmidt2011considering} developed an approach to model such joint functions of covariates and space using a GP directly on the joint domain of $X_{obs}$ and $s$ which can be used to model such data. 

    However, it is a possible to also model such data using RF-GP. As the missing covariate $X_{miss}=f(X_{obs},s)$ a fixed function of space and observed covariates  
    $$m(X) + w(s) = m(X_{obs},f(X_{obs},s)) + w(s) = m^*(X_{obs},s) + w(s).$$ Using this reparametrization, the fixed effect $m^*$  is a function of the observed covariates and space, and there is a stationary random effect $w(s)$. One can then use RF-GP adding some spatial features in the random forest part to model $m^*$ and it is still justified to use a stationary covariance for the spatial random effect. We can use generic features like latitude-longitude, basis functions or features tailored to the specific application. While the interpretation of $m$ as just the effect of the observed covariates $X_{obs}$ no longer holds, we expect the method to give good predictions as it is still correctly specified as sum of the $m^*$ and $w$ terms. We do this in our real data analysis, one covariate is distance to the river and is thus a spatial feature. \\

    \noindent\textbf{Anisotropy: } We note that while we have developed our method using isotropic covariances, in principle it can also be used with anisotrpic or non-stationary covariances. 
    Often in the real world, direction of the data becomes important, where the distance with respect to one coordinate has more effect than the distance with respect to the other coordinate. In these scenarios, anisotropy is a more realistic assumption. We now have added a simple proof-of-concept experiment showing that our proposed approach can be used with anisotropic covariances and can correctly identify scales of anisotropy, 
    improving predictive performance.
    For more details, please refer to Supplementary Materials Section \ref{sim:anisotropy}. We do plan to invetigate the performance of RF-GP with anisotropic models in more details in the future.\\

    \noindent\textbf{Covariance non-stationarity:} 
As mentioned earlier, conceptually, RF-GP works with any choice of covariance family for the GP. So if a form of covariance non-stationarity is reasonable for an applicaiton, like the class of non-stationary covariances of \cite{paciorek2006spatial}, the method can in principle be used with a working covariance matrix derived from that non-stationary GP. However, as non-stationary covariances have many more parameters, the current implementation of RF-GP which estimates these hyper-parameters 
 using cross-validation is not ideal to accommodate non-stationarity covariances.  
 In the future we plan to directly use the marginalized probit likelihood (Section \ref{sec:prediction}) to get MLE for the covariance parameters which should help circumvent this issue.

    We also note that most non-stationary models of covariance we know of, make very parametric assumptions about the form of non-stationarity, which are again prone to misspecification. Also, non-stationary models are considerably more parameter rich, and they can often be poorly estimated with limited data.  
    While we will explore this extension in the future, we believe that for many forms of  
    non-stationarity, RF-GP with a stationary GP and perhaps with added-spatial-features in the RF part can be deployed as they trade off some accuracy for a substantial amount of parsimony and computational efficiency. 
}

\section{Proof outline}
\subsection{Assumptions}\label{sec:assumptions}
We first provide more technical versions of  
Assumptions \ref{as:data_gen} and \ref{as:working_cov}. 
\begin{assume}{\color{black}[Conditions on generating model]\label{as:data_gen} We assume the data is generated from the following model:
\begin{equation}
    \label{eqn:true_model}
        Y \in \left\{ 0, 1\right\}; \:\:\:\: \mathbb{E}(Y| X) = g(m(X));
\end{equation}}

\begin{enumerate}[label=\textbf{1.\alph*}]
\item \label{item:assumption_x} $X = (X_{1}, X_{2}, \cdots, X_{D}) \: X_{d} \overset{i.i.d}{\sim} \mathrm{Unif}[0,1], \:\:\forall d = 1, 2, \cdots, D$.
\item \label{item:assumption_m} $m : \mathbb{R}^D \to \mathbb{R}$ is an additive function, i.e.  
$m(X) = \sum_{d = 1}^D m_d(X_d)$, where each component $m_d$ is continuous on $[0,1]$.
{\color{black}\item \label{item:assumption_h} $g: \mathbb{R} \to (0,1)$ is a \textit{strictly increasing} function.
\item \label{item:assumption_w}   $\{(Y_i,X_{i-l})\}_i$ is a stationary ergodic {\em absolutely regular} ($\beta$-mixing) process for any finite lag $l \in \mathbb Z$, the set of integers. 
\item \label{item:assumption_differentiability}  There exists $M< \infty$, such that $z \mapsto g(z)$ is continuously differentiable on $[-M, M]$, and $|m(\bx)|\leqslant M, \: \forall \bx \in [0,1]^D$.}

\end{enumerate}
\end{assume}

\begin{assume}[Regularity condition on working precision matrix]\label{as:working_cov} 
We make the following regularity and sparsity assumptions on the working precision matrix $\bQ$ in \eqref{eqn:DART}:
\begin{enumerate}[label=\textbf{2.\alph*}]
\item \label{item:assumption_chol} The Cholesky factor of the working precision matrix $\bQ$ has the following sparsity structure: 
	\begin{equation}
	\label{eqn:assumption_chol}
	\bQ^{\frac 12} = \left(\begin{array}{ccccc}
	\bL_{q \times q} & 0 & 0 & \cdots & \cdots\\
	\multicolumn{2}{c}{\rho^\top_{1 \times (q+1)}} & 0 & \cdots & \cdots\\
	0 & \multicolumn{2}{c}{\rho^\top_{1 \times (q+1)}} & 0 & \cdots\\
	\vdots & \multicolumn{3}{c}{\ddots} & \vdots \\
	\cdots & 0 & 0 & \multicolumn{2}{c}{\rho^\top_{1 \times (q+1)}}
	\end{array} \right), 
	\end{equation}
	where ${\rho}  = (\rho_q, \rho_{q-1},\cdots, \rho_{0} )^\top \in \mathbb{R}^{q+1}$ for some fixed $q \in \mathbb N$, and $\bL$ is a fixed lower-triangular $q \times q$ matrix.

\item \label{item:assumption_diagonas_dominance}  The working precision matrix $\bQ$ is diagonally dominant, i.e. $\bQ_{dd} - \sum_{l \neq d} |\bQ_{dl}| > \xi$ for some constant $\xi >0, \forall l  \in \{1,2,\cdots, d \}$. 
\end{enumerate}
\end{assume}

\subsection{Main ideas}\label{sec:outline}
We provide detailed technical proofs of all the results in Section \ref{sec:proofs} of the Supplementary Material. In this section, we briefly highlight the main new ideas used to address theoretical challenges exclusive to the binary data setting. As discussed at the beginning of Section \ref{sec:consistency}, these challenges do not occur in the established theory of RF and RF-GLS for continuous data due to fundamental differences in both the data generation process and the Gini  loss functions used for node splitting in the respective settings.

The inception of our theoretical explorations is the equivalence of the Gini impurity and the variance-difference loss (Theorem \ref{theorem:Classification_equals_CART}). This result, which was also central to the development of the entire method, helps reconcile much of the theory of tree and forest estimators for binary and continuous data types.
The asymptotic limit of the split criterion is a key quantity in studying tree and forest estimators. For continuous data, the limit of this GLS-loss is established in \cite{saha2023random} using the additive representation from a mixed effects model perspective $Y_i \overset{ind}{=} m(X_i) + w_i + \eps_i$. As $X_i$ and $\eps_i$ are i.i.d. processes, $w_i$ is a $\beta$-mixing dependent process, and the three processes are independent of each other, the limit of the GLS-loss can be studied via respective limits of linear or quadratic forms for each of the three processes. 

Such a linear separable representation of mean and dependent error does not hold for binary data generated from a general model in \eqref{eqn:genbin} (or for a generalized mixed effects model of the form (\ref{eqn:hgnlm}). We study the the limit of the GLS-loss for correlated binary outcomes and establish the following result on the limiting theoretical criterion. \\

\begin{lemma}[Limit of the GLS split criterion for binary data]
	\label{lemma:DART-theoretical}
	Under Assumptions \ref{as:data_gen} and \ref{as:working_cov}, for an RF-GP tree for dependent binary data built with fixed data-independent partitions, the empirical split criterion (\ref{eqn:DART}) for splitting a node $\calA$ at $(d,c)$ creating child nodes $\calA^L$ and $\calA^R$, converges almost surely to the following:
	\begin{equation}
	\label{eqn:lim_DART}
	\begin{array}{cc}
	&\Delta_{\bQ}^*((d,c))= \gamma(\bQ) \nu^*(d,c) \mbox{ where } \nu^*(d,c)=  Vol(\mathcal A) \Big[\mathbb{V}(Y | {X} \in \mathcal A) - \\
 &\mathbb{P}({X} \in \mathcal A^R | {X} \in \mathcal A)\mathbb{V}(Y | {X} \in \mathcal A^R)
	- \mathbb{P}({X} \in \mathcal A^L |  {X} \in \mathcal A)\mathbb{V}(Y | {X} \in \mathcal A^L) \Big].
	\end{array}                    
	\end{equation}
where $\gamma(\bQ)$ is some constant not dependent on $n$ and $Vol(\calA)=P(\bX \in \calA)$. 
\end{lemma}

The limiting term $\nu^*(d,c)$ on the right-hand side of (\ref{eqn:lim_DART}) is the theoretical split criterion. A central component in establishing consistency of tree and forest estimators is to show the regularity of an estimator built with this theoretical criterion. In particular, if $\nu^*(d,c)$ is constant for all possible splits $(d,c)$ of a node $\calA$, then the estimator must be constant in that node. For the continuous setting, \cite{scornet2015consistency} established this property by proving the result first for the case with only $D=1$ covariate. Then for $D > 1$, they exploit the additive form of $m$ (Assumption \ref{item:assumption_m}) and that $\mathbb{E}(Y|X)=m(X)$ to have
\begin{equation}\label{eq:add}
\int_{x_2,\ldots,x_D} \mathbb{E}(Y|X=\bx) dx_2\ldots dx_D = m_1(x_1) + \mbox{constant}
\end{equation}
thereby reducing the problem to a 1-dimensional setting that has already been addressed.

For binary data, the mean of the outcome is a nonlinear transformation of the covariate effect $m(X)$, given as $\mathbb E(Y | X) =  g\left(m(X)\right)$, 
where $g : \mathbb R \to (0, 1)$ is nonlinear function. Even with the component-wise additivity Assumption \ref{item:assumption_m} for $m(\cdot)$, $g(m)$ is not additive and hence the identity (\ref{eq:add}) does not hold. \blue{In the following result, we develop a novel proof for controlling the variation of $p$ in the leaf nodes that relaxes the additivity assumption. The proof is technical and uses the Taylor expansion and the regularity of the link.}\\ 

\begin{lemma}[Regularity of the non-additive mean function]
\label{lemma:noves_mean_noseparable} If Assumption \ref{as:data_gen} is satisfied and $\Delta_{\bQ}^*((d,c)) \equiv 0$ (in \eqref{eqn:lim_DART}) for all permissible splits $(d,c)$ in a given cell $\mathcal{A}$, then the mean of the binary outcome conditional on covariates, i.e., $\mathbb{E}(Y \given X)$ is constant on  $\mathcal{A}$.\\
\end{lemma}

After addressing these challenges specific to the binary data setting,
 the rest of the proof for consistency of the mean function (Theorem \ref{th:main_gh}) roughly adapts this established theory of consistency of least square estimators \cite{gyorfi2006distribution}, throughout making careful adjustments for the binary nature of the data. 

{\color{black} Our results in Section \ref{sec:main} only rely on  
a distribution-free moment-based data generating process for the correlated binary outcome. However, the main special case of interest is the generalized mixed effects model of the form \eqref{eqn:hgnlm} for which we present the results in Section \ref{sec:mm}. For these results we need to show that with suitable assumptions on the random effect $w$ in \eqref{eqn:hgnlm}, the data generating model in \eqref{eqn:hgnlm} satisfies Assumption \ref{as:data_gen}. This establishes the consistency of $p(X) = \mathbb{E}(Y|X)$ estimation in mixed effects model framework \eqref{eqn:hgnlm} as a corollary of Theorem \ref{th:main_gh}:\\

\begin{corollary}[Consistency for mean function for mixed effects models]\label{th:main_gh_mixed_corollary}
	For data generated from the mixed effects model \eqref{eqn:hgnlm}, under Assumptions \ref{item:assumption_x} - \ref{item:assumption_h} and Assumptions \ref{as:working_cov} - \ref{as:tn_rate}, RF-GP is $\mathbb{L}_2$-consistent for the mean function, i.e., 
	$
	\lim_{n \to \infty} \mathbb{E} \int \left(\bar{p}_n(X) - p(X) \right)^2 \, dX = 0$, if the spatial random effect $w$ satisfies the following
    \begin{enumerate}[label=\textbf{*1.\alph*}]\setcounter{enumi}{3}
        \item \label{item:assumption_w_new}  $w \independent{X}$, and $\{w_i\}$ is a stationary, ergodic, {\em absolutely regular} ($\beta$-mixing) spatial process.
\item \label{item:assumption_differentiability_new}  $w$ has finite first moment, i.e $\mathbb E |w| < \infty$, and there exists $M< \infty$, such that $z \mapsto \mathbb{E}_W(h(z + W))$ is continuously differentiable on $[-M, M]$, and $|m(\bx)|\leqslant M, \: \forall \bx \in [0,1]^D$.
    \end{enumerate}
\end{corollary}

To prove Corollary \ref{th:main_gh_mixed_corollary}, we show that if $w$ in \eqref{eqn:hgnlm} is a $\beta-$ mixing process, then the resulting outcome $\{ Y_i\}$ is also a $\beta-$ mixing process using the following technical lemma:\\

\begin{lemma}[Latent $\beta$-mixing representation]
\label{lemma:joint_mixing}
If $\{ X_i\}$  satisfies Assumption \ref{item:assumption_x} and $\{ w_i\}$ satisfies Assumption \ref{item:assumption_w_new},
then $\{ (Y_i,X_{i-l})\}_i$ is a stationary ergodic $\beta$-mixing process for any fixed lag $l \in \mathbb Z$. \\
\end{lemma}

Lemma \ref{lemma:joint_mixing} proves that assumptions \ref{item:assumption_w_new} and \ref{item:assumption_differentiability_new} on the spatial process $w$ in the data generating model \eqref{eqn:hgnlm} are sufficient conditions to satisfy Assumptions \ref{item:assumption_w} and \ref{item:assumption_differentiability} respectively, which in turn proves that Corollary \ref{th:main_gh_mixed_corollary} is a direct consequence of Theorem \ref{th:main_gh}. }

To extend the result on the consistency of the mean function $\mathbb{E}(Y|X)$ estimation to the consistency of the covariate effect $m(X)$ (Theorem \ref{th:main_m}) estimation in mixed effects model framework in \eqref{eqn:hgnlm} via link-inversion, a final component is regularity (continuity and boundedness) of the inverse function. Continuity is immediate as the mean function $\mathbb{E}(Y|X)$ is continuous under Assumptions \ref{item:assumption_m}, \ref{item:assumption_h}, and \ref{item:assumption_differentiability_new}. To prove boundedness, we first prove that the GLS estimator (\ref{eqn:DART_mean}) is bounded. \\

\begin{lemma}[Uniform bound of the estimator] \label{lemma:bounded_estimate} \blue{Under Assumption 
\ref{as:working_cov}, } 
the GLS \blue{regression tree} estimator (\ref{eqn:DART_mean}) \blue{for binary outcomes $Y_i$} 
is uniformly bounded in $n$. 
\end{lemma}

This result in turn allows the inverse to be defined for a compact support which combined with continuity ensures uniform bounds for the inverse function and helps prove consistency of the covariate effect (Theorem \ref{th:main_m}) via application of the continuous mapping Theorem. The uniform bound also improves the required scaling (Assumption \ref{as:tn_rate}) compared to that required in the continuous case. 

\section{Time-series GP with autoregressive error}\label{sec:supar}
We also study the property of RF-GP for binary time series with an autoregressive covariance structure for the temporal random effects. Assuming a lag or order $q$, i.e., $\mathrm{AR}(q)$ structure, the hierarchical model can be expressed as
\begin{equation}\label{eq:ar}
\begin{aligned}
Y_i &\overset{ind}{\sim} \mbox{Bernoulli}(h(m(X_i) + w_i)) \\
w_i &= a_1 w_{i-1} + a_2 w_{i-2} + \ldots + a_q w_{i-q} + \eta_i 
\end{aligned}
\end{equation}
Here $\eta_i$ is a realization of a white noise process at time $i$. It can be shown that $w_i$ is a $\beta$-mixing process \citep{mokkadem1988mixing}. When fitting RF-GP, we use a working covariance matrix also having an autoregressive structure (as justified in Section \ref{sec:correlation_estimation}). The lag order for the working covariance matrix need not match the auto-regression lag assumed for the data generation, and for all choices of lag the working covariance matrix satisfies the regularity assumptions \ref{as:working_cov}. Hence, we immediately have the following result:
\begin{corollary}\label{prop:ar_application} Consider a binary time series (\ref{eq:ar}) in a generalized mixed effects model setup with probit or logit link and with temporal random effects distributed as a 
sub-Gaussian stable $\mathrm{AR}(q)$ process. Let $\bQ$ denote a diagonally dominant working precision matrix from a stationary $\mathrm{AR}(q')$ process. 
Then RF-GP using $\bQ$ produces an $\mathbb L_2$ consistent estimate of $\mathbb E(Y | X)$ and $m(X)$. 
\end{corollary}
\section{Technical Proofs}\label{sec:proofs}
\subsection{Proof of Theorem \ref{theorem:Classification_equals_CART}}
\label{Appendix:Theorem_1}
\begin{proof}
In order to prove this, we begin with the following proposition showing that the Gini node impurity measure is a constant multiplier of the node variance. 
\begin{proposition}
\label{proposition:Gini_equas_variance} 
For given data $(Y_i, X_i) \in \{0, 1\} \times \mathbb R^{D}; i = 1, 2, \cdots, n$, and a node $\mathcal T \subseteq \mathbb R^D$
\begin{equation}
    \label{eqn:Gini_equas_variance}
     \frac{1}{2}I(\mathcal{T}) = \frac{1}{\sum_{i = 1}^n \mathds 1 _{\{X_i \in \mathcal T\}}} \sum_{i = 1}^n \mathds 1 _{\{X_i \in \mathcal T\}} (Y_i - \bar{Y}^{(\mathcal T)})^2,
\end{equation}

where, $ \bar{Y}^{(\mathcal T)} = \frac{1}{\sum_{i = 1}^n \mathds 1 _{\{X_i \in \mathcal T\}}}\sum_{i = 1}^n \mathds 1 _{\{X_i \in \mathcal T\}}Y_i$.
\end{proposition}

\begin{proof}
We begin our proof with the following observation:

$$
\begin{aligned}
\bar{Y}^{(\mathcal T)} &= \frac{1}{\sum_{i = 1}^n \mathds 1 _{\{X_i \in \mathcal T\}}}\sum_{i = 1}^n \mathds 1 _{\{X_i \in \mathcal T\}}Y_i= \frac{\sum_{i = 1}^n \mathds 1 _{\{Y_i = 1\} \cap \{X_i \in \mathcal T \}}}{\sum_{i = 1}^n \mathds 1 _{\{X_i \in \mathcal T\}}}= p^{(\calT)},
\end{aligned}
$$

where $p^{(\mathcal T)}$ is 
\begin{equation*}
    p^{(\mathcal{T})} = \frac{\sum_{i = 1}^n \mathds 1 _{\{Y_i = 1\} \cap \{X_i \in \mathcal T \}}}{\sum_{i = 1}^n \mathds 1 _{\{X_i \in \mathcal T\}}}.
\end{equation*}  and the second inequality follows from the binary nature of $Y$. Substituting $p^{(\calT)}$ for $\bar Y^{(\mathcal T)}$ in \eqref{eqn:Gini_equas_variance} we have,
$$
\begin{aligned}
&\frac{1}{\sum_{i = 1}^n \mathds 1 _{\{X_i \in \mathcal T\}}} \sum_{i = 1}^n \mathds 1 _{\{X_i \in \mathcal T\}} (Y_i - \bar{Y}^{(\mathcal T)})^2\\
=& \frac{1}{\sum_{i = 1}^n \mathds 1 _{\{X_i \in \mathcal T\}}} \sum_{i = 1}^n \mathds 1 _{\{X_i \in \mathcal T\}} (Y_i - p^{(\calT)})^2 \\
=& \frac{1}{\sum_{i = 1}^n \mathds 1 _{\{X_i \in \mathcal T\}}} \left[\sum_{i = 1}^n \mathds 1 _{\{Y_i = 1\} \cap \{X_i \in \mathcal T \}} (1 - p^{(\calT)})^2 + \mathds 1 _{\{Y_i = 0\} \cap \{X_i \in \mathcal T \}} ( - p^{(\calT)})^2\right] \\ 
=& \left(1 - p^{(\calT)}\right)^2 \frac{\sum_{i = 1}^n \mathds 1 _{\{Y_i = 1\} \cap \{X_i \in \mathcal T \}}}{\sum_{i = 1}^n \mathds 1 _{\{X_i \in \mathcal T\}}} +  (- p^{(\calT)})^2 \left( 1 - \frac{\sum_{i = 1}^n \mathds 1 _{\{Y_i = 1\} \cap \{X_i \in \mathcal T \}}}{\sum_{i = 1}^n \mathds 1 _{\{X_i \in \mathcal T\}}} \right)\\
&= p^{(\calT)}\left(1 - p^{(\calT)}\right)^2 + p^{{(T)}^2}\left(1 - p^{(\calT)}\right) \\
&= p^{(\calT)}(1-p^{(\calT)})\\
&= \frac{1}{2}I(\mathcal{T})\\
\end{aligned}
$$

where, the last equality follows from \eqref{eqn:Gini}. This completes the proof of the Proposition \ref{proposition:Gini_equas_variance}.
\end{proof}

Substituting the value of $I(\mathcal T), I(\mathcal L)$ and $I(\mathcal R)$ from Proposition \ref{proposition:Gini_equas_variance} in \eqref{eqn:split_Gini}, we have,

\begin{equation}
    \begin{aligned}
        \Delta_n^{CT}((d,c)) &= \frac{2}{\sum_{i = 1}^n \mathds 1 _{\{X_i \in \mathcal T\}}} \sum_{i = 1}^n \mathds 1 _{\{X_i \in \mathcal T\}} (Y_i - \bar{Y}^{(\mathcal T)})^2 \\ & - \frac{\sum_{i = 1}^n \mathds 1 _{\{X_i \in \mathcal L\}}}{\sum_{i = 1}^n \mathds 1 _{\{X_i \in \mathcal T\}}} \frac{2}{\sum_{i = 1}^n \mathds 1 _{\{X_i \in \mathcal L\}}} \sum_{i = 1}^n \mathds 1 _{\{X_i \in \mathcal T\}} (Y_i - \bar{Y}^{(\mathcal L)})^2 \\ &- \frac{\sum_{i = 1}^n \mathds 1 _{\{X_i \in \mathcal R\}}}{\sum_{i = 1}^n \mathds 1 _{\{X_i \in \mathcal T\}}}\frac{2}{\sum_{i = 1}^n \mathds 1 _{\{X_i \in \mathcal R\}}} \sum_{i = 1}^n \mathds 1 _{\{X_i \in \mathcal R\}} (Y_i - \bar{Y}^{(\mathcal R)})^2 \\ &= \frac{2}{\sum_{i = 1}^n \mathds 1 _{\{X_i \in \mathcal T\}}} \bigg[ \sum_{i = 1}^n \mathds 1 _{\{X_i \in \mathcal T\}} (Y_i - \bar{Y}^{(\mathcal T)})^2 - \sum_{i = 1}^n \mathds 1 _{\{X_i \in \mathcal L\}} (Y_i - \bar{Y}^{(\mathcal L)})^2 \\ &- \sum_{i = 1}^n \mathds 1 _{\{X_i \in \mathcal R\}} (Y_i - \bar{Y}^{(\mathcal R)})^2 \bigg]\\ & = 2 v_{n}^{RT}((d,c))
    \end{aligned}
\end{equation}

This completes the proof of the Theorem \ref{theorem:Classification_equals_CART}.
\end{proof}

\subsection{Proof of Proposition \ref{lemma:link_inverse_existance}}
\begin{proof}
Part (a): Since, $1 > h > 0 $, and $f(t) \geqslant 0$, we have 
$$
0 < \int 
h(z_0 + t)d\mathbb F_w(t) = g_{h,\mathbb F_w}(z_0) < \int_{\mathbb R} d\mathbb F_w(t) = 1.
$$
Additionally, for any $ z_1 > z_0$ we have 
$$
g_{h,\mathbb F_w}(z_1) - g_{h,\mathbb F_w}(z_0) = \int
\left(h(z_1 + t) - h(z_0 + t)\right)d\mathbb F_w(t) > 0
$$
as $h$ is strictly increasing. 
Hence $g_{h,\mathbb F_w}$ is strictly increasing. Letting $I$ be the support of $m(\cdot)$ and $im(I)$ be its image after applying $m$, there exists an inverse map $g_{h,\mathbb F_w}^+ : im(I) \mapsto I$, such that $g_{h,\mathbb F_w}^+(g_{h,\mathbb F_w}(z_0)) = z_0, \forall z_0 \in I$.

Part (b): For any $\bx \in \mathbb R^D$, we have,
\begin{equation*}
    \begin{aligned}
p(\bx) = E (Y | X = \bx) = \mathbb E_W \left(\Phi(m(\bx) + w)\right) &= \int_{- \infty}^\infty \Phi(m(\bx) + w)d \mathbb F w \\
&= \int_{- \infty}^\infty \int_{-\infty}^{m(\bx) + w} \phi(z) dz d \mathbb F w \\
&= Pr(Z-w\leq m(\bx)) \mbox{ where } Z \sim N(0,1) \independent w\\
&= \Phi\left(\frac{m(\bx)}{\sqrt{1+ \sigma^2}}\right). 
\end{aligned}
\end{equation*}
 
We then have the following expression of $g_{h,\mathbb F_w}^{+}(\cdot)$ for probit link:
\begin{equation*}
    m(\bx)  = g_{h,\mathbb F_w}^{+}(p(\bx)) 
    = \left( 1 + \sigma^2\right)^{\frac{1}{2}} \Phi^{-1} \left( p(\bx) \right).
\end{equation*}
\end{proof}

\subsection{Proofs for results of Section \ref{sec:outline}}\label{sec:proofsoutline}

\begin{proof}[Proof of Lemma \ref{lemma:DART-theoretical}]
We note that the split criterion (\ref{eqn:DART}) is a function of expressions of the form $\bZ^\top \bQ \bY$, $\bZ^\top \bQ \bZ$, $\bY^\top \bQ \bY$. Each entry of all these expressions is a quadratic form in $\bQ$. 

For any two $n \times 1$ vectors $\bu$ and $\bv$, under Assumption \ref{item:assumption_chol}, we can write a quadratic form of $\bQ$ as \begin{align}\label{eq:qf}
\bu^\top\bQ\bv 
&= \gamma \sum_i u_iv_i + \sum_{j \neq j' = 0}^q \rho_j\rho_{j'} \sum_{i} u_{i-j}v_{i-j'}+\sum_{i \in \mathcal{A}}\sum_{i' \in \mathcal{A}'} \tilde\gamma_{i,i'}u_iv_{i'},
\end{align}
where $\gamma=\|\rho\|_2^2$, ${\mathcal{A}}, {\mathcal{A}}' \subset \{1,2,\cdots,n \}$ with $|{\mathcal{A}}|, |{\mathcal{A}}'| \leq 2q$, $\tilde \gamma_{i,i'}$'s are fixed (independent of $n$) functions of $\bL$ and $\rho$ in \eqref{eqn:assumption_chol}, and $u_i$, $v_i$ for $i \leq 0$ are defined to be zero. \pstar

First, we focus on the asymptotic limit of the node representative of the leaf nodes given by (\ref{eqn:DART_mean}). We show that alike the continuous case, in the binary case too, the GLS-loss-based node representative converges to the mean of the response, conditional on the covariate belonging to the specific leaf node. We will show that for a data-independent, a pre-fixed tree with leaf nodes denoted by $\calB_1, \calB_2,  \ldots, \calB_K$, we have 
\begin{equation}
\label{eqn:lim_beta}
    {\hat\beta}_{{l}} \overset{a.s.}{\to} \mathbb{E}(Y | X \in \mathcal{B}_{{l}}) \text{ as $n \to \infty$} \mbox{ for } l=1,\ldots,K.
\end{equation}
where, ${\hat\beta}_{{l}}$ corresponds to the estimate (\ref{eqn:DART_mean}) corresponding to node $\calB_l, l = 1, 2, \cdots, K$. 

As ${\hat\beta} = \left( {Z}^\top\bQ {Z}\right)^{-1} {Z}^\top\bQ {Y}$, we consider the asymptotic limits of $\left(\frac{1}{n}{Z}^\top\bQ {Z}\right)^{-1}$ and $\frac{1}{n}{Z}^\top\bQ {Y}$ separately. 
 We first consider the asymptotic limit of $\frac{1}{n}{Z}^\top\bQ {Y}$. 
Following \ref{eq:qf} we have: \pstar

\begin{equation}\label{eq:cor}
	\begin{aligned}
	\frac{1}{n}({Z}^\top\bQ{Y})_l &= \frac{1}{n}\left[\gamma\sum_i{Z}_{i,l}Y_{i} + \sum_{j \neq j' = 0}^q \rho_j\rho_{j'} \sum_i{Z}_{i-j,l} Y_{i-j'} + O_b(1) \right]\\
	\end{aligned}
	\end{equation}
where $O_b(1)$ denotes a sequence of random variables which are uniformly bounded.

\begin{enumerate}
    \item \textbf{Asymptotic limit of $\frac{1}{n}\sum_i{Z}_{i,l}Y_{i}$ :}    
    We use the Uniform Law of Large Number (ULLN) for $\beta$-mixing process in Theorem 1 from \cite{nobel1993note}, which shows that if a class of functions with uniformly integrable envelop satisfies uniform law of averages with respect to an i.i.d. process, then it also satisfies the uniform law of averages for stationary, $\beta$-mixing stochastic process with the identical marginal distribution. {\color{black} As $(Y_i,X_i)$ is stationary ergodic and $\beta$-mixing, we consider the class of function with a singleton function $(Y_i, X_i) \mapsto {Z}_{i,l}Y_{i}$. Since ${Z}_{i,l}Y_{i} \leqslant {Z}_{i,l} \leq 1$}, which is uniformly integrable, we can apply Theorem 1 from \cite{nobel1993note} and have that strong law holds for ${Z}_{i,l}Y_{i}$, i.e.
    
    \begin{equation}
    \label{eqn:limit_beta_diagonal}
        \lim_{n \to \infty}\frac{1}{n}{Z}_{i,l}Y_{i} \overset{a.s.}{=} \mathbb{E}\left({Z}_{1,l}Y_1 \right) =  \mathbb{E}\left(Y | X \in \mathcal{B}_l\right)Vol(\mathcal{B}_l)
    \end{equation}

    \item \textbf{Asymptotic limit of $\frac{1}{n} \sum_i{Z}_{i-j,l} Y_{i-j'}$ :}  \blue{We can use a very similar ULLN as above using the sequence $\{Y_i,Z_{i-(j-j')}\}_i$ which is also stationary, ergodic and $\beta$-mixing under Assumption \ref{as:data_gen}. We then have}
    
        \begin{equation}
            \label{eqn:limit_beta_offdiagonal}
        \lim_{n \to \infty}\frac{1}{n} \sum_i{Z}_{i-j,l} Y_{i-j'} \overset{a.s.}{=} \mathbb{E}\left({Z}_{2q - j ,l}Y_{2q - j'}  \right) =  \mathbb{E}\left(Y\right)Vol(\mathcal{B}_l)
    \end{equation}
    where, the last equality follows from the independence of ${Z}_{i-j,l}$ and $ Y_{i-j'}$ for $j \neq j'$.
\end{enumerate}

Combining \eqref{eqn:limit_beta_diagonal} and \eqref{eqn:limit_beta_offdiagonal}, we have:

    \begin{equation}
    \label{eqn:lim_ZtQy}
        \lim_{n \to \infty}\frac{1}{n}({Z}^\top\bQ{Y})_l \overset{a.s.}{=} \gamma \mathbb{E}(Y\mathds{I}(X \in \mathcal{B}_l))+Vol(\mathcal{B}_l) \mathbb{E}(Y)\sum_{j \neq j' = 0}^q \rho_j\rho_{j'}.
    \end{equation}

This is the same limit as that for continuous $Y$ in \cite{saha2023random}. For prefixed data-independent partition, $\left( {Z}^\top\bQ {Z}\right)^{-1}$ is not affected by distribution of the outcome $Y$. Hence, this quantity also asymptotes to the same limit as in the continuous case. 

\begin{equation}
\label{eqn:lim_ZtQZ}
    	\begin{aligned}
	\left(\lim_{n \to \infty}\frac{1}{n}{Z}^\top\bQ {Z}\right)^{-1}
	&= \gamma ^{-1} \left[diag({b}^{-1}) -\frac{\left(\sum_{j \neq j' = 0}^q \rho_j\rho_{j'}\right)}{\gamma + \sum_{j \neq j' = 0}^q \rho_j\rho_{j'}}{1}{1}^\top \right].
	\end{aligned}
\end{equation}
where ${b}$ is the diagonal of the volume of the leaf nodes.

Multiplying \eqref{eqn:lim_ZtQZ} and  \eqref{eqn:lim_ZtQy} and simplifying, we have \eqref{eqn:lim_beta}. The GLS cost function corresponding to membership $\bZ$ is $\left({Y} - {Z}{\hat{\beta}}( Z) \right)^\top \bQ\left({Y} - {Z}{\hat{\beta}}( Z) \right)
	= {Y}^\top\bQ{Y} -  {Y}^\top\bQ{Z} {\hat{\beta}}( Z)$. Since $v_{n, Q}^{DART}((d,c))$ in \eqref{eqn:DART} is a difference of cost function corresponding to the membership matrix $\bZ$ before and after the split, the ${Y}^\top\bQ{Y}$ component cancels out and it is enough to focus on asymptotic limit of differences of ${Y}^\top\bQ{Z} {\hat{\beta}}( Z)$. Since we have already shown the asymptotic limits of $\frac{1}{n}({Z}^\top\bQ{Y})_l$ and ${\hat\beta}$ to be identical for both the continuous and binary outcome, following \cite{saha2023random}, we have that the asymptotic limit of $v_{n, Q}^{DART}((d,c))$ in  \eqref{eqn:DART} will coincide for both binary and continuous outcome, i.e. $\lim_{n \to \infty} v_{n, Q}^{DART}((d,c)) \overset{a.s.}{=}	\Delta_{\bQ}^*((d,c)) $.\\
\end{proof}

\begin{proof}[\blue{Proof of Lemma \ref{lemma:noves_mean_noseparable}}]{\color{black}
From Assumption \ref{as:data_gen}, 
we have $p(X_i) := \mathbb E(Y_i| X_i) = g(m(X_i))$. As $m$ is continuous, and image of $g$ lies in $[0,1]$, $p$ }is continuous by the dominated convergence theorem. 

From Lemma \ref{lemma:DART-theoretical} we have 
$$
\begin{aligned}
\Delta_{\bQ}^*((d,c))=& \gamma(\bQ) Vol(\mathcal A) \Big[\mathbb{V}(Y | {X} \in \mathcal A) - \mathbb{P}({X} \in \mathcal A^R | {X} \in \mathcal A)\mathbb{V}(Y | {X} \in \mathcal A^R)\\
	&- \mathbb{P}({X} \in \mathcal A^L |  {X} \in \mathcal A)\mathbb{V}(Y | {X} \in \mathcal A^L) \Big]
\end{aligned}
$$

We begin with the simple case with $D = 1$ covariate. In this scenario, $d$ is always $1$ and we are only optimizing $\Delta_{\bQ}^*((d,c))$ over $c$. So any node $\mathcal A$ will be an interval by construction and without loss of generality we assume $\mathcal A = [\gamma_1, \gamma_2]$. \blue{Writing $\gamma=\gamma(Q)$ we have}

\begin{equation}\label{eq:theoretical_split}
\begin{aligned}
&\Delta^*((1,c))\left(\gamma Vol(\mathcal A)\right)^{-1}\\ &=   \mathbb{V}(Y | {X} \in \mathcal A) - \mathbb{P}({X} \in \mathcal A^R | {X} \in \mathcal A)\mathbb{V}(Y | {X} \in \mathcal A^R) - \mathbb{P}({X} \in \mathcal A^L |  {X} \in \mathcal A)\mathbb{V}(Y | {X} \in \mathcal A^L) \\
	& = \mathbb{E}(Y^2 | {X} \in \mathcal A) - \mathbb{E}(Y | {X} \in \mathcal A)^2 - \mathbb{P}({X} \in \mathcal A^R | {X} \in \mathcal A) \left( \mathbb{E}(Y^2 | {X} \in \mathcal A^R) - \mathbb{E}(Y | {X} \in \mathcal A^R)^2\right)\\
	&- \mathbb{P}({X} \in \mathcal A^L | {X} \in \mathcal A) \left( \mathbb{E}(Y^2 | {X} \in \mathcal A^L) - \mathbb{E}(Y | {X} \in \mathcal A^L)^2\right)\\
	&- \mathbb{P}({X} \in \mathcal A^L | {X} \in \mathcal A) \left( \mathbb{E}(Y | {X} \in \mathcal A^L) - \mathbb{E}(Y | {X} \in \mathcal A^L)^2\right)\\
	&= - \mathbb{E}(Y | {X} \in \mathcal A)^2  + \mathbb{P}({X} \in \mathcal A^R | {X} \in \mathcal A) \mathbb{E}(Y | {X} \in \mathcal A^R)^2 + \mathbb{P}({X} \in \mathcal A^L | {X} \in \mathcal A) \mathbb{E}(Y | {X} \in \mathcal A^L)^2\\
	& = - \frac{1}{\left(\gamma_2 - \gamma_1\right)^2} \left( \int_{\gamma_1}^{\gamma_2} p(t)\partial t\right)^2 + \frac{\left(c - \gamma_1\right)}{\left(\gamma_2 - \gamma_1\right)} \frac{1}{\left(c - \gamma_1\right)^2} \left( \int_{\gamma_1}^{c} p(t)\partial t\right)^2\\
	& \quad + \frac{\left(\gamma_2 - c\right)}{\left(\gamma_2 - \gamma_1\right)} \frac{1}{\left(\gamma_2 - c\right)^2} \left( \int_{c}^{\gamma_2} p(t)\partial t\right)^2.
\end{aligned}
\end{equation}

Following \cite{scornet2015consistency}, we have that

\begin{equation}\label{eq:th_split_2}
\Delta_{\bQ}^*((1,c))\left(\lambda Vol(\mathcal A)\right)^{-1} = \frac{1}{\left(c - \gamma_1 \right)\left(\gamma_2 - c \right)} \left(H(c) - J \frac{c - \gamma_1}{\gamma_2 - \gamma_1} \right)^2
\end{equation}

where, $\lambda:=\gamma(\bQ)$, $H(c) = \int_{\gamma_1}^c  p(t)\partial t$ and $J = \int_{\gamma_1}^{\gamma_2} p(t)\partial t$. 

Since $\lambda, Vol(\mathcal A)$ are constants and $\Delta^*((1,c)) = 0;\:\: \forall c \in [\gamma_1, \gamma_2]$, we have that $H(c) =  J \frac{c - \gamma_1}{\gamma_2 - \gamma_1}$. This proves the linearity of $H(\cdot)$ on $[\gamma_1, \gamma_2]$, which, combined with continuity of $p(t)$ implies that $p(t)$ is constant on $[\gamma_1, \gamma_2]$.

To extend the aforementioned $1$-dimensional results for general multi-dimensional covariate case, the proof deviates from the continuous response setting of \cite{scornet2015consistency}. 
For continuous (Gaussian) $Y$, under assumption \ref{item:assumption_m} of additivity of $m(\cdot)$, i.e. $m(\bx) = \sum_{d=1}^D m_d (x_d)$, due to the linear nature of the integrand {\color{black}($g(\bx) = \bx$) we have $ \mathbb E \left(Y | X\right) = m(X)$.} Thus $\mathbb{E}(Y|X)$ can be expressed as a sum of $m_d(x_d)$ and the term $\sum_{l \neq d} m_l(x_l)$ which is a constant w.r.t $x_d$. This is convenient as we see from (\ref{eq:theoretical_split})  that the theoretical split only depends on the conditional means in the parent and children nodes. Due to additivity, the conditional mean will be integrals of only the component $m_d(x_d)$ plus some constant. To elucidate, for $D > 1$ and choosing, without loss of generality, $d=1$, one can obtain
\begin{equation*}
\psi(t) = \int_{x_2,\ldots,x_D} \mathbb{E}(Y|X=(t,x_2,\ldots,x_D)) dx_2\ldots dx_D = m_1(t) + \mbox{constant}.
\end{equation*}
This reduces $\Delta_{\bQ}^*((1,c))$ to the same form as (\ref{eq:th_split_2}) with the function $\psi(t)$ replacing $p(t)$ and proves that the one dimensional $\phi(t)$ is constant implying that $m_1(t)$ is a constant.

In our scenario for binary data, this reduction of the problem from the multi-dimensional case to the established one-dimensional case is not straightforward due to non-additivity. The link function $g(\cdot)$ is usually non-linear for binary outcomes. Additivity of $m(\cdot)$, i.e. $m(\bx) = \sum_{d=1}^D m_d (x_d)$ in Assumption \ref{item:assumption_m}, does not translate to separability structure of $p(\cdot)$. This poses a major roadblock in further development. 

We present a novel proof, that relaxes the assumption on additivity of $\mathbb E (Y| X)$. 
For a general multivariate case, we consider a node $\mathcal A = \prod_{d = 1}^D [\gamma_1^{(d)}, \gamma_2^{(d)}]$. 
Since $m(\cdot)$ is additive by Assumption \ref{item:assumption_m}, it is enough to prove that $m_d(\cdot)$ is constant on $[\gamma_1^{(d)}, \gamma_2^{(d)}]$ for $d \in \{1,2,\cdots, D \}$. We prove this by contradiction. Let us assume, that $\exists d \in \{1,2,\cdots, D \}$, and corresponding $c^{(1)}_d$ and $c^{(2)}_d$ such that $m_d(c_d^{(1)}) \neq m_d(c_d^{(2)})$. Without loss of generality, $m_d(c_d^{(1)}) < m_d(c_d^{(2)})$. Now, using the same steps as equations (\ref{eq:theoretical_split}) and (\ref{eq:th_split_2}) we have, 

\begin{subequations}
\begin{equation}
\label{eqn:constant}
  \psi_d(x_d)=   \int_{\gamma_1^{(1)}}^{\gamma_2^{(1)}}\ldots \int_{\gamma_{1}^{(d - 1)}}^{\gamma_{2}^{(d - 1)}}\int_{\gamma_{1}^{(d + 1)}}^{\gamma_{2}^{(d + 1)}}\ldots \int_{\gamma_{1}^{(D)}}^{\gamma_{2}^{(D)}} p(\bx) \partial x_1\ldots \partial x_{d-1}\partial x_{d+1}\partial x_{D} = J_d
\end{equation}
where $J_d$ is a constant. 

From the definition of Riemann integrability, we can approximate the aforementioned integral through a Riemann sum as follows: 
We divide each of $[\gamma_1^{(l)}, \gamma_2^{(l)}]; l \neq d$ into equal $K$ parts and consider the Riemann sum corresponding to the obtained $K^{D - 1}$ partitions. For any $\varepsilon > 0$, for $l = 1, 2 $, $\exists N_l \in \mathbb N$, such that for all $K > N_l$ we have,
\begin{equation}
    \label{eqn:Riemann_sum}
    \begin{aligned}
    \bigg| \psi_d(c_d^{(l)}) - \sum_{j^{(1)} = 1}^K\ldots\sum_{j^{(d - 1)} = 1}^K\sum_{j^{(d + 1)} = 1}^K\ldots \sum_{j^{(D)} = 1}^K p(\bx^{(l,d)}_{j^{(1)}\ldots j^{(d - 1)}j^{(d + 1)}\ldots j^{(D)}}) \prod_{j \neq d} (\gamma_2^{(j)} - \gamma_1^{(j)})/n^{D - 1}\bigg| < \eps,
    \end{aligned}
\end{equation}
where 
$\bx^{(l,d)}_{j^{(1)}\ldots j^{(d - 1)}j^{(d + 1)}\ldots j^{(D)}} = $ \\ $ (\alpha_1^{j^{(1)}}, \ldots, \alpha_{d - 1}^{j^{(d - 1)}}, c_d^{(l)}, \alpha_{d+1}^{j^{(d + 1)}}, \ldots, \alpha_D^{j^{(D)}})$, with $\alpha_q^{j^{(q)}} = \gamma_1^{(q)} + (\gamma_2^{(q)} - \gamma_1^{(q)})j^{(q)}/n$, $\forall q = 1,\ldots,d-1, d+1, \ldots, D $ and $\forall j^{(q)} = 1, 2, \cdots, K$.

From \eqref{eqn:constant}, we have that $\psi_d(c_d^{(1)}=\psi_d(c_d^{(2)} = J_d$. 
Hence, from \eqref{eqn:Riemann_sum}, for $K > max(N_1, N_2)$, we have:
\begin{equation}
    \label{eqn:Riemann_sum_difference}
    \begin{aligned}
        2\varepsilon \geqslant \prod_{j \neq d} (\gamma_2^{(j)} - \gamma_1^{(j)})/K^{D - 1} \bigg|&\sum_{j^{(1)} = 1}^K\ldots\sum_{j^{(d - 1)} = 1}^K\sum_{j^{(d + 1)} = 1}^K\ldots \sum_{j^{(D)} = 1}^K p(\bx^{(2,d)}_{j^{(1)}\ldots j^{(d - 1)}j^{(d + 1)}\ldots j^{(D)}})\\  
        &- \sum_{j^{(1)} = 1}^K\ldots\sum_{j^{(d - 1)} = 1}^K\sum_{j^{(d + 1)} = 1}^K\ldots \sum_{j^{(D)} = 1}^K p(\bx^{(1,d)}_{j^{(1)}\ldots j^{(d - 1)}j^{(d + 1)}\ldots j^{(D)}}) \bigg|  
    \end{aligned}
\end{equation}

Note that {\color{black}$p(x)=g(m(x))$ and $g$ is differentiable} (Assumption \ref{item:assumption_differentiability}). Hence, for any choice of  $(j^{(1)},\ldots, j^{(d - 1)},j^{(d + 1)},\ldots, j^{(D)}) \in \{ 1,2,\cdots, K\}^{D-1}$, using Mean Value Theorem along with additivity of $m(\cdot)$, we have,

\begin{equation}
    \label{eqn:MVT}
\begin{aligned}
        &p(\bx^{(2,d)}_{j^{(1)}\ldots j^{(d - 1)}j^{(d + 1)}\ldots j^{(D)}}) - p(\bx^{(1,d)}_{j^{(1)}\ldots j^{(d - 1)}j^{(d + 1)}\ldots j^{(D)}})\\
        &={\color{black}g}(m(\bx^{(2,d)}_{j^{(1)}\ldots j^{(d - 1)}j^{(d + 1)}\ldots j^{(D)}})) - {\color{black}g}(m(\bx^{(1,d)}_{j^{(1)}\ldots j^{(d - 1)}j^{(d + 1)}\ldots j^{(D)}}))\\
    &=  {\color{black}g}'(m^{(0)}_{j^{(1)}\ldots j^{(d - 1)}j^{(d + 1)}\ldots j^{(D)}}, \mathbb F_W) \left(m(\bx^{(2,d)}_{j^{(1)}\ldots j^{(d - 1)}j^{(d + 1)}\ldots j^{(D)}}) - m(\bx^{(1,d)}_{j^{(1)}\ldots j^{(d - 1)}j^{(d + 1)}\ldots j^{(D)}}) \right)\\
    & = {\color{black}g}' ( m^{(0)}_{j^{(1)}\ldots j^{(d - 1)}j^{(d + 1)}\ldots j^{(D)}}) \left(m_d(c_d^{(2)}) - m_d(c_d^{(1)})\right)\\
    & \geqslant M' \left(m_d(c_d^{(2)}) - m_d(c_d^{(1)})\right)
\end{aligned}
\end{equation} 
where, $m^{(0)}_{j^{(1)}\ldots j^{(d - 1)}j^{(d + 1)}\ldots j^{(D)}}$ is a point between $m(\bx^{(1,d)}_{j^{(1)}\ldots j^{(d - 1)}j^{(d + 1)}\ldots j^{(D)}})$ and $m(\bx^{(1,d)}_{j^{(1)}\ldots j^{(d - 1)}j^{(d + 1)}\ldots j^{(D)}})$, $M'$ is an uniform positive lower bound on the derivative ${\color{black}g}'$. Existence of $M'$ is guaranteed from combining the fact that ${\color{black}g}$ is strictly increasing, hence ${\color{black}g}' > 0$, and ${\color{black}g}'$ is continuous (Assumption \ref{item:assumption_differentiability}) on $m([0,1]^D) \subseteq [-M,M]$. 

Since we assumed, $m_d(c_d^{(1)}) < m_d(c_d^{(2)})$, the additivity of $m(\cdot)$ along with the strictly increasing nature of $h$ implies $p(\bx^{(1,d)}_{j^{(1)}\ldots j^{(d - 1)}j^{(d + 1)}\ldots j^{(D)}}) < p(\bx^{(2,d)}_{j^{(1)}\ldots j^{(d - 1)}j^{(d + 1)}\ldots j^{(D)}})$. Hence, using \eqref{eqn:MVT}, \eqref{eqn:Riemann_sum_difference} can be simplified as

\begin{equation}
    \label{eqn:lower_bound}
    \begin{aligned}
        2\varepsilon \geqslant &\prod_{j \neq d} (\gamma_2^{(j)} - \gamma_1^{(j)})/K^{D - 1}  \sum_{j^{(1)} = 1}^K\ldots\sum_{j^{(d - 1)} = 1}^K\sum_{j^{(d + 1)} = 1}^K\ldots \sum_{j^{(D)} = 1}^K \bigg(p(\bx^{(2,d)}_{j^{(1)}\ldots j^{(d - 1)}j^{(d + 1)}\ldots j^{(D)}})\\ 
        &\:\:\:\:\:\:\:\:\:\:\:\:\:\:\:\:\:\:\:\:\:\:\:\:\:\:\:\:\:\:\:\:\:\:\:\:\:\:\:\:\:\:\:\:\:\:\:\:\:\:\:\:\:\:\:\:\:\:\:\:\:\:\:\:\:\:\:\:\:\:\:\:\:\:\:\:\:\:\:\:\:\:\:\:\:\:\:\:\:\:\:\:\:\:\:\:\:\:\:\:\:\:\:\:\:\:\:\:\:\:\:\:- p(\bx^{(1,d)}_{j^{(1)}\ldots j^{(d - 1)}j^{(d + 1)}\ldots j^{(D)}})\bigg)\\
        &\geqslant \prod_{j \neq d} (\gamma_2^{(j)} - \gamma_1^{(j)})/K^{D - 1}  \sum_{j^{(1)} = 1}^K\ldots\sum_{j^{(d - 1)} = 1}^K\sum_{j^{(d + 1)} = 1}^K\ldots \sum_{j^{(D)} = 1}^K M' \left(m_d(c_d^{(2)}) - m_d(c_d^{(1)})\right) \\
        & = M'(m_d(c_d^{(2)}) - m_d(c_d^{(1)}))\prod_{j \neq d} (\gamma_2^{(j)} - \gamma_1^{(j)})
\end{aligned}
\end{equation}

Since we assumed $m_d(c_d^{(1)}) < m_d(c_d^{(2)})$ and \eqref{eqn:lower_bound} holds for arbitrary $\varepsilon > 0$, choosing small enough $\varepsilon$, we have a contradiction. \\
\end{subequations}

\end{proof}

{\color{black}\begin{proof}[Proof of Lemma \ref{lemma:joint_mixing}]  
Let $A_i = (X_i,X_{i-l})$ for any fixed $l \in \mathbb Z$. From Assumption \ref{item:assumption_x}, $\{ X_{i-l}\}$ is an i.i.d. process. Then $A_i \independent A_{i+l+1}$ for all $i$, i.e., $A_i$ is an $m$-dependent process which implies it is also $\beta$-mixing \citep{bradley2005basic}. As 
 Assumption \ref{item:assumption_w_new} dictates that $\{ w_i\}$ is also a $\beta$-mixing ergodic process, independent of  $\{ X_i\}$. 
 $\{w_i,A_{i} \} = \{w_i,X_i,X_{i-l} \}$ is a $\beta$-mixing ergodic process \citep[Section 5,][]{bradley2005basic}. Next, we consider the following latent variable representation common for binary data \citep{albert1993bayesian}
\begin{equation}
    \label{eqn:new_formulation}
    Y_i = 
     \begin{cases}
       1 &\quad U_i \le h(m(X_i) + w_i)\\
       0 &\quad\text{otherwise,} \\ 
     \end{cases}, \qquad U_i \iid U[0,1], \qquad U \independent{\left\{X, w\right\}}.
\end{equation}
By construction in \eqref{eqn:new_formulation}, the i.i.d. process $\{ U_i\}$ is independent of the $\beta$-mixing process $\{w_i,X_i,X_{i-l}\}$.  
Hence, a second application of the aforementioned result proves that $\{U_i, W_i, X_i, X_{i-l}\}$ is a $\beta$-mixing ergodic process. Next we observe that if $\{ H_i\}$ is a $\beta-$ mixing ergodic process, so is $\{ f \left(H_i\right)\}$ as long as $E|f(H_i)|$ exists, since the $\sigma-$algebra generated by $f \left(H_i\right)$ is contained in the $\sigma-$algebra generated by $H_i$. As $Y_i = \mathds 1 \left( U_i \geq h \left(m \left( X_i\right) + w_i \right)\right)$, this implies that $(Y_i,X_{i-l})$ is a function of $\{U_i, W_i, X_i, X_{i-l}\}$ and thus $\left\{ (Y_i,X_{i-l})\right\}$ is a $\beta$-mixing ergodic process.\\ 
\end{proof}}

\begin{proof}[Proof of Lemma \ref{lemma:bounded_estimate}]
Let $\bp_n$ denote the estimate (\ref{eqn:DART_mean}) of the node representatives. This quantity has the same functional form as in the continuous case. Hence from Equation (S14) of \cite{saha2023random}, we have the following \blue{under Assumption \ref{as:working_cov}}:

\begin{equation}\label{eq:bound}
\begin{aligned}
\| \bp_n\|_{\infty} &\leqslant \frac{\left(\gamma + \sum_{j \neq j' = 0}^q |\rho_j\rho_{j'}| +\sum_{i \in {\mathcal{A}}_1}\sum_{i' \in {\mathcal{A}}_2} |\tilde\gamma_{i,i'}| \right)}{\xi}\max_{i} |y_i|,
\end{aligned}
\end{equation}

where, $\xi = \min_{i=1,\ldots,q+1} (\bQ_{ii} - \sum_{j \neq i} |\bQ_{ij}|) > 0$.  
As $|y_i| \leqslant 1$, from \eqref{eq:bound},  we have that $p_n(\cdot)$ is bounded. \pstar
\end{proof}

\subsection{\blue{Detailed outline of} proof of consistency of \blue{the mean function estimate}}
\label{sec:main_consistency}
In this Section, we provide the proofs of the main results, Theorems \ref{th:main_gh} and \ref{th:main_m} and Corollary \ref{prop:spatias_application}. Least square estimators try to minimize the $\mathbb L_2$ risk. Since the true distribution of the data is unknown, the theoretical $\mathbb L_2$ risk is estimated through empirical $\mathbb L_2$ risk. In order to produce consistent estimators, first we choose a class of function $\mathcal F_n$ and then select the function belonging to that class of function, which minimizes the empirical risk. This class is often data-driven and increases as the sample size increases (method of sieves, \cite{grenander1981abstract}). In this estimation process, two kinds of errors contribute to the $\mathbb L_2$ risk:

\begin{enumerate}
    \item \textbf{Approximation error}: This measures how well the best estimator from the selected class of functions ($\mathcal F_n$) approximates the true function $m(\cdot)$.
    \item \textbf{Estimation error}: This measures the distance of the proposed estimator from the best estimator in $\mathcal F_n$. This is usually controlled by controlling the maximum deviation of the empirical $\mathbb L_2$ risk from the theoretical $\mathbb L_2$ risk for $\mathcal F_n$. 
\end{enumerate}

In order to ensure consistency, we need to control both kinds of error and ensure both vanish asymptotically, as sample size increases. This requires a careful choice of $\mathcal F_n$, as there is a tradeoff between these two types of errors. For a large $\mathcal F_n$, the Approximation error will be very small (as there it may very well contain $m(\cdot)$ or a nearby member), but might not uniformly control the Estimation error, as there are too many members in $\mathcal F_n$. Whereas, for a small $\mathcal F_n$, though it is easy to control the Estimation error, the class of function might be too small to control the approximation error.

This framework of consistency analysis of the least square-based approach was first introduced in \cite{nobel1996histogram} and later generalized in  \cite{gyorfi2006distribution}. \cite{scornet2015consistency} adopted the theory developed in \cite{gyorfi2006distribution} to establish consistency results for Breiman's RF for continuous outcomes with i.i.d. errors. In \cite{saha2023random}, the results of \cite{gyorfi2006distribution} were further generalized to accommodate GLS estimators which minimize a quadratic form loss. They also relaxed the i.i.d Assumption for continuous data and established consistency for  absolutely regular ($\beta$-mixing) dependent stochastic processes. In both \cite{scornet2015consistency} and \cite{saha2023random}, the structure of the data generation procedure, the continuous outcome $Y$ was modeled as an additive function of the covariate effect and the error process. This played a pivotal role \blue{in controlling variation of the mean function in the leaf nodes of the tree.}. For binary outcome, due to the inherent dependency of $Y$ on a non-linear function of covariate effect and the dependent error, \blue{the mean of $Y$ is not additive in the components of $X$. In Lemma \ref{lemma:noves_mean_noseparable} we provided a novel proof for controlling leaf node variation that relaxes additivity of the mean.} We highlighted additional technical challenges in Section \ref{sec:outline} and provided their proofs in Section \ref{sec:proofsoutline}. Equipped with these additional pieces, in this section, we reconcile the proof of the consistency of the RF-GP approach for correlated binary outcomes with the proof for continuous outcomes in \cite{saha2023random}. 

\noindent\underline{\textbf{{Approximation error}\label{sec:approx}}}\newline

In order to control the approximation error, we need to:

\begin{enumerate}[label=\textbf{A.\roman*}]
    \item \textbf{Variation of the mean in leaf nodes}\label{item:mean_variation_node}: Show that the variation of the mean function vanishes within the leaf nodes, as the sample size increases. 
    
    \item \textbf{Mean approximation with $\mathcal{F}_n$}\label{item:fn_approximation}: Consider a suitable class of functions $\calF_n$, for which the approximation error can be controlled by controlling the variation of average within the leaf nodes. Since by \ref{item:mean_variation_node}, this variation becomes arbitrarily close to $0$ as sample size increases, the approximation error is controlled asymptotically.
\end{enumerate}

Establishing \ref{item:mean_variation_node} can be broken down into the following steps:

\begin{enumerate}[label=\textbf{A.i.\arabic*}]\addtocounter{enumi}{-1}
    \item \textbf{Theoretical split criterion}:\label{item:theoreticas_split} The very first step of the argument involves defining a theoretical version $v^{*}$ (that does not depend on $n$) of the empirical split criteria (dependent on $n$). Lemma \ref{lemma:DART-theoretical} provides this limit of the empirical split criterion (\ref{eqn:DART}) and shows that it is equivalent to the limit of the split criterion $v_n^{RT}$ in \eqref{eqn:CART} regression trees in Classical RF for continuous data. 
    
    \item \textbf{Variation of the mean function in theoretical tree}: \label{item:theoreticas_tree_m} If a {\em theoretical tree} is built based on the limiting theoretical split criteria (\ref{eqn:lim_DART}), the variation of the true mean function within the leaf nodes needs to vanish as the depth of the tree grows. The central argument for proving this is to show that is the theoretical split criterion is zero for all permissible splits of a cell, then the mean function is constant on the given cell. The theoretical split criterion (\ref{eqn:lim_DART}) for binary data has the same functional form (up to a constant) as that for the continuous case. However, due to the non-linear link $h$, the mean function for binary data is no longer additive in the components of $m$ even if $m$ is additive. Hence, the proof for this result differs from the continuous case and is provided in Lemma \ref{lemma:noves_mean_noseparable}.

    \item \textbf{Stochastic equicontinuity of empirical split criteria}:  \label{item:equicontinuity} For a suitable subset of splits, the empirical split criteria (\ref{eqn:DART}) needs to be stochastically equicontinuous with respect to the splits (collection of covariate choices and cutoff values used to sequentially build the tree). Stochastic equicontinuity postulates that as the sample size increases, the empirical split criteria for two splits can be arbitrarily close with high probability by choosing set of cutoff values for the respective splits to be close enough to each other. The split criteria is primarily a function of the quantities $\hat \beta(\bZ)$, $\bZ^\top\bQ\bY/n$, $\bZ^\top\bQ\bZ/n$, \pstar and $\bY^\top\bQ\bY/n$. 
    The proof of Lemma \ref{lemma:DART-theoretical} shows that (\ref{eqn:DART}) has the first two quantities has the same limit as in the continuous case, the third quantity does not depend on the outcome type and the last quantity is bounded in $n$ as $Y$ is binary. 
    These are sufficient to establish stochastic equicontinuity for the binary case.

    \item \textbf{Asymptotic equality of theoretical and empirical optimal splits}: \label{item:theoreticas_empiricas_cut} We consider the theoretical and empirical optimal splits obtained by optimizing the theoretical ($v^*$) and optimal split criteria ($v_n$). As the sample size increases, at any fixed depth of the tree, the theoretical and empirical optimal splits need to become arbitrarily close with high probability. The proof for this is the same is in the continuous case for i.i.d. data \citep{scornet2015consistency} as it  doesn't involve finite sample structure of the split criterion, but rather depends on its stochastic equicontinuity \ref{item:equicontinuity} of the split criterion. As our finite sample split criterion with binary outcome to have similar regularity, we have \ref{item:theoreticas_empiricas_cut} for RF-GLS with binary outcome. 
\end{enumerate}

Finally, for \ref{item:fn_approximation}, we can build a function class $\calF_n$ and a valid candidate member $f_0^{(n)}$ in  for estimating the mean function $p(\cdot)$  in the same way as in \cite{saha2023random} for continuous data, as the split criteria used are the same as are their limits under the two different data generation procedure (Lemma \ref{lemma:DART-theoretical}). \\

\noindent\underline{\textbf{{Truncation error}\label{sec:trunc}}}\\
\newline
The analysis of the estimator error of RF-GP for binary data benefits from the bounded nature of the response. Unlike the continuous case analysis, which accommodates unbounded errors, an additional level of complexity through truncation. First, both approximation error and estimation error are controlled under truncation  to establish $\mathbb L_2$ consistency. Next, the results are extended to the untruncated scenario by showing that the truncated estimator is asymptotically equal to the untruncated estimator by imposing assumptions on the tail decay of the error process.

In our scenario, due to the binary nature of the outcome, $Y$ is bounded. We show in Lemma \ref{lemma:bounded_estimate} that this leads to a uniformly (in $n$) bounded RF-GP estimate of $\mathbb E (Y|X)$. This simplifies the proof of consistency by mitigating the need for any truncation.   
Consequently, we do not also need any assumption on the tail behavior of the dependent error process and can work with improved scaling for $t_n$, i.e. the maximum number of nodes of a tree. Since the binary outcome is bounded by definition and RF-GP estimates are also bounded (Lemma \ref{lemma:bounded_estimate}), we can take $\beta_n$ in Theorem 3 of \citet{scornet2015consistency} to be bounded, which implies the improved scaling for $t_n$.\\

\noindent\underline{\textbf{{Estimation error}\label{sec:est}}}\\\newline
For tree estimators based on OLS loss, it suffices to control the squared error loss estimation error. When using a split criterion based on GLS loss (\ref{eqn:DART}), one needs to control estimation error corresponding to both squared error loss and cross-product errors arising from minimizing the quadratic forms \blue{in the GLS loss and estimator.} These terms respectively correspond to the diagonal and off-diagonal terms of the working precision matrix $\bQ$ and, under Assumption \ref{item:assumption_chol}, it is equivalent to proving the following \pstar

\begin{enumerate}[label=(\alph*)]
    \item \textbf{Squared error components:}
     \begin{equation}
        \label{eqn:dependent_squared_error_component}
        			 \lim_{n \to \infty}\mathbb{E} \left[ \sup_{f \in \mathcal{F}_n}  |\frac{1}{{n}}\sum_i|f(X_i) - Y_{i}|^2  - \mathbb{E}_{ \dot\calD_n}|f(\dot X_1) - \dot Y_{1}|^2|\right] = 0.
    \end{equation}
    \item \textbf{Cross-product components: } For all $1 \leqslant j \leqslant q$
        \begin{equation}
        \label{eqn:dependent_cross-product_component}
	\begin{aligned}
				 \lim_{n \to \infty}\mathbb{E}\Bigg[& \sup_{f \in \mathcal{F}_n}  \Big|\frac{1}{n}
				\sum_{i}
				(f( X_{i}) - Y_{i}) (f(X_{i-j}) - Y_{i-j})\\
				&- \mathbb{E}_{\dot \calD_n}(f(\dot X_{1+j}) - \dot Y_{1+j})(f(\dot X_{1}) - \dot Y_{{1}})\Big|\Bigg] = 0.
				\end{aligned}
    \end{equation}
\end{enumerate}
where,         {\color{black} $\dot X_i$, and $\dot Y_{i} \in \{ 0, 1 \}; \: \mathbb E ( \dot Y_{i}) =g(m(\dot X_{i}))$} be such that $\dot \calD_n:=\{(\dot X_i, \dot Y_i) | i=1,\ldots,n\}$ be identically distributed as $\calD_n$ and $\dot \calD_n \independent \calD_n$.\\

\noindent\underline{\textbf{{$\mathbb{L}_2$-consistency of RF-GP}}}\\

We present a general technical result on $\mathbb{L}_2$-consistency of a wide class of GLS estimates for binary outcome under $\beta$-mixing (absolutely regular) \blue{binary} processes. We will use this result to prove the consistency of the RF-GP estimate (Theorem \ref{th:main_gh}), but the result is not specific to forest and tree estimators and can also be used for general nonparametric GLS estimators of binary outcomes. 

We will consider an optimal estimator $p_n \in \calF_n$ with respect to quadratic loss:
\begin{equation}\label{eqn:p_n}
		p_n \in 
		 \arg \min_{f \in \mathcal{F}_n} \frac{1}{{n}}( f ( X) - Y)^\top \bQ( f ( X) - Y).
\end{equation}

\begin{theorem}\label{th:gyorfi}
Consider a stationary, binary $\beta$-mixing process from the data generating mechanism in \eqref{eqn:true_model}. Let $p(X)=\mathbb{E}(Y|X)$ and  $p_n(.,\Theta) : \mathbb{R}^D \to \mathbb{R}$ denote a GLS optimizer of the form \eqref{eqn:p_n} with the matrix $\bQ$ satisfying Assumption \ref{as:working_cov} \pstar and $\calF_n$ being a data-dependent function class. If $p_n$ and $\calF_n$ satisfies the following conditions:
	\begin{enumerate}[label=\textbf{C.\roman*}]
	\item \label{item:gyorfi_assumption_fn_bounded}\textbf{Bounded $\mathcal{F}_n$:}
$\mathcal{F}_n$ is uniformly bounded, i.e. there exists $B < \infty$, such that $B \geq \sup_{f \in \calF_n}$.
	    \item \label{item:gyorfi_assumption_approximation_error}\textbf{Approximation error:}
	    $
	    \lim_{n \to \infty} \mathbb E_{\Theta} \left[\inf_{f \in \mathcal{F}_n}  \mathbb{E}_{X} |f(X) -p(X) |^2 \right] =0.
	    $

\item \label{item:gyorfi_assumption_estimation_error}\textbf{Estimation error:} Conditions \eqref{eqn:dependent_squared_error_component} and \eqref{eqn:dependent_cross-product_component} hold. \\

Then we have 
$$
		\begin{aligned}
		\lim_{n \to \infty}\mathbb{E}\left[\mathbb{E}_{X} ( p_n( X, \Theta) - p( X))^2)\right]&=0, \mbox{ and }\\
		\lim_{n \to \infty} \mathbb{E}_{X} ( \bar{p}_n( X) - p( X))^2 &= 0; 
		\end{aligned}
		$$
		where $\bar{p}_n( X) =\mathbb{E}_{\Theta} p_n( X, \Theta)$ and $X$ is a new sample independent of the data.
			\end{enumerate}
\end{theorem}

\begin{proof}[Proof of Theorem \ref{th:gyorfi}]
The statement and proof of this theorem follow that of Theorem 5.1 in \cite{saha2023random} for continuous outcomes, adapted here to a bounded function and outcome. For the sake of completeness, we provide the proof in detail here. 

Under Assumption \ref{item:assumption_diagonas_dominance}, to show $\mathbb L_2$ consistency of $p_n(\cdot)$, \blue{following Eqn. (S13) of \cite{saha2023random},} it is enough to show that 
\begin{equation}
    \label{eqn:gyorfi_quadratic}
    \lim_{n \to \infty}\mathbb{E}\left[\mathbb{E}  \left[\rho^\top\left(p_n(\dot X^{(q+1)}, \Theta) - p(\dot X^{(q+1)})\right)\right]^2\right]=0, 
\end{equation}

where $\dot X^{(q+1)} := \left( \dot X_{(q+1)}, \dot X_{q},\cdot, \dot X_{1}\right)$. 
Here the inner expectation is with respect to the distribution of the new data $\dot D_n$ and the outer expectation is with respect to the original data $\calD_n$.

Focussing on the inner expectation, we have
$$
\begin{aligned}
&\mathbb{E} \left[\rho^\top \left( p_n(\dot X^{(q+1)}) - \dot Y^{(q+1)} \right) \right]^2\\ 
&= \mathbb{E} \left[\rho^\top \left( p_n(\dot X^{(q+1)}) -  p(\dot X^{(q+1)}) \right) + \rho^\top \left( p(\dot X^{(q+1)}) - \dot Y^{(q+1)} \right)  \right]^2\\
& = \mathbb{E} \left[\rho^\top \left( p_n(\dot X^{(q+1)}) -  p(\dot X^{(q+1)}) \right)\right]^2 + \mathbb{E} \left[\rho^\top \left( p(\dot X^{(q+1)}) - \dot Y^{(q+1)} \right) \right]^2\\ 
&+ 2\mathbb{E} \left[\rho^\top \left( p_n(\dot X^{(q+1)}) -  p(\dot X^{(q+1)}) \right)\rho^\top \left( p(\dot X^{(q+1)}) - \dot Y^{(q+1)} \right)  \right]\\
& = \mathbb{E} \left[\rho^\top \left( p_n(\dot X^{(q+1)}) -  p(\dot X^{(q+1)}) \right)\right]^2 + \mathbb{E} \left[\rho^\top \left( p(\dot X^{(q+1)}) - \dot Y^{(q+1)} \right) \right]^2.
\end{aligned}
$$

Hence, we have,
$$
\begin{aligned}
\mathbb{E} \left[\rho^\top \left( p_n(\dot X^{(q+1)}) - p(\dot X^{(q+1)} \right) \right]^2
&= \mathbb{E} \left[\rho^\top \left( p_n(\dot X^{(q+1)}) - \dot Y^{(q+1)} \right) \right]^2 - \mathbb{E} \left[\rho^\top \left(p(\dot X^{(q+1)}) -\dot Y^{(q+1)} \right) \right]^2\\
& =  A^2  + 2\left( \left(\mathbb{E} \left[\rho^\top \left(p(\dot X^{(q+1)}) -\dot Y^{(q+1)} \right) \right]^2\right)^{\frac{1}{2}} \right)A
\end{aligned}
$$
where, 
$$
A := \left(\mathbb{E} \left[\rho^\top \left( p_n(\dot X^{(q+1)}) - \dot Y^{(q+1)} \right) \right]^2 \right)^{\frac{1}{2}} - \left(\mathbb{E} \left[\rho^\top \left(p(\dot X^{(q+1)}) -\dot Y^{(q+1)} \right) \right]^2\right)^{\frac{1}{2}}  
$$
As $\left( \left(\mathbb{E} \left[\rho^\top \left(p(\dot X^{(q+1)}) -\dot Y^{(q+1)} \right) \right]^2\right)^{\frac{1}{2}} \right)$ is non-random (does not depend on $\calD_n$), for the outer expectation it is enough to show $\mathbb E A^2 \to 0$. Now,
$$
\begin{aligned}
\mathbb EA^2 &\leqslant  2\mathbb{E}\Bigg[ \left(\mathbb{E}\left[\rho^\top \left( p_n(\dot X^{(q+1)}) - \dot Y^{(q+1)} \right) \right]^2\right)^{\frac{1}{2}} -  \inf_{f \in \mathcal{F}_n} \left(\mathbb{E} \left[\rho^\top\left(f(\dot X^{(q+1)}) - \dot Y^{(q+1)} \right) \right]^2\right)^{\frac{1}{2}}\Bigg]^2 \\
&+2\mathbb{E}\Bigg[ \inf_{f \in \mathcal{F}_n} \left(\mathbb{E} \left[\rho^\top\left(f(\dot X^{(q+1)}) - \dot Y^{(q+1)} \right) \right]^2\right)^{\frac{1}{2}}
- \left(\mathbb{E} \left[\rho^\top \left(p(\dot X^{(q+1)}) - \dot Y^{(q+1)} \right) \right]^2\right)^{\frac{1}{2}}\Bigg]^2
\end{aligned}
$$

By triangular inequality the second quantity can be bounded above as follows:
$$
\begin{aligned}
&\mathbb{E}\Bigg[ \inf_{f \in \mathcal{F}_n} \left(\mathbb{E} \left[\rho^\top \left(f(\dot X^{(q+1)}) - \dot Y^{(q+1)} \right) \right]^2\right)^{\frac{1}{2}}- \left(\mathbb{E} \left[\rho^\top \left(p(\dot X^{(q+1)}) - \dot Y^{(q+1)} \right) \right]^2\right)^{\frac{1}{2}}\Bigg]^2\\
& \leqslant \mathbb{E} \inf_{f \in \mathcal{F}_n} \left(\mathbb{E} \left[\rho^\top \left(f(\dot X^{(q+1)}) - p(\dot X^{(q+1)}) \right) \right]^2\right) \\
& \leqslant    
\blue{\|\rho\|^2}\left[\inf_{f \in \mathcal{F}_n}  \mathbb{E}_{\dot X_1} |f(\dot X_1) -p(\dot X_1) |^2 \right]
\end{aligned}
$$
By \ref{item:gyorfi_assumption_approximation_error}, the last term is the approximation error with asymptotically goes to $0$. Hence, it is enough to show $\mathbb E A_1^2 \to 0$, where,  
$$
\begin{aligned}
A_1 &:= \left(\mathbb{E}\left[\rho^\top \left( p_n(\dot X^{(q+1)}) - \dot Y^{(q+1)} \right) \right]^2\right)^{\frac{1}{2}} -  \inf_{f \in \mathcal{F}_n} \left(\mathbb{E} \left[\rho^\top\left(f(\dot X^{(q+1)}) - \dot Y^{(q+1)} \right) \right]^2\right)^{\frac{1}{2}}\\
 &\leqslant \sup_{f \in \mathcal{F}_n} \Bigg\{ \left(\mathbb{E} \left[\rho^\top \left( p_n(\dot X^{(q+1)}) - \dot Y^{(q+1)} \right) \right]^2\right)^{\frac{1}{2}} - \left(\frac{1}{{n}} \sum_i \left[\rho^\top ( p_n( X^{(i)})- {Y}^{(i)})\right]^2 \right)^{\frac{1}{2}}\\
&+\left(\frac{1}{{n}} \sum_i \left[\rho^\top ( p_n( X^{(i)})- {Y}^{(i)})\right]^2 \right)^{\frac{1}{2}} - \left( \frac{1}{{n}}\left( p_n( X)- {Y}\right)^\top\bQ\left( p_n( X)- {Y}\right)\right)^{\frac{1}{2}}\\
&+\left( \frac{1}{{n}}\left( p_n( X)- {Y}\right)^\top\bQ\left( p_n( X)- {Y}\right)\right)^{\frac{1}{2}} - \left( \frac{1}{{n}}\left(f( X)- {Y}\right)^\top\bQ\left(f( X)- {Y}\right)\right)^{\frac{1}{2}}\\
&+ \left( \frac{1}{{n}}\left(f( X)- {Y}\right)^\top\bQ\left(f( X)- {Y}\right)\right)^{\frac{1}{2}} -\left(\frac{1}{{n}} \sum_i \left[\rho^\top ( f( X^{(i)})- {Y}^{(i)})\right]^2 \right)^{\frac{1}{2}}\\
&+\left(\frac{1}{{n}} \sum_i \left[\rho^\top ( f( X^{(i)})- {Y}^{(i)})\right]^2 \right)^{\frac{1}{2}}  - \left(\mathbb{E} \left[\rho^\top \left(f(\dot X^{(q+1)}) - \dot Y^{(q+1)} \right) \right]^2\right)^{\frac{1}{2}}\Bigg\}
\end{aligned}
$$

Here $\bX^{(i)} = (X_i, X_{i-1}, \ldots, X_{i-q})^\top $, $\bY^{(i)}$ is similarly defined. Let $b_t$ denote the $t^{th}$ term in above equation for $t = 1,2, \cdots, 5$. The $3^{rd}$ term $b_3$ is negative as $p_n$ by definition (\eqref{eqn:p_n}) is the minimizer of the GLS loss. Also, by triangular inequality
\begin{align*}
A_1 &\geqslant \left(\mathbb{E} \left[\rho^\top \left(p(\dot X^{(q+1)}) - \dot Y^{(q+1)} \right) \right]^2\right)^{\frac{1}{2}} -  \inf_{f \in \mathcal{F}_n} \left(\mathbb{E} \left[\rho^\top \left(f(\dot X^{(q+1)}) - \dot Y^{(q+1)} \right) \right]^2\right)^{\frac{1}{2}}\\
&\geqslant  - \inf_{f \in \mathcal{F}_n} \left(\mathbb{E} \left[\rho^\top \left(f(\dot X^{(q+1)}) - p(\dot X^{(q+1)}) \right) \right]^2\right)^{\frac{1}{2}}
\end{align*}

Denoting the right-hand side of the above equation by $a$,  we have $A_1^2 \leq  a^2 + 4  (b_1^2 + b_2^2 + b_4^2 + b_5^2)$. Hence, to show $\mathbb{E}(A_1^2)$ asymptotically vanishes it is enough to show the terms $\mathbb E(a^2)$ and $\mathbb E(b_t^2)$ for $t \neq 3$ vanish asymptotically. 

By the approximation error condition  \ref{item:gyorfi_assumption_approximation_error}, $\mathbb E (a^2)$ is asymptotically $0$. 

\pstar Using the structure of $\bQ$, $\mathbb E(b_1^2)$ and $\mathbb E(b_5^2)$ are bounded above by the following: 

\begin{align*}
&\mathbb{E} \Bigg[ \sup_{f \in \mathcal{F}_n}  \Bigg|\frac{1}{{n}}\sum_i \left[\rho^\top \left(f(X^{(i)})-Y^{(i)}\right)\right]^2  - \mathbb{E}\left[\rho^\top \left(f(\dot X^{(q+1)}) - \dot Y^{(q+1)} \right) \right]^2 \Bigg| \Bigg]\\
&\leqslant  \mathbb{E} \Bigg[ \sup_{f \in \mathcal{F}_n}  \Bigg| \gamma \left(\frac{1}{{n}}\sum_i(f(X_i) - Y_{i})^2  - \mathbb{E}(f(\dot X_1) - \dot Y_{1})^2 \right)  \\
&+ 2 \sum_{j=1}^q \sum_{j'\neq j}^q \rho_{j'}\rho_{j'-j} \Bigg( \frac{1}{n}
\sum_{i}
(f(X_{i}) - Y_{i}) (f(X_{i-j}) - Y_{i-j})\\
&- \mathbb{E}_(f(\dot X_{i}) - \dot Y_{i})(f(\dot X_{i-j}) - \dot Y_{{i-j}})\Bigg) \Bigg|\Bigg]
\end{align*}
By the estimation error condition \ref{item:gyorfi_assumption_estimation_error}, this goes to $0$ asymptotically.  

The terms $b_2^2$ and $b_4^2$ are differences between the full quadratic form of $\bQ$ with one that only considers row $q+1$ onwards. Hence, each of them only involves terms from the first $q$ rows of the Cholesky matrix.
Using Assumption \ref{item:assumption_chol} \pstar, this becomes

$$
\left(\sum_{1 \leq i,j \leq q} |(\bL^\top\bL)_{ij}| \right)\frac 1n \mathbb E \max_{1 \leqslant i \leq q} \sup_{f \in \calF_n} (f(X_i) - Y_i)^2\\
$$

Using uniform bound of $\mathcal F_n$ (Condition \ref{item:gyorfi_assumption_fn_bounded}) and \blue{that $Y$ is binary}, this converges to $0$ as $n \to \infty$. This completes the proof of Theorem \ref{th:gyorfi}.
\end{proof}

\subsection{Proof of Theorem \ref{th:main_gh}: \blue{Consistency of mean function estimate}}\label{sec:main_proof}
We will establish Theorem \ref{th:main_gh} with the help of Theorem \ref{th:gyorfi} by choosing a uniformly bounded function class $\calF_n$ that contains the estimator (\ref{eqn:DART_mean}) but simultaneously controls the approximation and estimation errors. 
We consider $\mathcal F_n$ to be a suitable subset of the class of all piece-wise constant functions on the data-driven obtained partitions. For a given partition $\mathcal{P}_n(\Theta)$, we define $\mathcal F_n$ as

\begin{equation}\label{eq:class}
\calF_n = \calF_n(\Theta)= \{p_n\} \cup \bigg\{\cup_{{x}_{l} \in \mathcal{B}_l \in  \mathcal{P}_n(\Theta)} \sum_{l = 1}^{t_n} \blue{p}(\bx_l)I(\bx \in \calB_l)\bigg\},
\end{equation}

where $\forall l = 1,2, \cdots, t_n, \calB_l$ is the subset of $[0, 1]^D$ corresponding to the $l^{th}$ leaf node and $\bx_l$ is any point in $\calB_l$. 

\begin{enumerate}
    \item Using Assumption \ref{item:assumption_m}, exploiting the continuity of $m_d, d = 1, 2, \cdots, D$, we have that  $\bigg\{\cup_{{x}_{l} \in \mathcal{B}_l \in  \mathcal{P}_n(\Theta)} \sum_{l = 1}^{t_n} \blue{p}(\bx_l)I(\bx \in \calB_l)\bigg\} $ is bounded. From Lemma \ref{lemma:bounded_estimate}, we have that $p_n$ is bounded. This implies that condition \ref{item:gyorfi_assumption_fn_bounded} of Theorem \ref{th:gyorfi} is satisfied.
    \item To show the Approximation Error \ref{item:gyorfi_assumption_approximation_error} vanishes, we have shown in Section \ref{sec:approx} that variation of the true function $p$ in the leaf nodes vanishes, hence following \cite{scornet2015consistency} a function class mitigates the approximation error as long as it contains \blue{functions of of the general form} $\sum_l \blue{p}(z_l)I(x \in \calB_l)$ where $z_l$ is some fixed member of $\calB_l$.

   \item  To show that \ref{item:gyorfi_assumption_estimation_error} in Theorem \ref{th:gyorfi} is satisfied in the present scenario, we first consider the i.i.d. counterpart of  in \eqref{eqn:dependent_squared_error_component}  and \eqref{eqn:dependent_cross-product_component} as follows:

\begin{enumerate}
    \item \textbf{Squared error component under independent process: }
			    \begin{equation}
        \label{eqn:independent_squared_error_component}
        		 \lim_{n \to \infty}\mathbb{E} \left[ \sup_{f \in \mathcal{F}_n}  |\frac{1}{{n}}\sum_i|f(X_i) - Y^*_{i}|^2  - \mathbb{E}_{ \dot\calD_n}|f(\dot X_1) - \dot Y_{1}|^2|\right] = 0 ,
    \end{equation}
    where {\color{black}$ Y_i^* \in \{ 0, 1 \}; \mathbb E \left( Y_i^*\right) = g(m( X_i)))$ and $\{Y_i^* \}$} denote an i.i.d. process, independent of $\calD_n$ and $\dot \calD_n$, such {\color{black}$\left(X_i, Y^*_i\right)$ is identically distributed as $\left(X_i, Y_i\right)$}.

	\item  \textbf{Cross-product component under independent process: } For all $1 \leqslant j \leqslant q$,

			    \begin{equation}
        \label{eqn:independent_cross-product_component}
				\lim_{n \to \infty} \mathbb{E}\Bigg[ \sup_{f \in \mathcal{F}_n}  \Big|\frac{1}{n}
				\sum_{i}
				(f( \tilde X_{i}) - \tilde Y_{i}) (f(\ddot X_{i-j}) - \ddot Y_{i-j}) 
				- \mathbb{E}_{\dot \calD_n}(f(\dot X_{1+j}) - \dot Y_{1+j})(f(\dot X_{1}) - \dot Y_{{1}})\Big|\Bigg] = 0,
    \end{equation}
    where, {\color{black}$\tilde Y_i \in \{ 0, 1 \}; \mathbb E \left(\tilde Y_i\right) = g(m(\tilde X_i)))$, $\ddot Y_i \in \{ 0, 1 \}; \mathbb E \left(\ddot Y_i\right) = g(m(\ddot X_i)))$ and for all $1 \leqslant j \leqslant q$, $\{(\tilde X_{i}, \ddot X_{i-j})\}_{i \geq j+1}$ and $\{(\tilde Y_{i}, \ddot Y_{i-j})\}_{i \geq j+1}$ are bivariate i.i.d. processes, independent of $\calD_n$ and $\dot \calD_n$, such that $(\tilde X_{i},\ddot X_{i-j},\tilde Y_{i},\ddot Y_{i-j})$  is be identically distributed as  $(X_i,X_{i-j},Y_i,Y_{i-j})$ for all $i$.} 
\end{enumerate}

In \cite{nobel1993note}, a general ULLN for $\beta$-mixing process was established, which showed that a class of functions $\mathcal F_n$ satisfies ULLN corresponding to an i.i.d counterpart of a $\beta$-mixing process $\{H_i\}$, $\mathcal F_n$ also satisfies ULLN for the true process $\{H_i\}$ if the class of functions $\mathcal F_n$ has a bounded envelope. 

As for binary data, we have shown in Lemma \ref{lemma:bounded_estimate} that the RF-GP estimate (\ref{eqn:DART_mean}) is uniformly bounded in $n$. Hence choosing a function class $\mathcal F_n$ which is bounded and for which the i.i.d. ULLN's \eqref{eqn:independent_squared_error_component} and \eqref{eqn:independent_cross-product_component} holds will suffice. When considering the function class $\calF_n$ of all possible piecewise constant functions on the partitions generated by the data, \cite{scornet2015consistency} proved \eqref{eqn:independent_squared_error_component} and \cite{saha2023random} proved \eqref{eqn:independent_cross-product_component} for any i.i.d. processes without assumptions on the nature of the outcome. Hence, these results hold for our binary setting for i.i.d. process and application of the result \cite{nobel1993note} subsequently proves the  analogs \eqref{eqn:dependent_squared_error_component} and \eqref{eqn:dependent_cross-product_component} under $\beta$-mixing dependence.
 \end{enumerate}

This shows that all 3 conditions of Theorem \ref{th:gyorfi} are satisfied by the chosen function class and the result is proved. 

\subsection{\blue{Proof of Corollary \ref{cor:pdf}: Consistency of PDF estimate}}
\blue{
Let $X^{(-j)}$ denote the vector formed by removing the $j^{th}$ covariate $X^{(j)}$ from $X$. Note that the PDF for the $j^{th}$ covariate is $p_j(X^{(j)}) = \int p(X) d(X^{(1)}) \ldots dX^{(j-1)} dX^{(j+1)} \ldots dX^{(D)} = \mathbb E_{X^{(-j)}} p(X)$ and similarly for the RF-GLS estimate $\hat p_j(X^{(j)})$.
\begin{align*} \mathbb E \int \Big(p_j(X^{(j)}) -  \hat p_j(X^{(j)})\Big)^2 dX^{(j)} =& \mathbb E \int \Big(\mathbb E_{X^{(-j)}} [ p(X) - \hat p(X)]\Big)^2 dX^{(j)} \\
\leq & \mathbb E \int \mathbb E_{X^{(-j)}} \Big( p(X) - \hat p(X)\Big)^2 dX^{(j)} \\
= & \mathbb E \int \Big( p(X) - \hat p(X)\Big)^2 dX \\
 \to &\, 0.
\end{align*}
Here the inequality follows from Jensen's inequality and the last limit is from Theorem \ref{th:main_gh}.
}

\subsection{\blue{Proof of Corollary \ref{cor:cate}: Consistency of CATE estimate}}

\blue{As each of $\{Y^{[k]}_i,X^{[k]}_i\}_i$ satisfies Assumptions 1-3, from Theorem \ref{th:main_gh}, we have $ \mathbb E \int (p_k(X) - \hat p_k(X))^2 dX \to 0$ for $k \in \{0,1\}$. Then we have 
\begin{align*} \mathbb E \int (\tau(X) - \hat \tau_T(X))^2 dX =& \mathbb E \int \left[(p_1(X) - p_0(X)) - (\hat p_1(X) - \hat p_0(X))\right]^2 dX \\
<& 2 \left[\mathbb E \int (p_1(X) - \hat p_1(X))^2 dX + \mathbb E \int (p_0(X) - \hat p_0(X))^2 dX \right] \\
\to & \, 0. 
\end{align*}}

\subsection{Proof of Theorem \ref{th:main_m}}
Equipped with Corollary \ref{th:main_gh_mixed_corollary} of Theorem \ref{th:main_gh}, we can establish  Theorem \ref{th:main_m}. First we prove that $\exists \varepsilon_1, \varepsilon_2 > 0$ that do not depend on $n$, such that $p(\bx)=\mathbb E (Y | X = \bx) \in [\varepsilon_1, 1 - \varepsilon_2]$. This is a direct consequence of boundedness of $m(\cdot)$ (Assumption \ref{item:assumption_m})  and continuity of $g_{h,\mathbb F_w}$ {\color{black}(Assumption \ref{item:assumption_differentiability_new})}. As $g_{h,\mathbb F_w}(\cdot)$ is strictly monotonic (Assumption \ref{item:assumption_h}),  there exists $\eps_1$ and $\eps_2$ such that $[\eps_1,1-\eps_2] = Im(p)= g_{h,\mathbb F_w}(Im(m))$. 

Next, we define an extension of the link-inversion function $g_{h,\mathbb F_w}(\cdot)$ from Proposition \ref{lemma:link_inverse_existance} as $\kappa: \mathbb R \mapsto \mathbb R$ as follows:
\begin{subequations}
\begin{equation}
    \label{eqn:inverse_link}
    \kappa(\alpha)     = 
	\begin{cases}
       z &\quad \text{ if } \alpha \in [\varepsilon_1,1 - \varepsilon_2] \text{ such that } g_{h,\mathbb F_w}(z)=\alpha\\
       g^+_{h,\mathbb F_w}(\varepsilon_1)  &\quad\text{ if } \alpha <  \varepsilon_1\\
       g^+_{h,\mathbb F_w}(1 - \varepsilon_2)&\quad\text{ if } \alpha > 1 - \varepsilon_2
     \end{cases}
\end{equation}

Since the image of $p(\cdot)=g_{d,\mathbb F_w}(m(\cdot))$ is $[\eps_1,1-\eps_2]$, the function $\kappa(\cdot)$ restricted to $g_{h,\mathbb F_w}(Im(m))$ agrees exactly with the inverse $g^+_{h,\mathbb F_w}(\cdot)$. From Theorem \ref{th:main_gh}, $\bar{p}_n$ is an $\mathbb L_2$ consistent estimator of $\mathbb E (Y | X)$. Being inverse of a strictly increasing continuous function $z \mapsto g_{h,\mathbb F_w}$, $\kappa$ is a continuous function. Also as $\bar p_n$ is uniformly bounded (Lemma \ref{lemma:bounded_estimate}), it suffices to define $\kappa$ on a compact support that does not depend on $n$. This combined with the definition of $\kappa$ which also does not depend on $n$, we have $\kappa$  to be also uniformly bounded. Hence by continuous mapping theorem, $\kappa\circ\bar{p}_n$ is an    $\mathbb L_2$ consistent estimator of $\kappa(\mathbb E (Y | X)) = m(X)$, i.e., 
\begin{equation}
    	\lim_{n \to \infty} \mathbb{E} \int \left(\kappa(\bar{p}_n(X)) - m(X) \right)^2 \, dX = 0
\end{equation}
\end{subequations}

\subsection{Proof of Corollary \ref{prop:spatias_application}}
We only need to show that {\color{black}Assumptions \ref{item:assumption_w_new}, \ref{item:assumption_differentiability_new}, and} \ref{as:working_cov} are satisfied for a probit or logit link $h$, a Mat\'ern GP $w(
\cdot)$ and a Nearest Neighbor Gaussian Process (NNGP) working precision matrix $\bQ$. 

A Mat\'ern Gaussian Process on a 1-dimensional regular lattice with half-integer smoothness can be characterized via stochastic partial differential equations (SPDE) which coincides with the SPDE for a continuous-space Autoregressive (AR) process \citep{hartikainen2010kalman,rasmussen2003gaussian}. When such a continuous-space process is sampled on a discrete set of points (regular lattice), the sampled process becomes an Autoregressive moving average (ARMA) process \cite{ihara1993information}. As ARMA processes are $\beta$-mixing \citep{mokkadem1988mixing}, assumption \ref{item:assumption_w} is satisfied. 

Next, we check {\color{black}Assumption \ref{item:assumption_differentiability_new}} for a probit or a logit link $h(\cdot)$. Continuity of $g_{h,\mathbb F_w}(\cdot)$ is immediate due to the bounded nature of $h(\cdot)$ by application of dominated convergence theorem (DCT). To assess differentiability, we note that 
$$\frac{g_{h,\mathbb F_w}(z_2) - g_{h,\mathbb F_w}(z_1)}{z_2-z_1} = \int \frac{h(z_2+w)-h(z_1+w)}{z_2-z_1} \mathbb F_w. $$
As the derivatives of both probit $h(z) = \Pr(Z < z)$, $Z \sim N(0,1)$ and logit $h(z) = (1+\exp(-z))^{-1}$ are bounded  in $\mathbb R$ by some $H$. We have 
$\bigg |\frac{h(z_2+w)-h(z_1+w)}{z_2-z_1}\bigg| \leq H$ and $g_{h,\mathbb F_w}(\cdot)$ is differentiable by another application of DCT. Also, using continuity of the derivative of $h$ for both logit and probit link and one further application of DCT proves {\color{black}Assumption \ref{item:assumption_differentiability_new}}. 

Finally, we check Assumption \ref{item:assumption_chol} for a precision matrix $\bQ$ from a Nearest Neighbor Gaussian Process (NNGP) built from a Mat\'ern GP with covariance function $C$. When using $q$ nearest neighbors, the matrix $\bQ$ admits a Cholesky decomposition $\bV^\top\bV$ where the $i^{th}$ row of the lower triangular matrix $\bV$ has sparse rows. For subsets $A,B \subset \{1,\ldots,n\}$, let $\bC_{A,B}$ denote the matrix $(C(s_i,s_j))_{i \in A, j \in B}$. Also let $N[i]$ denote the set of $\max(q,i-1)$ neaerest neighbors of the location $s_i$ among $s_1,\ldots,s_{i-1}$. The $i^{th}$ row consists of components of the vector $\rho_i = (1, -\bC_{i,N[i]}\bC_{N[i],N[i]}^{-1}) / \sqrt{1-\bC_{i,N[i]}\bC_{N[i],N[i]}^{-1}\bC_{N[i],i}}$ respectively at the positions $i,N[i]$ and zeros elsewhere. Thus for a 1-dimensional regular lattice design, for $i \geq q+1$ this simplifies to have $N[i]=\{i-1,\ldots,i-q\}$ and $\rho_i = \rho =  (1, -\bC_{q+1,1:q}\bC_{1:q,1:q}^{-1}) / \sqrt{1-\bC_{q+1,1:q}\bC_{1:q,1:q}^{-1}\bC_{1:q,q+1}^\top}$. Hence, for this design, $\bV$ has the required banded structure with $\rho$ provided above and 
 $\bL$ such that $\bL^\top\bL = \bC_{1:q,1:q}^{-1}$. 

\subsection{\blue{Proof of Proposition \ref{prop:spatias_application_2D}}}

\blue{We first state the technical version of Proposition \ref{prop:spatias_application_2D}. \\
\begin{proposition}
\label{prop:compacttech}
Consider binary data generated on a two-dimensional lattice $\{(k,l) : 1 \leq k \leq n_1, 1 \leq l \leq n_2, k,l \in \mathbb Z\}$ 
from a generalized mixed effects model (\ref{eqn:hgnlm}) with $X$ and $m$ as in Assumptions \ref{item:assumption_x} and \ref{item:assumption_m}, a probit or logit link $h$, and $w \sim GP(0, C_{\theta_0})$ (independent of $X$), where $C_{\theta}$ is any class of stationary compactly supported covariance functions. Let $\phi:=\phi(\theta)$ denote the component of the parameter $\theta$ of $C_\theta$ which controls the magnitude of spatial correlation, i.e., $C_\theta(d) \to 0$ as $\phi \to \infty$ for any distance $d$. 
Let $p_n$ and $m_n$ denote the RF-GP estimates of respectively $p(X)=\mathbb{E}(Y | X)$ and $m(X)$, obtained using a working precision matrix $Q$ from a Nearest Neighbor Gaussian Process (NNGP) based on $C(\theta)$. Then as either $n_1 \to \infty$ or $n_2 \to \infty$, under the scaling of Assumption \ref{as:tn_rate}, there exists a $K > 0$ such that for any $\theta$ with $\phi(\theta) > K$, $p_n$ and $m_n$ are respectively $\mathbb L_2$ consistent estimates of $p(X)$ and $m(X)$.\\
\end{proposition}

\noindent \begin{proof} Without loss of generality, we consider $n=n_1n_2$ datapoints observed on a $n_1 \times n_2$ lattice (matrix) design with the column dimension $n_2 \to \infty$ and row dimension fixed at $n_1$. We enumerate processes on the lattice column-wise, i.e., first column is read first, then second column and so on. Let $a := a(\theta)$ denote the threshold of the compactly-supported function $C_\theta$. Then $w_i \independent w_j$ for any $|j-i| > n_1 a$, so $w_i$ is an $m$-dependent Gaussian process which implies ergodicity and $\beta$-mixing \citep{bradley2005basic}. So conditions \ref{item:assumption_w_new} and \ref{item:assumption_differentiability_new} are immediately satisfied. 

Hence, only Assumption \ref{as:working_cov} needs to be proved for the NNGP working precision matrix $Q$ on this lattice. However, for this two-dimensional design Assumption \ref{item:assumption_chol} is not satisfied. This is because the Cholesky factor of $Q$ does not have identical sparse and banded rows after the first few rows. To elucidate with an example, consider a $3 \times 5$ lattice and NNGP precision matrix $Q$ constructed using two nearest neighbors. As the lattice is read columnwise, the 2 nearest neighbor directed graph dictating the sparsity of the Cholesky factor of $Q$ is given in Figure \ref{fig:lattice}. This leads to 3 subsets of locations, colored differently in the figure. These will correspond to different types of rows in the Cholesky factor. The black set of locations will correspond to the first part of the sequence for which the neighbor structure is atypical (non-repetitive) because of not having enough neighbors. These are the analog of the locations which corresponded to the $L$ submatrix in Assumption \ref{item:assumption_chol}. This set (first column in the example of Figure \ref{fig:lattice} and more generally first few columns) will always be of fixed cardinality which does not grow with $n_2$ and will only depend on $n_1$ and $q$. 

The red set of locations (in this case top row except the first few column) will have a repetitive pattern of neighbors (neighbors to the west and south-west), while the remaining  locations (green set) will also have a repetitive but different pattern of neighbors (neighbors to the west and north). These will correspond to different types of rows in the Cholesky factor. 

\begin{figure}[t]
    \centering
    \includegraphics[width=0.5\linewidth]{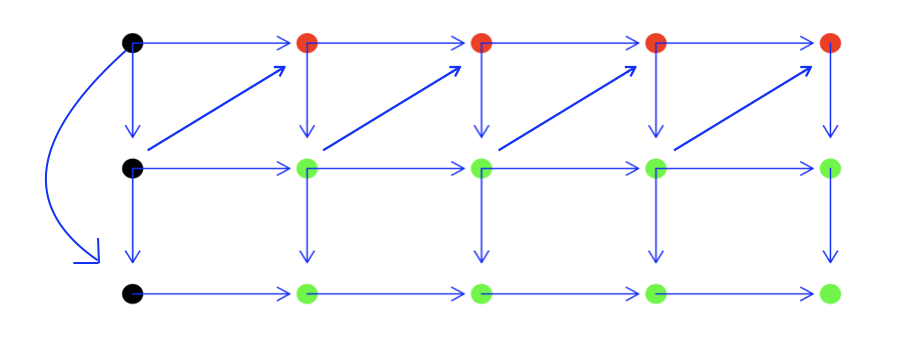}
    \caption{\blue{NNGP with $2$ nearest-neighbor directed graph on a $5\times 3$ lattice. Different colors correspond to locations which yield different rows in the Cholesky factor of the NNGP precision matrix $Q$.}}
    \label{fig:lattice}
\end{figure}
\end{proof}

So, the Cholesky factor $Q^{1/2}$ for the example in Figure \ref{fig:lattice} is given by

\begin{equation}
	\label{eqn:assumption_chol_gen}
	\bQ^{\frac 12} = \left(\begin{array}{ccccccccc}
	\multicolumn{3}{c}{\bL_{3 \times 3}} & 0 & 0 & \cdots & \cdots & \cdots & \cdots\\
    \psi_2 & \psi_1 & 0 & \psi_0 & 0 & 0 & 0 & 0 & 0  \\  
    0 & \rho_2 & 0 & \rho_1 & \rho_0 & 0 & 0 & 0 & 0  \\
    0 & 0 & \rho_2 & 0 & \rho_1 & \rho_0 & 0  & 0 & 0  \\
    0 & 0 & 0 & \psi_2 & \psi_1 & 0 & \psi_0 &  0 & 0  \\ 
	\multicolumn{3}{c}{\hdots} & \multicolumn{3}{c}{\hdots} & \multicolumn{3}{c}{\hdots} 
	\end{array} \right), 
	\end{equation}

We see that the Cholesky factor has more complex banding than what is needed in Assumption \ref{item:assumption_chol}. After the first 3 rows (for the black points), each row with $\psi=(\psi_2,\psi_1,\psi_0)'$ corresponding to the red points is intersparsed between two consecutive rows of $\rho=(\rho_2,\rho_1,\rho_0)'$ and this pattern of one row of $\psi$ followed by 2 rows of $\rho$ continues in a repetitive manner as $n_2$ grows.

We can write the matrix in (\ref{eqn:assumption_chol_gen}) as a block banded Cholesky factor of the form
	\begin{equation}
    \label{eqn:assumption_chol_block}
	R_1 = \left(\begin{array}{cccccc}
	\bL & 0 & 0 & \cdots & \cdots\\
    R_1 & R_0 & 0 & \cdots & \cdots\\
	0 & R_1 & R_0 & 0 & \cdots\\
	\vdots & \multicolumn{3}{c}{\ddots} & \vdots \\
	\cdots & 0 & 0 & R_1 & R_0
	\end{array} \right), 
	\end{equation}
Here $L$ corresponds to first $q_0$ columns of the lattice with nonrepetitive neighbor structure ($q_0=1$ in Figure \ref{fig:lattice}), the matrix $R=(R_1,R_0)$ specifies the block banding structure. In (\ref{eqn:assumption_chol_gen}), we have 

	\begin{equation}
	R_1 = \left(\begin{array}{ccccc}
	\psi_2 & \psi_1 & 0  \\
    0 & \rho_2 & 0  \\
    0 & 0 & \rho_2  \\
	\end{array} \right),\, R_0 = \left(\begin{array}{ccccc}
	\psi_0 & 0 & 0  \\
    \rho_1 & \rho_0 & 0  \\
    0 & \rho_1 & \rho_0  \\
	\end{array} \right).
	\end{equation}

This more general banding structure (\ref{eqn:assumption_chol_block}) is a block-matrix version of (\ref{eqn:assumption_chol}) with $R$ being the matrix analog of the vector $\rho$ in (\ref{eqn:assumption_chol_block}). Each block of rows with $R$ in corresponds to a column of the lattice, and thus $Q^{1/2}$ in (\ref{eqn:assumption_chol_block}) specifies a vector-autoregression (VAR) structure,   generalizing the AR structure of (\ref{eqn:assumption_chol}).  
The diagonal block $R_0$ is a lower triangular matrix and contains the lag 0 VAR coefficients (analog of $\rho_0$ in (\ref{eqn:assumption_chol}). 
The offdiagonal $R_1$ corresponds to the lag 1
VAR coefficients. In the more general case, there can be more lagged VAR coefficient matrices $R_2, R_3, \ldots, R_{q_1}$ and $R=(R_{q_1},\ldots,,R_1,R_0)$ will be the $n_1 \times (q_1+1)n_1$ coefficient matrix. 

We will go through parts of the proof of Theorem \ref{th:main_gh} that relies on the banding condition (Assumption \ref{item:assumption_chol}) and modify them for the new banding structure of the form (\ref{eqn:assumption_chol_block}).

First we note that quadratic forms involving $Q$ will not have the expression in (\ref{eq:qf}) which was central to deriving the limit of the GLS loss and GLS estimate for a fixed partition tree.  
For any vector $u$ indexed on the lattice, let $u_{[i]}$ denote the sub-vector corresponding to the $i^{th}$ column of the lattice. 
We can write a quadratic form in $Q$ as:  
\begin{align}\label{eq:qf_gen}
\bu^\top\bQ\bv 
&= \sum_i u_{[i]}'(\sum_{j=0}^{q_1} R_j'R_j) v_{[i]} + \sum_{j \neq j' = 0}^{q_1} \sum_{i} u_{[i-j]}'R_j'R_{j'}v_{[i-j']} + \sum_{i \in \mathcal{A}}\sum_{i' \in \mathcal{A}'} u_{[i]}'\tilde\Gamma_{i,i'}v_{[i']},
\end{align}
where the sum over $i$ runs from $1$ to $n_2$, ${\mathcal{A}}, {\mathcal{A}}' \subset \{1,2,\cdots,n_2 \}$ with $|{\mathcal{A}}|, |{\mathcal{A}}'| \leq 2q_0$, 
and $\tilde \Gamma_{i,i'}$'s are fixed (independent of $n_2$) matrices with entries that are functions of $L$ and $R$. For $i \leq 0$,  
$u_{[i]}$ and $v_{[i]}$ are defined to be zero vectors.

To calculate the limits of the GLS estimate ${\hat\beta} = \left( {Z}^\top\bQ {Z}\right)^{-1} {Z}^\top\bQ {Y}$ and the GLS loss, we need to derive the asymptotic limits of 
$\left(\frac{1}{n}{Z}^\top\bQ {Z}\right)^{-1}$ and $\frac{1}{n}{Z}^\top\bQ {Y}$.  
Using \eqref{eq:qf_gen} we have:

\begin{equation}\label{eq:cor_gen}
	\begin{aligned}
	\frac{1}{n} ({Z}^\top\bQ{Y})_l &= \frac 1{n_1} tr\left((\sum_{j=0}^{q_1} R_j'R_j) \frac 1{n_2} \sum_i Z_{[i],l}' Y_{[i]} \right) \\
    & \, + \frac 1{n_1} tr\left(( \sum_{j \neq j' = 0}^{q_1} R_j'R_{j'}) \frac 1{n_2} \sum_i Z_{[i-j],l}' Y_{[i-j']} \right) + \frac 1n O_b(1),
\end{aligned}
\end{equation}
where $O_b(1)$ denotes a sequence of random variables which are uniformly bounded. 

Each matrix of the form $\frac 1{n_2} \sum_i Z_{[i-j],l}' Y_{[i-j']}$ above is a fixed dimensional ($n_1 \times n_1$) matrix whose limit (as $n_2 \to \infty$) can be obtained entrywise similarly as in Lemma \ref{lemma:DART-theoretical} using the $\beta$-mixing property. We then have the following limits:

\begin{equation}\label{eq:gls_lim_block}
\begin{aligned}
    \frac{1}{n} ({Z}^\top\bQ{Y})_l &\to \gamma E(Y \given X \in \calB_l) Vol(\calB_l) + \delta E(Y) Vol(\calB_l), \\
    \frac{1}{n} ({Z}^\top\bQ{Z})_{ll'} &\to \gamma Vol(\calB_l)\mathbb I(l=l') +  \delta Vol(\calB_l)Vol(\calB_{l'}), \\
    \mbox{ where } \gamma  = \frac 1{n_1}tr&\left(\sum_{j=0}^{q_1} R_j'R_j\right) \mbox{ and } \delta= \frac 1{n_1} \sum_{j,j' = 0}^{q_1} 1'R_j'R_{j'}1 -\gamma.
\end{aligned}
\end{equation} 
 
Hence, $(\frac 1n Z'QZ)^{-1}$ and $\frac 1n Z'QY$ have the similar functional forms of limits as in the proof of Lemma \ref{lemma:DART-theoretical} but with different choices of $\gamma$ and $\delta$. Consequently, the GLS estimate and the loss will also have same limits as in Lemma \ref{lemma:DART-theoretical} as $\gamma$ and $\delta$ cancels out. 

Next we verify Assumption \ref{item:assumption_diagonas_dominance} for $Q$ in this design. Note that $Q^{1/2}$ only has a maximum of $q_0+n_1$ unique row structures. 
As the entries of all of these rows are based on kriging weights 
and the covariance matrix $C_\theta$ has a parameter $\phi$ controlling the magnitude of correlation, there exists some $K$ such that for any $\phi > K$, $Q$ is strictly diagonally dominant with a 
positive threshhold $\xi > 0$ 
\citep[see proof of Proposition 3.1 in ][for more details]{saha2023random}. 
Hence, $Q$ satisfies Assumption \ref{item:assumption_diagonas_dominance} which implies that there is an uniform bound on the GLS estimator (in a similar way as in Lemma \ref{lemma:bounded_estimate}). An important byproduct of the diagonal dominance is that 
$\gamma - |\delta| \geq \xi$. 

The final part is to prove an analog of Theorem \ref{th:gyorfi} when $Q$ does not satisfy Assumption \ref{item:assumption_chol}, instead $Q^{1/2}$ is as in \eqref{eqn:assumption_chol_gen} and the design is a two-dimensional lattice. As these only only impacts the estimation error part, the rest of the proof related to approximation error remains the same. 

Controlling the estimation error in Theorem \ref{th:gyorfi} relies on the ULLNs (\ref{eqn:dependent_squared_error_component}) and (\ref{eqn:dependent_cross-product_component}). First we note that analogs of these two can be established separately for each of the elements of the subsequences $(Y_{i]},X_{[i-j]})$ for all finite lags $j$. 
This is because these are regularly-spaced subsequences of the full stationary ergodic $\beta$-mixing process $\{Y_i,X_{i-j}\}$. 
Hence, they will also be stationary ergodic $\beta$-mixing and the subsequence-specific ULLNs can be proved as in the proof of Theorem \ref{th:main_gh}. 

Equipped with the ULLNs, we now follow the last part of the proof of Theorem \ref{th:gyorfi} making adjustment for the different banding structure in $Q^{1/2}$. Let $\dot X^{(q_1+1)}$ denote a $n_1 \times (q_1 + 1)$ random matrix identically distributed as that of the $X$-process on $q_1 + 1$ consequtive columns of the lattice, and $\dot X^{(q_1+1)}$ is independent of $X$. 
We first prove that to show $\mathbb L_2$ consistency of $p_n$ 
it is enough to show that 
\begin{equation}
    \label{eqn:gyorfi_quadratic_block}
    \lim_{n \to \infty}\mathbb{E}\left[\mathbb{E}\; tr \left[f_n(\dot X^{(q_1+1)}) R^\top Rf_n(\dot X^{(q_1+1)})^\top\right]\right]=0, \mbox{ where } f_n=p_n-p.
\end{equation}

To see this, for any fixed function $f$, we have 
$$\begin{aligned}
    \mathbb{E}\; tr \left[f(\dot X^{(q_1+1)}) R^\top R f(\dot X^{(q_1+1)})^\top\right] & =  tr \left[ R^\top R\; \mathbb{E}\;\left[f(\dot X^{(q_1+1)})f(\dot X^{(q_1+1)})^\top\right] \right]\\
    & = tr \left[ R'R\; \left[ (Var(f(X_i)) I + ( \mathbb Ef(X_i))^2 11'  \right] \right] \\
    & = (Var(f(X_i)) tr(R'R) + ( \mathbb Ef(X_i))^2 1'R'R 1 \\
    & = \mathbb E(f(X_i)^2) n_1 \gamma + ( \mathbb Ef(X_i))^2 n_1 \delta  \\ 
    & \geq n_1 (\gamma - |\delta|) \mathbb E(f(X_i)^2). 
\end{aligned}$$

As $\gamma - |\delta| > 0$, using $f=f_n$ and taking the limit we see that proving \eqref{eqn:gyorfi_quadratic_block} will imply $\mathbb L_2$ consistency. Equation (\ref{eqn:gyorfi_quadratic_block}) thus replaces (\ref{eqn:gyorfi_quadratic}) and the proof can be completed by repeating the remaining steps in Theorem \ref{th:gyorfi}.}

\section{Implementation of RF-GP}

\subsection{Working correlation matrix}
\label{sec:correlation_estimation}
The estimate of the mean function $p(\cdot)$ in Section \ref{sec:gini} requires a suitable working covariance matrix $\bQ^{-1}$ that adequately captures the spatial dependence in the data. The finite sample performance of GLS estimators depends on how closely the working covariance matrix approximates the population covariance matrix $Cov( Y \given \bX)$. 
In RF-GLS with the continuous outcome, due to the linear relationship between the outcome and the spatial random effects, the population covariance of the outcome is the sum of the covariance of $\bw$ and the noise covariance. In the case of binary outcome with a non-linear link $h$, there is no closed form expression of the marginal population covariance matrix of $ Y \given \bX$, even when $\bw$ has a structured covariance matrix $\bC$. 

For non-Gaussian data, there is a long-established history of using working covariance matrices in generalized estimating equations \citep{liang1986longitudinal}. We follow the conventional GLM approach in considering working covariance matrices of the form 
\begin{equation}\label{eq:workingcov}
    \bQ^{-1} = \sigma^2 \bV^{1/2} \bR(\zeta) \bV^{1/2} 
\end{equation}
where $\bV$ is a diagonal matrix with entries $a(p(X_i))$, $a(\cdot)$ being a working variance function, and $\bR$ is a working correlation matrix parameterized by $\zeta$, capturing the spatial dependence. Common choices of $v(\cdot)$ can be the Bernoulli ($a(p)=p(1-p)$) or Gaussian ($a\equiv 1$) working variance functions \citep{mccullagh2019generalized}. We use the latter as other choices will involve the unknown mean function $p$. We show in Section \ref{sec:consistency} that this choice of the variance function along with the choice of the correlation matrix described below leads to consistent estimates for our algorithm. 

\blue{We outline some general strategies for choosing the class of working correlation matrix $R(\zeta)$ for the RF-GLS part of the method, and for tuning the hyper-parameters $\zeta$. The specific choice will depend on the application. For example, when analyzing time-series data and using the direct auto-regressive (AR) time series model (as described in Section \ref{sec:simulation_AR}), one can chooe $R(\zeta)$ to the class of AR covariances, i.e., $R(\zeta)=(\zeta^{|i-j|})$.  
Other time-series covariance models can also be considered, like AR of higher order or ARMA models, but we do not explore that here.} If working within the generalized mixed effects model framework of Section \ref{sec:rfgp}, the working spatial correlation matrix family $\bR(\zeta)$ can be chosen based on a crude first order Taylor series approximation, as follows  
\begin{equation*}
\begin{aligned}
    Cov(Y_i, Y_j \given \bX) &= E (Cov(Y_i, Y_j \given \bX, \bw) \given \bX) + Cov( \mathbb{E}(Y_i \given \bX, \bw), \mathbb{E}(Y_j \given \bX, \bw) \given \bX)\\
    &= 0 + Cov(h(m(X_i)+w(s_i)), h(m(X_j)+w(s_j))\given \bX) \\
    &\approx Cov(h(m(X_i))+w(s_i)h'(m(X_i)), h(m(X_j))+w(s_j)h'(m(X_j))\given \bX)\\
    &= h'(m(X_i)) h'(m(X_j)) C_{ij}.
\end{aligned}
\end{equation*}
\blue{Writing $h'(m(X_i))$ as some function $a(p(X_i))$ of the mean by using the link-inversion, and letting $C_{ij}=\sigma^2 R_{ij}$ for a stationary covariance, we have  
$Cov(Y \given X) \approx \sigma^2_w V^{1/2} R(\zeta) V^{1/2}$, i.e., the same form as in (\ref{eq:workingcov}).} Hence \blue{for mean estimation}, we choose $\bR(\zeta)$  to be the same correlation family $R$ which specifies the GP distribution of the spatial  effects $\bw$. However, we allow the parameter for $\zeta$ for $\bR$ to be different from the parameters $\theta$ of the GP covariance function $C$ generating the spatial effects. We use the popular Nearest Neighbor GP (NNGP) approximation of the working spatial matrix $\bR$. Theoretical justification of using the NNGP working covariance function is provided in Proposition \ref{prop:spatias_application}. Subsequent to choosing the family of the working precision matrix $Q$, we estimate the hyper-parameters  
specifying it using cross-validation, detailed afterward.

\blue{\subsection{Computational complexity}\label{sec:compute}
In our implementation of RF-GP, we use the nearest neighbor Gaussian Process (NNGP) family of covariance functions \citep{datta2016nearest} to model the working covariance for the mean estimation using RF-GLS.  
Given a direction and a split, the time to evaluate a loss function \eqref{eqn:DART} and node representative assignment \eqref{eqn:DART_mean} is linear in $n$, as long as we use a working precision matrix with a fixed number of non-zero elements in each row.  The working precision matrix obtained using NNGP satisfies this \citep[see ][for an explanation]{finley2019efficient}. So each evaluation of the split criterion is linear in $n$. 

We need to optimize the loss function over the split direction and over each possible gap for each covariate. Evaluating the loss at $n - 1$ possible gaps of each of the $m_{TRY}$ many covariates for a tree with $t$ nodes requires $O(n(n-1)tm_{TRY})$ complexity. So the entire RF-GLS algorithm has $O(n^2)$ computational complexity. For small sample sizes $(n < 500)$, as considered in the data analysis of this article, this quadratic computational complexity does not create a problem. 

For larger $n$, a potential way to bypass would be to restrict the number of gaps the loss function needs to be evaluated at for optimization. A way to do this would be rounding or binning the covariates to the nearest $\alpha$\% quantile of the covariate. This reduces the total number of gaps in the covariate to $q = [100/\alpha]$, irrespective of the value of $n$ and does the binning in a data-driven way. For example, binning to the nearest even quantile would create $q=50$ gaps which will be much less than $n$ when the latter is large. This ensures linear order computational complexity $O(qn)$ for the whole algorithm. This strategy can be used for mean estimation using RF-GLS for both continuous and binary data. We ran some simulation experiments using this quantile binning for RF-GLS and it drastically reduces the runtimes, scaling to datasets with tens of thousands of locations.   
We will investigate this computational strategy more extensively in our future work.}

\subsection{Estimation details}\label{sec:details}
For the estimation of the mean function in RF-GP, which uses RF-GLS, we fix the number of trees in the  to be $n_{tree} = 500$; the minimum cardinality of leaf nodes to be $t_c = 20$ and
the number of features to be considered at each split to be $M_{try} = \min \{1, \left[\frac{D}{3}\right] \}$, where $D$ is the number of features in the covariate space. We assume the exponential covariance function for $w$. 
Hence, as explained in Section \ref{sec:correlation_estimation}, we also use the same covariance function form to define the working correlation matrix in \ref{sec:correlation_estimation}.  
We need to estimate two sets of spatial parameters for fitting RF-GP. The first one is the spatial decay parameter of the exponential covariance function in the working correlation matrix $ R$, denoted as $\zeta$. The other one is the set of spatial parameters corresponding to the spatial process $w$, needed for the estimation of the \blue{covariate} effect and predictions. This is given by $\theta = \left( \sigma^2, \phi\right)$ where $\sigma^2$ is the variance and $\phi$ is the spatial decay.  
We estimate the parameters by 2-fold cross-validation with respect to the misclassification error. We perform the cross-validation through a grid search with $ \zeta \in \left\{ 0.5, 1, 2, \ldots, 10, \infty\right\}$, $\sigma^2 \in \left\{ 1, 2.5, 5, 7.5, \ldots, 25\right\}$, and $\phi = \left\{ 0.5, 1, 2, \ldots, 10, \infty\right\}$. Here we note that the first part of RF-GP, i.e., estimating the mean function using RF-GLS depends only on $\zeta$. We incorporate $\zeta = \infty$ in the cross-validation parameter set as it corresponds to an identity working covariance matrix, in which case, the mean function estimation with RF-GLS becomes equivalent to the classical RF. This choice is included to account for the possible absence of spatial correlation structure in the data and in practice implemented by using a very large $\zeta (=10000)$. Similar argument holds for $\phi = \infty$ case. The methods were implemented in \texttt{R}. For model fitting, we used \texttt{RandomForestsGLS} \citep{RandomForestsGLS} to implement RF-GP and \texttt{randomForest} \citep{randomForest} to implement the remaining methods. We used probit-NNGP \citep{saha2022scalable} implemented \texttt{BRISC} \citep{brisc} for binary prediction.

Estimating the covariate effect using link inversion (\ref{eqn:probit_link_inverse}) requires the mean estimates $\widehat{p(\bx)}=\widehat{\mathbb E \left( Y | X =  x\right)}$ to lie in $\left(0, 1\right)$. Neither RF nor RF-GP estimates are guaranteed to lie in $\left(0,1\right)$. For small imbalanced data, RF mean estimate at a covariate value can coincide with $0$ or $1$, depending on the predominant class. RF-GP estimates of the mean function are obtained using RF-GLS. Being a GLS estimate, it is not guaranteed to lie in  $\left( 0, 1\right)$ although it is uniformly bounded (Lemma \ref{lemma:bounded_estimate}). For mean estimation, we truncate the original estimate at $[0,1]$. For link inversion, we interpolate the estimates to $(0,1)$  
exploits the assumed regularity of the covariate effect $m \left ( \cdot\right)$. If the mean estimate at a covariate value lies outside $(0,1)$, 
we first generate points $\bx$ uniformly in the covariate space and obtain $\widehat p(\bx)$ from the RF-GLS estimate. Next, we train any statistical interpolator $model$ with only the $\widehat p(\bx)$'s which lie in $(0,1)$ and the corresponding $\bx$'s. Then for $\bx$'s for which $\widehat p(x)$ lies outside $(0,1)$, we replace $\widehat p(x)$ with the predicted value from the interpolator. We use classical RF for this interpolation as it preserves the range of the outcomes in the training data, although any other interpolator with the same property can also be used.

\section{Simulation experiments}
\label{subsec:simulation}
\subsection{\blue{Competing methods}}\label{sec:methodssim}
\blue{We compare the performance of our method with the following state-of-the-art competitors:}
{\color{black}
\begin{enumerate}
    \item RF \citep{breiman2001random}: Breiman's RF, without any spatial information.
    \item RF-Loc: Breiman's RF, with the spatial locations used as added-spatial-features. 
    \item RF-Sp \citep{hengl2018random}: Uses buffer distances as added-spatial-features. 
    We use pairwise Euclidean distances as buffer distances.
    \item Basis function GLMM (Basis GLMM): A non-linear mixed effect model (\ref{eqn:hgnlm}) with $m$ expressed using basis functions and $w$ modeled as a GP. Fitted as a 
    generalized linear mixed effects model with covariates generated by the multivariate linear splines used by \cite{holmes2003generalized}. 
    \item GAM: Generalized additive model with the spatial effect modeled by 2d smoother of the space. 
    \item Deep Learning: Feed-forward neural network, does not use spatial information
    \item DeepKriging: Feed-forward neural network with added-spatial-features as proposed in \cite{chen2020deepkriging}. Adds radial basis functions of coordinates (using the Wendland compactly supported correlation function) added-spatial-features.
    \item BART: Bayesian Additive Regression Tree of \cite{chipman2010bart}; without using any spatial information.
    \item BART-Loc: BART with added-spatial-features, using the spatial locations used as added-spatial-features.
    \item BARTDeepKriging: BART with added-spatial-features, using the same basis functions as added-spatial-features 
    as DeepKriging, i.e., the Wendland basis functions.
\end{enumerate}

We implement the competing methods using existing R packages. RF, RF-Loc, and RF-Sp are implemented using \texttt{randomforest} package \citep{randomForest}. We use \textit{spmodel} to implement the Basis GLMM. We model the covariance using exponential covariance function option of the package. GAM is implemented via \texttt{mgcv} package \citep{wood2003thin}. Deep Learning and DeepKriging are implemented using \texttt{keras3} package \citep{keras3}. BART, BART-Loc, and BARTDeepKriging are implemented via \texttt{BART} package \citep{BART}. 

For comparing the performance of the competing methods, we focus on three objectives:

\begin{enumerate}
    \item \textbf{Mean estimation}: Given a correlated binary outcome $Y$ and covariate $X$, we focus on estimating the conditional mean $p(X) = \mathbb E \left(Y | X\right )$. This works under the general moment-based data generating mechanism in \eqref{eqn:genbin} and need not be within the mixed effects model setup.
    \item \textbf{Covariate effect estimation:} In the mixed effects model setup of \eqref{eqn:hgnlm}, we focus on estimating the covariate effect, $m \left(X \right)$.
    \item \textbf{Prediction:} In the mixed effects model setup, given new covariates $X_{new}$ at new locations $s_{new}$, we predict $\mathbb{E}(Y_{new} \given X_{new}, s_{new})$.
\end{enumerate}

All except $4$ of the competing methods cannot perform all three tasks. For example, none of the added-spatial-features methods provide  estimates of the conditional mean $\mathbb E(Y \given X)$ or the covariate effect $m$ as they are not from the mixed effects model framework and use the covariates and the spatial locations together to estimate of $\mathbb E(Y \given X,s)$. Hence, they are only suitable for spatial predictions.

We provide a list of competing methods and the tasks they can perform in Table \ref{tab:truthTables}. Existing implementation of BART and RF do not estimate the covariate effect $m$, but we have integrated our link-inversion approach into BART and RF to obtain estimates of the covariate effect. }

\begin{table}[h]
\captionsetup{labelfont={color=black},font={color=black}}
 \caption{{\color{black}Summary of competing methods and what they estimate.}}
    \label{tab:truthTables}  
\begin{center}
{\color{black}\begin{tabular}{||c | c | c | c||}
 \hline
\backslashbox{ \textbf{Method} }{\textbf{Task}} & \begin{tabular}{@{}c@{}}Mean \\ $\mathbb E \left( Y | X\right)$\end{tabular}  & \begin{tabular}{@{}c@{}}Covariate effect \\ $m\left( X\right)$ in \eqref{eqn:hgnlm}\end{tabular}  & \begin{tabular}{@{}c@{}}Prediction \\ $\mathbb{E}(Y_{new} \given X_{new}, s_{new})$\end{tabular} \\ [0.5ex]
 \hline\hline
 RF-GP & \checkmark & \checkmark & \checkmark \\
 \hline
 RF & \checkmark & \checkmark & \checkmark \\
 \hline
 RF-Loc & $\times$ & $\times$ & \checkmark \\
 \hline
 RF-Sp & $\times$ & $\times$ & \checkmark \\
 \hline
 Basis GLMM & \checkmark & \checkmark & \checkmark \\
 \hline
 GAM & $\times$ & \checkmark & \checkmark \\
  \hline
 Deep Learning & $\times$ & $\times$ & \begin{tabular}{@{}c@{}} Estimates $Y_{new} $ \\ given $ X_{new}, s_{new}$\end{tabular}  \\
 \hline
 DeepKriging & $\times$ & $\times$ & \begin{tabular}{@{}c@{}} Estimates $Y_{new} $ \\ given $ X_{new}, s_{new}$\end{tabular}  \\
 \hline
 BART & \checkmark & \checkmark & \checkmark \\
 \hline
  BART-Loc & $\times$ & $\times$ & \checkmark \\
 \hline
 BARTDeepKriging & $\times$ & $\times$ & \checkmark \\ [1ex]
 \hline
\end{tabular}}
\end{center}
\end{table}

{\color{black}\subsection{Simulation from the binary auto-regressive process}\label{sec:simulation_AR}
We first consider a data generation process which does not come from the mixed effects model of (\ref{eqn:hgnlm}). We simulate correlated binary data from \eqref{eqn:genbin} for a direct binary auto-regressive $AR(1)$ process, using the R package \texttt{bindata} \citep{bindata}. We vary the serial correlation parameters $\rho \in \{0.3, 0.5 \}$ to generate datasets of varying magnitudes of correlation. The mean $p$ in \eqref{eqn:genbin} is given as $\Phi \left( f(X)/\alpha\right)$, where $f$ is 
the Friedman function, given as
\begin{equation}
\label{eqn:friedman}
    f(X) := \left[10\sin \left( \pi *X_1 *X_2\right) + 20 (X_3 - 0.5)^2 + 10 (X_4 - 0.5) + 5 (X_5 - 0.5)\right]/5,
\end{equation}
and $\alpha = 3.5$, and $30$ for $\rho = 0.3$ and $0.5$ respectively. The i.i.d covariates are generated uniformly from $[0, 1]^5$.

In Table, \ref{tab:AR_process}, we compare the performance of the $4$ methods that estimate $p \left( X \right) = \mathbb E \left( Y | X\right)$ in this setup, with respect to the Mean Integrated Squared Errors (MISE). RF, BART and Basis GLMM perform comparably, with BART being slightly better. RF-GP significantly outperforms BART, RF, and Basis GLMM, especially in the high correlation setup. This is expected since a) RF and BART do not account for the correlation structure, and b) Basis GLMM accounts for the correlation structure but models the covariate effect via basis functions, which does not perform well even for moderately high dimension (5 in this specific example).  In the high correlation case, the MISE of BART is $100\%$ larger than that of RF-GP, while the MISE of basis function based GLMM is nearly $200\%$ larger.}

\begin{table}[t]
\captionsetup{labelfont={color=black},font={color=black}}
 \caption{{\color{black}MISE for $\mathbb E \left( Y | X\right)$ estimation for binary $AR(1)$ process. Lower is better. The best performance in each column is marked in bold.}}
    \label{tab:AR_process}  
\begin{center}
{\color{black}\begin{tabular}{||c | c | c||}
 \hline
\backslashbox{ \textbf{Method} }{\textbf{Setup}}   & $\rho = 0.3$  & $\rho = 0.5$ \\ \hline\hline
 RF-GP & \textbf{0.01124} &  \textbf{0.00869} \\
\hline
 RF &  0.01745 & 0.01819 \\
 \hline
 Basis GLMM & 0.01622 &  0.02296    \\
 \hline
 BART & 0.01488  & 0.01663 \\[1ex]
 \hline
\end{tabular}}
\end{center}
\end{table}

{\color{black}\subsection{Simulation from generalized mixed effects model in \eqref{eqn:hgnlm}}
\label{sec:simulation_mixed}
 We simulate data from the spatial generalized mixed effects model (\ref{eqn:hgnlm}) with probit link. We consider $m$ to be $f(X)$ \eqref{eqn:friedman}. The i.i.d covariates are generated uniformly from $[0, 1]^5$. The locations are generated on a unit square $[0,1]^2$. The spatial error $w(s)$ is generated from an exponential GP ($\nu=1/2$ in (\ref{eq:matern})).
 We vary the spatial decay parameter $\phi$ (i.e. the inverse of the range parameter for the exponential covariance)  $ \in \{ 2,3,4,5\}$. We vary the spatial variance $(\sigma^2)$ over $\left\{ 2,\ldots,9\right\}$ to compare how the methods compare as the spatial variance increases. In order to rule out small sample variations, we simulate $100$ times for each parameter combination. Following \cite{Saha2022}, we use a $90\% - 10\%$ test-train split. 
We first divide the spatial domain $[0,1]^2$ into $10 \times 10$ equal square boxes with and then randomly choose one box from each row and column (i.e. total $10\%$ of the boxes). We keep the data with spatial locations within those boxes as test/holdout data. In this approach, entire regions are held out for testing as opposed to random locations, which may be proximal to other locations in the training set. 

First, we assess the estimation performance of the methods
for both the covariate effect $m(X)$ and the marginal mean function $p(X) =  \mathbb{E}\left( Y | X\right)$. Tables \ref{tab:m_mixed} and \ref{tab:p_mixed} show the Mean Integrated Squared Errors (MISE) for the estimation of $m$ and $p$,  
respectively, under different choices of the spatial covariance parameters generating the data. We observed that for both functions, RF-GP generally outperforms its competition for high spatial correlation scenarios (low $\sigma^2$ or high $\phi$). For low spatial correlation scenarios (low $\sigma^2$ or high $\phi$), the RF and BART perform comparably to RF-GP. This is primarily due to the fact that when the spatial signal is low, RF and BART which do not use any spatial information, are nearly correctly specified. However, as RF and BART does not account for spatial error, they perform considerably worse in high spatial signal.

Basis GLMM accounts for spatial correlation using GP random effects, but uses basis functions to model the covariate effect, which does not perform well as the covariate dimension is $5$. GAM on the other hand models the covariate effect additively, using smooth functions of covariates, hence cannot model the interaction term \citep{zhan2024neural}. Since our mean function is a version of the Friedman function, which includes the interaction term, it is natural that GAM will not do well. 

Nearly in all the scenarios, the MISE of RF-GP performs is either lowest or very close to the lowest, whereas for each alternate method under certain scenarios the MISEs are $100\%-300\%$ larger than that of RF-GP.

\begin{table}[H]
\captionsetup{labelfont={color=black},font={color=black}}
 \caption{{\color{black}MISE for $m(X)$ estimation in mixed effects models \eqref{eqn:hgnlm}. Lower is better. The MISE corresponding to the best performing method for each setup is marked in bold.}}
    \label{tab:m_mixed}  
\begin{center}
{\color{black}\begin{tabular}{||r|r|r|r|r|r|r||}
  \hline
  \hline
$\phi$ & $\sigma^2$ & BART & Basis GLMM & GAM & RF-GP & RF \\
  \hline
2 & 2 & 2.07 & \textbf{1.47} & 2.27 & 1.55 & 1.90 \\   \hline
  2 & 3 & 2.84 & 4.01 & 2.35 & \textbf{1.69} & 2.16 \\   \hline
  2 & 4 & 2.93 & 3.91 & 2.33 & \textbf{1.80} & 2.20 \\   \hline
  2 & 5 & 2.74 & 5.04 & 2.35 & \textbf{2.16} & 2.31 \\   \hline
  2 & 6 & 3.05 & 5.40 & \textbf{2.40} & 2.51 & 2.84 \\   \hline
  2 & 7 & 3.51 & 5.06 & 2.43 & \textbf{1.76} & 2.41 \\   \hline
  2 & 8 & 3.63 & 8.43 & 2.46 & \textbf{1.84 }& 2.75 \\   \hline
  2 & 9 & 3.75 & 8.80 & 2.45 & \textbf{2.04} & 2.62 \\   \hline
  3 & 2 & 1.30 & 3.43 & 2.28 & \textbf{1.27} & 1.86 \\   \hline
  3 & 3 & 1.71 & 3.43 & 2.35 & \textbf{1.55} & 1.90 \\   \hline
  3 & 4 & 1.99 & 3.49 & 2.37 & \textbf{1.88} & 1.82 \\   \hline
  3 & 5 & 2.13 & 3.86 & 2.42 & \textbf{1.65} & 2.06 \\   \hline
  3 & 6 & 2.98 & 4.62 & 2.43 & \textbf{1.83} & 2.26 \\   \hline
  3 & 7 & 2.39 & 5.43 & 2.46 & 2.00 & \textbf{1.96} \\   \hline
  3 & 8 & 2.70 & 4.44 & 2.46 & \textbf{2.20} & 2.42 \\   \hline
  3 & 9 & 2.68 & 4.70 & 2.49 & \textbf{1.72} & 2.41 \\   \hline
  4 & 2 & \textbf{1.44} & 2.98 & 2.31 & 1.50 & 1.84 \\   \hline
  4 & 3 & 1.64 & 2.91 & 2.36 & \textbf{1.43} & 1.30 \\   \hline
  4 & 4 & 1.73 & 3.56 & 2.36 & \textbf{1.71} & 1.85 \\   \hline
  4 & 5 & 1.77 & 3.63 & 2.39 & \textbf{1.62} & 1.67 \\   \hline
  4 & 6 & 1.80 & 3.64 & 2.43 & \textbf{1.46} & 2.14 \\   \hline
  4 & 7 & 1.88 & 3.96 & 2.46 & \textbf{1.44} & 2.22 \\   \hline
  4 & 8 & 2.21 & 3.82 & 2.48 & \textbf{1.37} & 2.24 \\   \hline
  4 & 9 & 2.10 & 3.99 & 2.50 & \textbf{1.74} & 2.32 \\   \hline
  5 & 2 & \textbf{1.33 }& 2.77 & 2.29 & 1.36 & 1.67 \\   \hline
  5 & 3 & 1.60 & 2.85 & 2.36 & \textbf{1.23} & 1.61 \\   \hline
  5 & 4 & \textbf{1.36} & 2.80 & 2.36 & 1.53 & 1.58 \\   \hline
  5 & 5 & 1.51 & 2.84 & 2.39 & \textbf{1.43} & 1.75 \\   \hline
  5 & 6 & 1.79 & 3.14 & 2.42 & \textbf{1.60} & 1.73 \\   \hline
  5 & 7 & 1.88 & 3.05 & 2.45 & 1.54 & \textbf{1.28} \\   \hline
  5 & 8 & 1.85 & 2.95 & 2.46 & \textbf{1.75} & 1.87 \\   \hline
  5 & 9 & 2.02 & 3.52 & 2.50 & \textbf{1.72} & 1.81 \\   \hline
 \hline
\end{tabular}}
\end{center}
\end{table}

\begin{table}[H]
\captionsetup{labelfont={color=black},font={color=black}}
 \caption{{\color{black}MISE for $\mathbb{E}(Y \given X)$ estimation in mixed effects models \eqref{eqn:hgnlm}. Lower is better. The best performance is marked in bold.}}
    \label{tab:p_mixed}  
\begin{center}
{\color{black}\begin{tabular}{||r|r|r|r|r|r||}
  \hline
$\phi$ & $\sigma^2$ & BART & Basis GLMM & RF-GP & RF \\
  \hline \hline
2 & 2 & 0.027 & 0.027 & \textbf{0.021} & 0.029 \\ \hline
  2 & 3 & 0.032 & 0.058 &\textbf{ 0.031} & 0.031 \\ \hline
  2 & 4 & 0.035 & 0.045 & \textbf{0.024} & 0.035 \\ \hline
  2 & 5 & 0.034 & 0.044 &\textbf{ 0.028} & 0.032 \\ \hline
  2 & 6 & 0.039 & 0.039 & \textbf{0.028} & 0.036 \\ \hline
  2 & 7 & 0.040 & 0.034 & \textbf{0.025} & 0.036 \\ \hline
  2 & 8 & 0.040 & 0.035 & \textbf{0.025} & 0.038 \\ \hline
  2 & 9 & 0.046 & 0.032 & \textbf{0.020} & 0.040 \\ \hline
  3 & 2 & 0.022 & 0.054 & \textbf{0.022} & 0.024 \\ \hline
  3 & 3 & 0.024 & 0.049 & \textbf{0.020} & 0.024 \\ \hline
  3 & 4 & 0.029 & 0.043 & \textbf{0.022} & 0.028 \\ \hline
  3 & 5 & 0.028 & 0.036 &\textbf{ 0.019} & 0.028 \\ \hline
  3 & 6 & 0.031 & 0.034 & \textbf{0.026} & 0.028 \\ \hline
  3 & 7 & 0.029 & 0.032 & \textbf{0.026} & 0.030 \\ \hline
  3 & 8 & 0.032 & 0.029 & \textbf{0.029} & 0.032 \\ \hline
  3 & 9 & 0.034 & 0.028 & \textbf{0.020} & 0.033 \\ \hline
  4 & 2 & \textbf{0.017} & 0.049 & 0.020 & 0.020 \\ \hline
  4 & 3 & \textbf{0.020} & 0.044 & 0.024 & 0.021 \\ \hline
  4 & 4 & 0.024 & 0.037 & 0.025 & \textbf{0.023} \\ \hline
  4 & 5 & \textbf{0.024} & 0.036 & 0.025 & 0.025 \\ \hline
  4 & 6 & 0.024 & 0.033 & \textbf{0.019} & 0.025 \\ \hline
  4 & 7 & 0.026 & 0.032 & \textbf{0.020} & 0.027 \\ \hline
  4 & 8 & 0.025 & 0.031 & \textbf{0.018} & 0.026 \\ \hline
  4 & 9 & 0.027 & 0.027 & \textbf{0.021} & 0.026 \\ \hline
  5 & 2 & \textbf{0.016} & 0.052 & 0.019 & 0.017 \\ \hline
  5 & 3 & \textbf{0.017} & 0.044 & 0.019 & 0.019 \\ \hline
  5 & 4 & \textbf{0.019} & 0.039 & 0.023 & 0.021 \\ \hline
  5 & 5 & \textbf{0.020} & 0.034 & 0.023 & 0.021 \\ \hline
  5 & 6 & 0.020 & 0.033 & \textbf{0.019} & 0.021 \\ \hline
  5 & 7 & 0.021 & 0.030 & \textbf{0.018} & 0.022 \\ \hline
  5 & 8 & 0.023 & 0.029 & \textbf{0.019} & 0.023 \\ \hline
  5 & 9 & 0.024 & 0.026 & \textbf{0.017} & 0.025 \\ \hline
   \hline
\end{tabular}}
\end{center}
\end{table}

\begin{table}[H]
\captionsetup{labelfont={color=black},font={color=black}}
 \caption{{\color{black} MSE for $\mathbb{E}(Y_{new} \given X_{new}, s_{new})$ prediction in mixed effects models \eqref{eqn:hgnlm}. Lower is better. The best performance is marked in bold. Values are in $\times 10$ scale.}}
    \label{tab:prediction_mixed}  
\begin{center}
{\color{black}\begin{tabular}{||r|r|r|r|r|r|r|r|r|r|r|r|r||}
  \hline
$\phi$ & $\sigma^2$ & BART & \begin{tabular}{@{}c@{}}Basis \\ GLMM\end{tabular} & GAM & \begin{tabular}{@{}c@{}}Deep \\ Kriging\end{tabular} & \begin{tabular}{@{}c@{}}Deep \\ Learning\end{tabular} & \begin{tabular}{@{}c@{}}BART \\ Loc\end{tabular} & \begin{tabular}{@{}c@{}}BART \\ Deep \\ Kriging\end{tabular} & \begin{tabular}{@{}c@{}}RF \\ GP\end{tabular} & RF & \begin{tabular}{@{}c@{}}RF \\ Loc\end{tabular} & \begin{tabular}{@{}c@{}}RF \\ Sp\end{tabular} \\
  \hline\hline
2 & 2 & 0.55 & 0.44 & 0.79 & 1.64 & 0.98 & \textbf{0.38} & 0.50 & 0.42 & 0.55 & 0.44 & 0.58 \\ \hline
  2 & 3 & 0.76 & 2.25 & 1.01 & 1.61 & 1.06 & 0.47 & 0.56 & \textbf{0.44} & 0.75 & 0.49 & 0.59 \\ \hline
  2 & 4 & 0.85 & 2.63 & 1.06 & 1.64 & 1.14 & 0.56 & 0.60 & \textbf{0.48} & 0.87 & 0.58 & 0.59 \\ \hline
  2 & 5 & 0.91 & 2.88 & 1.17 & 1.55 & 1.23 & 0.56 & 0.66 & \textbf{0.55} & 1.00 & 0.64 & 0.61 \\ \hline
  2 & 6 & 1.02 & 3.03 & 1.25 & 1.42 & 1.25 & 0.59 & 0.68 & \textbf{0.53} & 1.06 & 0.66 & 0.61 \\ \hline
  2 & 7 & 1.11 & 3.16 & 1.27 & 1.50 & 1.27 & 0.61 & 0.69 & \textbf{0.52} & 1.14 & 0.67 & 0.61 \\ \hline
  2 & 8 & 1.18 & 3.30 & 1.34 & 1.46 & 1.29 & 0.62 & 0.73 & \textbf{0.58} & 1.13 & 0.68 & 0.62 \\ \hline
  2 & 9 & 1.25 & 3.41 & 1.38 & 1.34 & 1.30 & 0.62 & 0.75 & \textbf{0.55} & 1.20 & 0.70 & 0.63 \\ \hline
  3 & 2 & 0.63 & 1.90 & 0.93 & 1.70 & 0.98 & 0.49 & 0.55 & \textbf{0.48} & 0.63 & 0.54 & 0.60 \\ \hline
  3 & 3 & 0.78 & 2.22 & 1.14 & 1.87 & 1.10 & 0.56 & 0.64 & \textbf{0.55} & 0.82 & 0.63 & 0.63 \\ \hline
  3 & 4 & 0.88 & 2.47 & 1.18 & 1.73 & 1.17 & 0.65 & 0.68 & \textbf{0.56} & 0.96 & 0.69 & 0.65 \\ \hline
  3 & 5 & 1.03 & 2.63 & 1.27 & 1.64 & 1.22 & 0.73 & 0.76 & \textbf{0.57} & 1.05 & 0.76 & 0.69 \\ \hline
  3 & 6 & 1.17 & 2.85 & 1.34 & 1.69 & 1.33 & 0.74 & 0.75 & \textbf{0.61} & 1.15 & 0.81 & 0.70 \\ \hline
  3 & 7 & 1.22 & 3.09 & 1.42 & 1.57 & 1.40 & 0.77 & 0.81 & \textbf{0.67} & 1.20 & 0.84 & 0.72 \\ \hline
  3 & 8 & 1.35 & 3.27 & 1.47 & 1.61 & 1.42 & 0.80 & 0.80 & \textbf{0.66} & 1.33 & 0.86 & 0.72 \\ \hline
  3 & 9 & 1.40 & 3.33 & 1.50 & 1.48 & 1.46 & 0.84 & 0.81 & \textbf{0.70} & 1.40 & 0.94 & 0.72 \\ \hline
  4 & 2 & 0.68 & 1.84 & 0.98 & 1.62 & 1.01 & 0.51 & 0.58 & \textbf{0.49} & 0.71 & 0.57 & 0.62 \\ \hline
  4 & 3 & 0.84 & 2.23 & 1.10 & 1.85 & 1.13 & 0.64 & 0.69 & \textbf{0.59} & 0.90 & 0.68 & 0.69 \\ \hline
  4 & 4 & 0.97 & 2.51 & 1.18 & 1.75 & 1.18 & 0.72 & 0.77 & \textbf{0.67} & 1.02 & 0.75 & 0.71 \\ \hline
  4 & 5 & 1.10 & 2.70 & 1.31 & 1.71 & 1.28 & 0.80 & 0.82 & \textbf{0.65} & 1.10 & 0.85 & 0.77 \\ \hline
  4 & 6 & 1.18 & 2.93 & 1.36 & 1.70 & 1.34 & 0.82 & 0.85 & \textbf{0.66} & 1.23 & 0.86 & 0.74 \\ \hline
  4 & 7 & 1.26 & 3.00 & 1.46 & 1.77 & 1.44 & 0.87 & 0.85 & \textbf{0.71} & 1.32 & 0.93 & 0.78 \\ \hline
  4 & 8 & 1.32 & 3.08 & 1.51 & 1.59 & 1.47 & 0.94 & 0.88 & \textbf{0.74} & 1.38 & 0.94 & 0.76 \\ \hline
  4 & 9 & 1.37 & 2.96 & 1.60 & 1.62 & 1.52 & 1.01 & 0.90 & \textbf{0.71} & 1.46 & 0.97 & 0.74 \\ \hline
  5 & 2 & 0.69 & 1.84 & 1.02 & 1.82 & 1.00 & \textbf{0.58} & 0.65 & 0.63 & 0.77 & 0.61 & 0.65 \\ \hline
  5 & 3 & 0.88 & 2.16 & 1.16 & 1.80 & 1.14 & 0.70 & 0.74 & \textbf{0.64} & 0.95 & 0.71 & 0.72 \\ \hline
  5 & 4 & 1.01 & 2.52 & 1.19 & 1.88 & 1.21 & 0.75 & 0.81 & \textbf{0.68} & 1.06 & 0.78 & 0.75 \\ \hline
  5 & 5 & 1.13 & 2.60 & 1.32 & 1.91 & 1.29 & 0.89 & 0.84 & \textbf{0.77} & 1.21 & 0.90 & 0.85 \\ \hline
  5 & 6 & 1.23 & 2.85 & 1.42 & 1.82 & 1.37 & 0.96 & 0.90 & \textbf{0.72} & 1.32 & 0.95 & 0.81 \\ \hline
  5 & 7 & 1.29 & 2.87 & 1.51 & 1.88 & 1.46 & 1.00 & 0.92 & \textbf{0.77} & 1.39 & 1.02 & 0.85 \\ \hline
  5 & 8 & 1.43 & 3.03 & 1.58 & 1.74 & 1.50 & 1.10 & 0.94 & 0.81 & 1.50 & 1.09 & \textbf{0.81} \\ \hline
  5 & 9 & 1.48 & 2.98 & 1.67 & 1.75 & 1.52 & 1.16 & 0.95 & 0.86 & 1.51 & 1.11 & \textbf{0.84} \\
   \hline
\end{tabular}}
\end{center}
\end{table}

\begin{table}[H]
\captionsetup{labelfont={color=black},font={color=black}}
 \caption{{\color{black}Misclassification error for prediction in mixed effects models \eqref{eqn:hgnlm}. Lower is better. The best performance is marked in bold.}}
    \label{tab:misclass_mixed}  
\begin{center}
{\color{black}\begin{tabular}{||r|r|r|r|r|r|r|r|r|r|r|r|r||}
  \hline
$\phi$ & $\sigma^2$ & BART & \begin{tabular}{@{}c@{}}Basis \\ GLMM\end{tabular}  & GAM & \begin{tabular}{@{}c@{}}Deep \\ Kriging\end{tabular} & \begin{tabular}{@{}c@{}}Deep \\ Learning\end{tabular} & \begin{tabular}{@{}c@{}}BART \\ Loc\end{tabular} & \begin{tabular}{@{}c@{}}BART \\ Deep \\ Kriging\end{tabular} & \begin{tabular}{@{}c@{}}RF \\ GP\end{tabular} & RF & \begin{tabular}{@{}c@{}}RF \\ Loc\end{tabular} & \begin{tabular}{@{}c@{}}RF \\ Sp\end{tabular} \\ \hline\hline
2 & 2 & 0.23 & 0.18 & 0.27 & 0.28 & 0.26 & \textbf{0.19} & 0.21 & 0.19 & 0.23 & 0.19 & 0.22 \\ \hline
  2 & 3 & 0.24 & 0.38 & 0.28 & 0.27 & 0.28 & 0.20 & 0.21 & \textbf{0.20} & 0.24 & 0.20 & 0.21 \\ \hline
  2 & 4 & 0.25 & 0.41 & 0.30 & 0.26 & 0.29 & 0.21 & 0.22 & \textbf{0.20} & 0.26 & 0.21 & 0.21 \\ \hline
  2 & 5 & 0.26 & 0.43 & 0.30 & 0.26 & 0.30 & 0.21 & 0.22 & \textbf{0.18} & 0.28 & 0.22 & 0.21 \\ \hline
  2 & 6 & 0.28 & 0.44 & 0.31 & 0.25 & 0.28 & 0.21 & 0.20 & \textbf{0.20} & 0.28 & 0.22 & 0.20 \\ \hline
  2 & 7 & 0.28 & 0.44 & 0.31 & 0.24 & 0.29 & 0.20 & 0.19 & \textbf{0.18} & 0.29 & 0.21 & 0.18 \\ \hline
  2 & 8 & 0.28 & 0.44 & 0.32 & 0.25 & 0.30 & 0.20 & 0.20 & \textbf{0.18} & 0.28 & 0.21 & 0.20 \\ \hline
  2 & 9 & 0.28 & 0.45 & 0.32 & 0.23 & 0.29 & 0.20 & 0.19 & \textbf{0.16} & 0.28 & 0.20 & 0.18 \\ \hline
  3 & 2 & 0.24 & 0.35 & 0.28 & 0.29 & 0.27 & 0.21 & 0.22 & \textbf{0.20} & 0.24 & 0.21 & 0.23 \\ \hline
  3 & 3 & 0.26 & 0.39 & 0.29 & 0.28 & 0.29 & 0.20 & 0.23 & \textbf{0.20} & 0.25 & 0.21 & 0.22 \\ \hline
  3 & 4 & 0.26 & 0.39 & 0.30 & 0.28 & 0.30 & 0.22 & 0.23 & \textbf{0.21} & 0.26 & 0.23 & 0.21 \\ \hline
  3 & 5 & 0.27 & 0.41 & 0.31 & 0.26 & 0.31 & 0.22 & 0.22 & \textbf{0.21} & 0.27 & 0.22 & 0.21 \\ \hline
  3 & 6 & 0.28 & 0.43 & 0.31 & 0.27 & 0.32 & 0.21 & 0.21 & \textbf{0.20} & 0.29 & 0.22 & 0.21 \\ \hline
  3 & 7 & 0.29 & 0.43 & 0.33 & 0.24 & 0.30 & 0.21 & 0.20 & \textbf{0.20} & 0.29 & 0.22 & 0.20 \\ \hline
  3 & 8 & 0.28 & 0.43 & 0.34 & 0.24 & 0.30 & 0.20 & 0.20 & \textbf{0.19} & 0.28 & 0.21 & 0.20 \\ \hline
  3 & 9 & 0.31 & 0.45 & 0.35 & 0.23 & 0.31 & 0.21 & 0.20 & 0.19 & 0.31 & 0.21 & \textbf{0.19} \\ \hline
  4 & 2 & 0.25 & 0.35 & 0.30 & 0.29 & 0.27 & 0.21 & 0.22 & \textbf{0.21} & 0.24 & 0.22 & 0.23 \\ \hline
  4 & 3 & 0.26 & 0.37 & 0.30 & 0.29 & 0.28 & 0.22 & 0.23 & \textbf{0.21} & 0.27 & 0.22 & 0.23 \\ \hline
  4 & 4 & 0.27 & 0.38 & 0.31 & 0.28 & 0.29 & 0.23 & 0.23 & \textbf{0.21} & 0.28 & 0.23 & 0.22 \\ \hline
  4 & 5 & 0.29 & 0.40 & 0.34 & 0.27 & 0.30 & 0.24 & 0.23 & \textbf{0.22} & 0.28 & 0.24 & 0.22 \\ \hline
  4 & 6 & 0.29 & 0.42 & 0.33 & 0.27 & 0.32 & 0.23 & 0.23 & 0.20 & 0.30 & 0.24 & \textbf{0.20} \\ \hline
  4 & 7 & 0.30 & 0.43 & 0.34 & 0.26 & 0.32 & 0.24 & 0.22 & \textbf{0.20} & 0.31 & 0.24 & 0.21 \\ \hline
  4 & 8 & 0.32 & 0.44 & 0.36 & 0.26 & 0.34 & 0.24 & 0.22 & 0.21 & 0.32 & 0.23 & \textbf{0.20} \\ \hline
  4 & 9 & 0.32 & 0.44 & 0.35 & 0.26 & 0.35 & 0.24 & 0.22 & 0.21 & 0.32 & 0.24 & \textbf{0.21} \\ \hline
  5 & 2 & 0.24 & 0.34 & 0.31 & 0.29 & 0.26 & 0.23 & 0.24 & \textbf{0.21} & 0.24 & 0.22 & 0.23 \\ \hline
  5 & 3 & 0.27 & 0.38 & 0.33 & 0.29 & 0.29 & 0.24 & 0.24 & \textbf{0.23} & 0.27 & 0.24 & 0.24 \\ \hline
  5 & 4 & 0.28 & 0.38 & 0.32 & 0.28 & 0.30 & 0.25 & 0.23 & \textbf{0.23} & 0.28 & 0.25 & 0.24 \\ \hline
  5 & 5 & 0.30 & 0.40 & 0.33 & 0.28 & 0.30 & 0.25 & 0.24 & 0.22 & 0.30 & 0.26 & \textbf{0.22} \\ \hline
  5 & 6 & 0.32 & 0.41 & 0.34 & 0.28 & 0.31 & 0.25 & 0.23 & 0.22 & 0.31 & 0.26 & \textbf{0.21} \\ \hline
  5 & 7 & 0.32 & 0.42 & 0.35 & 0.27 & 0.33 & 0.25 & 0.24 & \textbf{0.20} & 0.33 & 0.26 & 0.22 \\ \hline
  5 & 8 & 0.33 & 0.42 & 0.37 & 0.27 & 0.35 & 0.26 & 0.24 & \textbf{0.22} & 0.33 & 0.26 & 0.22 \\ \hline
  5 & 9 & 0.32 & 0.43 & 0.38 & 0.27 & 0.36 & 0.26 & 0.23 & 0.23 & 0.34 & 0.26 & \textbf{0.22} \\ \hline
 \hline
\end{tabular}}
\end{center}
\end{table}

Next, we focus on prediction. We measure how closely the competing methods predict the true conditional probabilities of the Bernoulli process generating the observed outcomes in the test data, i.e. $\mathbb{E}\left( Y | X, s \right)$. In the simulation setup, we have knowledge of these true conditional probabilities $\mathbb{E}(Y \given X,s)$ generating the response in the test data, hence we make use of mean square error (MSE) for comparing the prediction performances (Table \ref{tab:prediction_mixed}). 
In the analysis of real data where we do not have the knowledge of true probabilities, we will use the outcomes in the test data to evaluate the methods using misclassification error. Hence, we also compare the performance of the competing methods using misclassification error in the simulation setup.

The prediction results are presented in Tables \ref{tab:prediction_mixed} and \ref{tab:misclass_mixed}. First, we note that the differences in the performance of the methods are more prominent when looking at the MSE (Table \ref{tab:prediction_mixed}) than when looking at the misclassification error (Table \ref{tab:misclass_mixed}). This is expected as the MSE is informed by the exact difference between true and predicted probabilities while the misclassification error only uses the binary discretization of the predicted probabilities and thus obfuscates difference between the estimates of probabilities predicted by two methods as long as the discretization leads to same predictions. However, the trends between methods are consistent when looking at either MSE or misclassification error.

RF-GP produces best or comparable-to-the-best prediction performance to all its competitors across all scenarios. Methods which do not use spatial information, e.g., BART and RF cannot keep up with the methods that use the spatial information. Hence, there is a huge disparity between them and the rest of the methods in terms of MSE, which only increases as the spatial signal strength increases.  
The performance of Basis GLMM suffers since basis functions do not perform well in high covariate dimension. GAM also does not perform well since they cannot model the interaction term. Among methods with added spatial features, DeepKriging and Deep Learning do not do well. BART-Loc, RF-Loc, and RF-Sp, although performing better than the rest, follow a similar pattern. In low spatial signal strength, BART-Loc outperforms RF-Sp, whereas RF-Sp outperforms BART-Loc and RF-Loc for moderate spatial signal strength. This is because RF-sp adds several ($O(n)$) spatial covariates, which can drown out the effect of the true covariates. This negatively impacts its performance in scenarios with low spatial signal, where the true covariate effect dominates the spatial effect. RF-GP combines the best of both worlds and outperforms both BART-Loc, RF-Loc and RF-Sp significantly for both low and moderate spatial signal strength. Even at a very low (covariate) signal-to (spatial) noise ratio, i.e., when $\sigma^2$ is large, RF-GP produces comparable results to that of RF-Sp.  
As far as spatial decay is concerned, the difference between RF-GP and the rest of the methods is highest when $\phi$ is small, i.e., the data are very highly correlated. RF-GP can also account for weak correlation in the data and outperforms the rest of the methods even when $\phi$ is high, albeit by smaller margins. Both BART-loc and RF-Sp, for some scenarios, lead to up to 40\% worse MSE than RF-GP.}

\blue{\subsection{Mean nonstationarity}
\label{sim:mean_nonstationary}
The aforementioned simulation setup considers the scenario where the outcome is stationary, or the spatial effect in the mixed effects model setup \eqref{eqn:hgnlm} is stationary. Often, this assumption does not hold true in the real world. We empirically show that our model can deal with model misspecification under nonstationarity. }
 \blue{We consider the case where spatial effect is not generated from a Gaussian Process, but is a fixed smooth spatial surface. This is similar to the situation when a spatially variable covariate is not observed, i.e., the true mean is non-stationary, resulting in a violation of stationarity assumption (see discussion in Section \ref{sec:ns}). Following \cite{Saha2022}, we consider the scenario in the mixed effects model setup \eqref{eqn:hgnlm} where the spatial effect is generated, not as a GP, but as a fixed function of space. We specify this function on the $2$-d spatial domain as the density of a bivariate mixture normal distribution ($\mu_1 = (0.25, 0.5); \:\mu_2 = (0.75, 0.5), \mathbf{\Sigma}_1^2 = 0.01\mathbf I;\:  \mathbf{\Sigma}_2^2= 0.0025\mathbf I $) with two components (with means representing the locations of the two-peaks in the surface) and with covariance matrices representing the slopes and dispersion around the peaks.We control the spatial variance by multiplying the generated spatial surface by $\sigma$.
 
 We primarily focus on prediction here. Since in the experiments above, RF-GP, RF-Sp and BART-Loc were the top performing methods for prediction, we only compare the performance of these three. Table \ref{tab:MSE_mixed_nonstationary} 
 shows the MSE for prediction under mean nonstationarity. The results are similar to that of the stationary case. RF-GP performs comparably to BART-Loc in lower spatial signal setup (given the spatial effect was generated as a fixed function of location, it is expected that BART-loc, which uses the locations as added features will perform well). As the strength of the spatial signal increases, RF-GP outperforms 
BART-Loc (11\%). On the other hand, for high values of $\sigma$, RF-GP performs comparably to RF-Sp. Given the low covariate signal to spatial signal ratio, it is expected that RF-Sp will perform well in this scenario \citep{saha2023random}). However, for low spatial signal strengths, RF-GP outperforms RF-Sp (17\%).

\begin{table}[H]
\captionsetup{labelfont={color=black},font={color=black}}
 \caption{{\color{black} MSE for prediction in mixed effects models \eqref{eqn:hgnlm} under nonstationarity. Lower is better. Values are in $\times 10$ scale.}}
    \label{tab:MSE_mixed_nonstationary}  
\begin{center}
{\color{black}\begin{tabular}{||r|r|r|r||}
  \hline
$\sigma^2$ & BART-Loc & RF-GP & RF-Sp \\
  \hline
2 & 0.539 & {0.541} & 0.633 \\
  3 & {0.546} & 0.559 & 0.601 \\
  4 & 0.613 & {0.584} & 0.565 \\
  5 & 0.583 & {0.526} & 0.524 \\
 \hline
\end{tabular}}
\end{center}
\end{table}

Mean non-stationarity is a case where added-spatial-features methods like BART-Loc and RF-Sp are expected to do well, as they are exactly correctly specified. It is reassuring to see the competitive (and sometimes superior) performance of RF-GP even in this setup which is misspecifed for RF-GP. This shows the value of the parsimony encoded in RF-GP by modeling the spatial effect as a GP. The brute-force approach of just throwing in spatial features in an unstructured way into a black box algorithm will perform well if there is enough training data to learn about the structure but will generally lose out to more parsimonious models when trained on small to medium-sized datasets. Also, we reemphasize that these added features methods are only suitable for spatial predictions and not for estimation of the mean function.}

\blue{\subsection{Anisotropic working covariance}
\label{sim:anisotropy}
The previous simulation framework assumed isotropy. Often in real life examples of geographical data, the direction of the data become important, where the distance with respect to one coordinate has more effect than the distance with respect to the other coordinate. In these scenarios, anisotropy is a more realistic assumption. We consider the setup in simulation from generalized mixed effects model in \eqref{eqn:hgnlm} in Section \ref{sec:simulation_mixed}, but simulate under scale anisotropy. The correlation $cor(w_i, w_j)$ is modeled as $ \exp \left( -\phi d_{\alpha_0} \left(s_i, s_j \right)\right)$, where $d_{\alpha_0} \left(s_i, s_j \right) = \left( (\alpha_0 * s_{i,1} - \alpha_0 * s_{j,1})^2 + (s_{i,2} - s_{j,2})^2\right)^{\frac{1}{2}}$. In practice, this is implemented using a scale change of the first coordinate of the location by a factor of $\alpha_0$. 

We compare the performance of 5 approaches --- 1) RF-GP using $\alpha=1$ (isotropic covariance), 2) RF-GP using $\alpha=\alpha_0=3$ (anisotropic covariance using the true scale of anisotropy), 3) another anisotropic RF-GP where the scale of anisotrpy $\alpha$ is not fixed but is selected by cross-validation, 4) RF-SP, and 5) BART-Loc. For the RF-GP where $\alpha$ is estimated,  
we cross-validate over $\alpha \in \{1, 2,,3, 4, 5 \}$ to estimate  
the optimal $\alpha$, namely $\hat\alpha$.  We see from Table \ref{tab:anisotropy} that whether we account for the true anisotropy ($\alpha = 3$) or not ($\alpha = 1$), RF-GP has better performance than RF-Sp. It has comparable performance to BART-Loc when we do not account for anisotropy, but outperforms it when we account for it. BART-Loc is a tree-based method with location as a covariate, hence its performance remains unchanged in both the cases and is expected to perform well with respect to anisotropy. Both RF-Sp and RF-GP's performance is affected by anisotropy. As expected, their misclassification error rate is higher when we do not account for anisotropy.  
When the scale of anisotropy is estimated in RF-GP, the average of $\hat\alpha$ over 100 simulations is \textbf{3.014}. This indicates that RF-GP can correctly identify the anisotropy scale. 

\begin{table}[H]
\captionsetup{labelfont={color=black},font={color=black}}
 \caption{{\color{black} Misclassification error for prediction in mixed effects models \eqref{eqn:hgnlm} under anisotropy. Lower is better. The best performance is marked in bold.}}
    \label{tab:anisotropy}
\begin{center}
{\color{black}\begin{tabular}{||r|r|r|r|r|}
  \hline
 BART-Loc & RF-GP ($\alpha =1$) & RF-GP ($\alpha =3$) & RF-GP (fitted $\alpha$) &RF-Sp \\
  \hline
0.280 & 0.280 & \textbf{0.261} & \textbf{0.261} & 0.318 \\ \hline
\end{tabular}}
\end{center}
\end{table}}

\blue{\subsection{Partial dependence functions estimation}
\label{sim:PDP}
In a simulation setup, we show that we can use the mean function estimate from RF-GP to successfully estimate the partial dependence functions corresponding to each variables.

We simulate from a binary auto-regressive $AR(1)$ process, with the following mean function:

\begin{equation}
    p(X) = f(X_1, X_2) = \Phi\left( \cos \left( \pi * X_1\right)\right);
\end{equation}

Next, we consider the partial dependence plots corresponding to the variables $X_1$ and $X_2$. Since our mean function $p(X)$ is only a function of $X_1$, ideally we would expect the PDF of $X_1$ to coincide with $g_1(X) := \Phi \left( \cos \left( \pi X\right)\right)$ and the PDP of $X_2$ to coincide with $g_2 (X) := \mathbb E_{X_1}\Phi\left( \cos \left( \pi * X_1\right)\right) = 0.5$. We plotted the PDFs in Figure \ref{Fig:PDP} computed using the mean function estimates from our proposed approach and see that RF-GP correctly estimated the true partial dependence functions for both the true and the redundant covariate.

\begin{figure}[H]
    \centering
    \begin{subfigure}[b]{0.5\textwidth}
        \centering
        \includegraphics[height=2.6in]{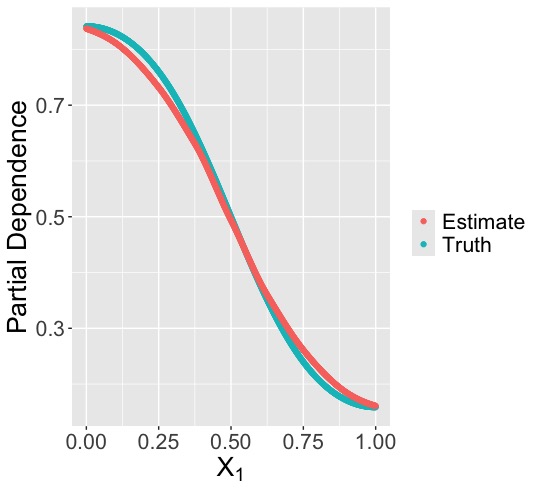}
        \caption{\blue{PDP for $X_1$}}
    \end{subfigure}%
    ~
    \begin{subfigure}[b]{0.5\textwidth}
        \centering
        \includegraphics[height=2.6in]{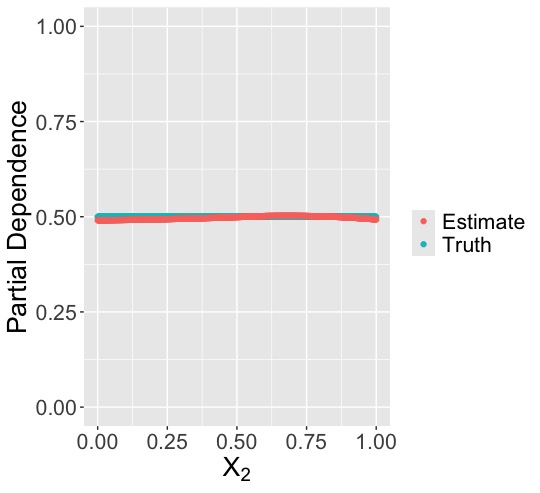}
        \caption{\blue{PDP for $X_2$}}
    \end{subfigure}
   \caption{ \blue{Partial dependence plot}}
    \label{Fig:PDP}
\end{figure}

\subsection{Conditional average treatment effect estimation:}
\label{sim:CATE}
We show that our proposed method can estimate the conditional average treatment (CATE) \eqref{eq:cate} between two groups with different treatments, while accounting for correlation. We simulate data from a generalized mixed effects model in \eqref{eqn:hgnlm}, where data points are randomly assigned to any of the one treatment (i.e., one of the two mean functions), which are given as follows: 
    \begin{equation}
        f_1 (X) = \sqrt{2} \Phi^{-1}((X - 0.5) ^2 ); \:\:\:\: f_2 (X) = \sqrt{2} * (\Phi^{-1}((X - 0.5) ^2 ) + \sin(\pi X))
    \end{equation}

\begin{figure}[H]
\centering
\includegraphics[scale=0.4]{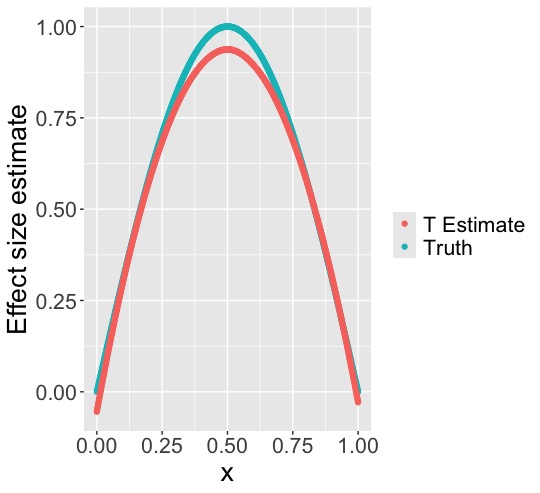}
\caption{\blue{CATE estimation}}
\label{FIG:CATE}
\end{figure}

Next, we obtain the T-learner estimate using our proposed approach, i.e., we estimate the conditional mean $\mathbb E \left( Y | X \right)$ within each group individually and take the difference of these two estimated function to estimate the CATE. Figure \ref{FIG:CATE} shows that our method provides a good T-learner estimate of the true CATE.}
\end{document}